\documentclass[a4paper,UKenglish,cleveref, autoref]{lipics-v2019}

\usepackage{tikz}
\usetikzlibrary{calc}
\usetikzlibrary{decorations.pathreplacing}
\usetikzlibrary{decorations.pathmorphing}
\usetikzlibrary{decorations.text}
\usetikzlibrary{shapes.geometric}
\usetikzlibrary{shapes}
\usepackage{marvosym}

\title{Finite Satisfiability of Unary Negation Fragment with Transitivity}

\author{Daniel Danielski}{University of Wroc\l{}aw, Poland}{}{}{}%mandatory, please use full name; only 1 author per \author macro; first two parameters are mandatory, other parameters can be empty.

\author{Emanuel Kiero\'nski}{University of Wroc\l{}aw, Poland}{kiero@cs.uni.wroc.pl}{https://orcid.org/0000-0002-8538-8221}{}

\authorrunning{D. Danielski and E. Kiero\'nski}%mandatory. First: Use abbreviated first/middle names. Second (only in severe cases): Use first author plus 'et al.'

\Copyright{D. Danielski and E. Kiero\'nski}%mandatory, please use full first names. LIPIcs license is "CC-BY";  http://creativecommons.org/licenses/by/3.0/


%mandatory

\category{}%optional, e.g. invited paper

\relatedversion{}%optional, e.g. full version hosted on arXiv, HAL, or other respository/website

\supplement{}%optional, e.g. related research data, source code, ... hosted on a repository like zenodo, figshare, GitHub, ...

\funding{Supported by {Polish National Science Centre}{} grant No {2016/21/B/ST6/01444}{}.}%optional, to capture a funding statement, which applies to all authors. Please enter author specific funding statements as fifth argument of the \author macro.

\acknowledgements{}%optional

\EventEditors{Peter Rossmanith, Pinar Heggernes, and Joost-Pieter Katoen}
\EventNoEds{3}
\EventLongTitle{44th International Symposium on Mathematical Foundations of Computer Science (MFCS 2019)}
\EventShortTitle{MFCS 2019}
\EventAcronym{MFCS}
\EventYear{2019}
\EventDate{August 26--30, 2019}
\EventLocation{Aachen, Germany}
\EventLogo{}
\SeriesVolume{138}
\ArticleNo{13}
\nolinenumbers %uncomment to disable line numbering
\hideLIPIcs  %uncomment to remove references to LIPIcs series (logo, DOI, ...), e.g. when preparing a pre-final version to be uploaded to arXiv or another public repository
\begin{document}
%%UNFO+TR
%%

\newcommand{\vectorize}[1]{\vec{#1}}
\newcommand{\uni}{\text{UNI}}
\newcommand{\one}{\text{ONE}}
\newcommand{\exi}{\text{EXI}}
\newcommand{\see}{\text{SEE}}
\newcommand{\reg}{\text{REG}}
\newcommand{\res}{\text{RES}}
\newcommand{\con}{\text{CON}}
\newcommand{\typ}{\text{TYPE}}
\newcommand{\diag}{\text{DIAG}}

% Wonky letters
\newcommand{\cE}{\mathcal{E}}%
\newcommand{\cT}{\mathcal{T}}%

\newcommand{\cK}{\mathcal{K}}%
\newcommand{\cP}{\mathcal{P}}%
\newcommand{\cL}{\mathcal{L}}%
\newcommand{\cB}{\mathcal{B}}%
\newcommand{\cR}{\mathcal{R}}%

\newcommand{\ff}{\mathfrak{f}}
\newcommand{\fp}{\mathfrak{p}}
\newcommand{\fh}{\mathfrak{h}}

\newcommand{\fd}{\mathfrak{d}}
\newcommand{\fr}{\mathfrak{r}}
\newcommand{\fs}{\mathfrak{s}}
\newcommand{\fc}{\mathfrak{c}}
\newcommand{\fg}{\mathfrak{g}}
\newcommand{\fq}{\mathfrak{q}}

%Logics
\newcommand{\ML}{\mbox{ML}}
\newcommand{\MLU}{\mbox{ML}${+}\langle U \rangle$}
\newcommand{\FO}{\mbox{\bf FO}}
\newcommand{\FOt}{\mbox{$\mbox{\rm FO}^2$}}
\newcommand{\GF}{\mbox{$\mbox{\rm GF}$}}
\newcommand{\GFt}{\mbox{$\mbox{\rm GF}^2$}}
\newcommand{\GFtEG}{\mbox{$\mbox{\rm GF}^2{+}{\rm EG}$}}
\newcommand{\GFtEQ}{\mbox{$\mbox{\rm GF}^2{+}{\rm EQ}$}}
\newcommand{\FOtEQ}{\mbox{$\mbox{\rm FO}^2{+}{\rm EQ}$}}
\newcommand{\Ct}{\mbox{$\mathcal{C}^2$}}
\newcommand{\FOtEC}{\mbox{$\mbox{\rm EC}^2$}}
\newcommand{\FOtECth}{\mbox{$\mbox{\rm EC}^2_3$}}
\newcommand{\FOtECt}{\mbox{$\mbox{\rm EC}^2_2$}}
\newcommand{\FOtECo}{\mbox{$\mbox{\rm EC}^2_1$}}
\newcommand{\FOtECk}{\mbox{$\mbox{\rm EC}^2_k$}}
\newcommand{\FOtEth}{\mbox{$\mbox{\rm EQ}^2_3$}}
\newcommand{\FOtEt}{\mbox{$\mbox{\rm EQ}^2_2$}}
\newcommand{\FOtEo}{\mbox{$\mbox{\rm EQ}^2_1$}}
\newcommand{\FOtEk}{\mbox{$\mbox{\rm EQ}^2_k$}}
\newcommand{\GFtEo}{\mbox{$\mbox{\rm GFEQ}^2_1$}}
\newcommand{\GFtEk}{\mbox{$\mbox{\rm GFEQ}^2_k$}}
\newcommand{\GFtEt}{\mbox{$\mbox{\rm GFEQ}^2_2$}}
\newcommand{\UNFO}{\mbox{UNFO}}
\newcommand{\GNFO}{\mbox{GNFO}}
\newcommand{\GNFOonedim}{\mbox{GNFO}$_1$}
\newcommand{\ALC}{\mbox{$\mathcal{ALC}$}}
\newcommand{\FF}{\mbox{FF}}
\newcommand{\UNFOt}{\mbox{UNFO$^2$}}
\newcommand{\UNFOEQ}{\mbox{UNFO}{+}\mbox{EQ}}
\newcommand{\UNFOTR}{\mbox{UNFO}{+}$\mathcal{S}$}
\newcommand{\UNFOSO}{\mbox{UNFO}{+}$\mathcal{SO}$}
\newcommand{\UNFOSOH}{\mbox{UNFO}{+}$\mathcal{SOH}$}
\newcommand{\UNFObinTR}{\mbox{UNFO}$_{bin}${+}$\mathcal{S}$}
\newcommand{\UNFOtTR}{\mbox{UNFO$^2$}{+}$\mathcal{S}$}
\newcommand{\UNFOtEQ}{\mbox{UNFO$^2$}{+}\mbox{EQ}}
\newcommand{\GNFOTR}{\mbox{BGNFO}{+}$\mathcal{S}$}
\newcommand{\GNFOEQ}{\mbox{BGNFO}{+}\mbox{EQ}}
\newcommand{\GNFOTRonedim}{\mbox{BGNFO}$_1${+}$\mathcal{S}$}
\newcommand{\GNFOEQonedim}{\mbox{BGNFO}$_1${+}\mbox{EQ}}
\newcommand{\GFtTR}{\mbox{GF}$^2${+}\mbox{TR}}

% Complexity classes
\newcommand{\NLogSpace}{\textsc{NLogSpace}}
\newcommand{\NP}{\textsc{NP}}
\newcommand{\PTime}{\textsc{PTime}}
\newcommand{\PSpace}{\textsc{PSpace}}
\newcommand{\ExpTime}{\textsc{ExpTime}}
\newcommand{\ExpSpace}{\textsc{ExpSpace}}
\newcommand{\NExpTime}{\textsc{NExpTime}}
\newcommand{\TwoExpTime}{2\textsc{-ExpTime}}
\newcommand{\TwoNExpTime}{2\textsc{-NExpTime}}
\newcommand{\ThreeNExpTime}{3\textsc{-NExpTime}}
\newcommand{\APSpace}{\textsc{APSpace}}
\newcommand{\AExpSpace}{\textsc{AExpSpace}}

% Other symbols of ours
\newcommand{\str}[1]{{\mathfrak{#1}}}
\newcommand{\restr}{\!\!\restriction\!\!}
\newcommand{\N}{{\mathbb N}}   % Natural numbers
\newcommand{\Q}{{\mathbb Q}}   % rationals
\newcommand{\Z}{{\mathbb Z}}   % integers

\newcommand{\rng}{\mathrm{Rng}}
\newcommand{\dom}{\mathrm{Dom}}
\newcommand{\var}{\mathrm{Var}}

\renewcommand{\phi}{\varphi} % Nicer-looking phi

\newcommand{\EQ}{\ensuremath{{\mathcal EQ}}}
\newcommand{\Sat}{\ensuremath{\textit{Sat}}}
\newcommand{\FinSat}{\ensuremath{\textit{FinSat}}}
\newcommand{\AAA}{\mbox{\large \boldmath $\alpha$}}
\newcommand{\BBB}{\mbox{\large \boldmath $\beta$}}
\newcommand{\GGG}{\mbox{\large \boldmath $\gamma$}}
\newcommand{\type}[2]{{\rm atp}^{{#1}}({#2})}
\newcommand{\gtype}[2]{{\rm gtp}^{{#1}}({#2})}
\newcommand{\sss}{\scriptscriptstyle}
\newcommand{\cBase}{{\sss \text{base}}}%
\newcommand{\cDist}{{\sss \text{dist}}}
\newcommand{\const}{{\sss \text{const}}}%DD
\newcommand{\cAux}{{\sss \text{aux}}}
\newcommand{\cTr}{{\sss \text{tr}}}
\newcommand{\cEq}{{\sss \text{eq}}}
\newcommand{\patel}{\mathfrak{p}}
\newcommand{\patgtp}{pattergtp}
\newcommand{\dec}[3]{{\rm dec}^{{#1}}_{{#2}}({#3})}
\newcommand{\ldec}[2]{{\rm ldec}^{{#1}}({#2})}
\newcommand{\origin}{origin}
\newcommand{\rank}[3]{\fr^{#1}_{#2}(#3)}

\maketitle  

\begin{abstract}
We show that the finite satisfiability problem for the unary negation fragment with an arbitrary number of
transitive relations is decidable and \TwoExpTime-complete. Our result actually holds  for a more general
setting in which one can require that  some binary symbols are interpreted as arbitrary transitive relations, some  as 
partial orders and some as equivalences.
We also consider finite satisfiability of various extensions of our primary logic, in particular capturing the concepts of
nominals and role hierarchies known from description logic.
As the unary negation fragment can express 
unions of conjunctive queries, our results have interesting implications for
the problem of finite  query answering, both in the classical scenario and in
the description logics setting.  
\end{abstract}

\section{Introduction}

{\bf Decidable fragments and unary negation.}
Searching for attractive fragments of first-order logic is an  important theme in
theoretical computer science. 
Successful examples of such fragments, with numerous applications, are modal and description logics. 
They have their own syntax, but naturally translate to first-order logic, via the \emph{standard translation}.
Several seminal decidable fragments of first-order logic were identified
by preserving one particular  restriction obeyed by 
this translation
and dropping all the others.  
Important examples of such fragments are
 two-variable logic, \FOt{}, \cite{Sco62},  the guarded fragment, \GF{}, \cite{ABN98}, and the fluted fragment, \FF, \cite{Q69,P-HST19}.
They restrict, respectively, the
number of variables, the quantification pattern and the order of variables in which they appear as arguments of predicates.
 A more recent proposal \cite{StC13} is the \emph{unary negation fragment}, \UNFO{}. This time we restrict the use of negations, 
allowing them only in
front of subformulas with at most one free variable.
\UNFO{} turns out to retain many good algorithmic and model theoretic properties of modal logic, including the finite model
property, a tree-like model property and the decidability of the satisfiability problem. 
We remark here that $\UNFO$ and $\GF$ have a common decidable generalization, \emph{the guarded negation fragment}, \GNFO{}, \cite{BtCS15}.

To justify the  attractiveness 
 of \UNFO{} let us look at one of the crucial problems in database theory, 
open-world query answering. Given an (incomplete) set of facts $\str{D}$, a set of constraints $\mathcal{T}$ and a query $q$, check 
if $\str{D} \wedge \mathcal{T}$ entails $q$. Generally, this problem is undecidable, and to make it decidable one needs to
restrict the class of queries and constraints. Widely investigated class of queries are (unions of) 
\emph{conjunctive queries}---(disjunctions of) sentences of the form $\exists \bar{x} \psi(\bar{x})$ where  $\psi$ is a conjunction of  atoms. An important class of constraints
are \emph{tuple generating dependencies}, TGDs, of the form $\forall \bar{x} \bar{y} (\psi(\bar{x}, \bar{y}) \rightarrow \exists \bar{z} \psi'(\bar{y}, \bar{z})   )$, where $\psi$ and $\psi'$ are, again, conjunctions of atoms. Conjunctive query answering against arbitrary TGDs is still undecidable (see, e.g., \cite{CGK13}), so TGDs need to be restricted further. Several classes of TGDs making the problem decidable
have been proposed. One interesting such class are \emph{frontier-one} TGDs, in which the \emph{frontier} of each dependency, $\bar{y}$, 
consists just of a single variable \cite{BLL09}. Frontier-one TGDs are a special case of frontier-guarded TGDs \cite{BLM10}.
Checking whether $\str{D}$ and $\mathcal{T}$ entail $q$ boils down to verifying (un)satisfiability of the formula $\str{D} \wedge \mathcal{T} \wedge \neg q$. It turns out that if $\mathcal{T}$ is a conjunction of frontier-one TGDs and $q$ is a disjunction of conjunctive queries
then the resulting formula belongs to \UNFO{}. 

\smallskip\noindent
{\bf Transitivity.}
A serious weakness of the expressive power of \UNFO{} is that it cannot express transitivity of a binary relation, nor  related 
properties like  being an equivalence, a partial order or a linear order. This limitation becomes particularly important when database or knowledge representation applications are
considered, as transitivity  is a  natural  property in many real-life situations. Just consider relations
like \emph{greater-than} or \emph{part-of}.
This weakness is shared by 
\FOt{}, \GF{} and \FF{}. Thus, it is natural to think about their
extensions, in which some distinguished binary symbols may be explicitly required to be interpreted as transitive relations.
It turns out that  \FOt{}, \GF{} and \FF{} do not cope well with transitivity, and the satisfiability problems for the obtained
extensions are undecidable \cite{GOR99,Gra99,P-HT19} (see also \cite{GMV99,Kie05,Kaz06}).
Some positive results were obtained for \FOt{}, \GF{} and \FF{} only  when one transitive relation is available \cite{P-H18,Kie05,P-HT19}
or when some further syntactic restrictions 
are imposed
\cite{ST04}.

\UNFO{} is an exception here, since its satisfiability problem remains decidable in the presence of arbitrarily many 
transitive relations. This has been explicitly stated in \cite{JLM18}, as a corollary from a  stronger
result that \UNFO{} is decidable when extended by  regular path expressions. Independently, the decidability of \UNFO{} with
transitivity, \UNFOTR{}, follows from \cite{ABBB16}, which deals with the decidability of 
a richer logic, the guarded negation fragment  with transitive relations restricted to non-guard positions, which embeds  \UNFOTR{}.
From both papers the \TwoExpTime-completeness of \UNFOTR{} can be inferred.

\smallskip\noindent
{\bf Our main results.}
A problem related to satisfiability is   \emph{finite satisfiability}, in which 
we ask about the existence of \emph{finite} models.
In computer science, the importance of decision procedures for finite satisfiability arises from the fact that most objects about which we may want to reason using logic, e.g., databases, are finite. Thus the ability of solving only  general satisfiability 
may not be fully satisfactory.
Both the above-mentioned decidability results implying the decidability of \UNFOTR{} are obtained by employing  tree-like model properties of the logics and then using automata techniques. Since tree-like unravelings of models are infinite, this approach works only for general satisfiability, and gives little insight into the decidability/complexity of 
finite satisfiability.
In this paper we consider the finite satisfiability problem for \UNFOTR{}. Actually, we made a step in
this direction already in our previous paper \cite{DK18} (see \cite{DK18b} for its longer version) 
where we proved a related result that \UNFO{} with equivalence relations, \UNFOEQ{}, has the finite model property and
 thus that its satisfiability and finite satisfiability problems coincide, both being \TwoExpTime-complete. 
Some ideas developed in \cite{DK18} are extended and applied also here, even though \UNFOTR{} does not have
the finite model property which becomes evident when looking at the following formula with transitive $T$, $\forall x \exists y Txy \wedge \forall x \neg Txx$, satisfiable only in infinite models.

Our main 
contribution is demonstrating the decidability of 
 finite satisfiability for \UNFOTR{} and establishing its \TwoExpTime-completeness. En route we obtain
a triply exponential bound on the size of minimal models of finitely satisfiable \UNFOTR{} formulas. 
Actually, our results hold for a more general setting, in which some relations may be required to be interpreted as equivalences,
some as partial orders, and some  just as arbitrary transitive relations.
Returning to database motivations, we get this way the decidability of 
the \emph{finite} open-world query answering for unions of conjunctive queries against frontier-one TGDs with equivalences, partial orders and arbitrary transitive relations. By \emph{finite open-world query answering} we mean the question if for given $\str{D}$, $\mathcal{T}$ and 
$q$, $\str{D}$ and $\mathcal{T}$ entail $q$ over \emph{finite} structures.

To the best of our knowledge, \UNFOTR{} is the first logic which allows one to use arbitrarily many transitive relations,
and, at the same time, to  speak non-trivially about relations of arbitrary arities, whose finite satisfiability problem is shown decidable.
In the case of  related logics of this kind, like the guarded fragment with transitive guards \cite{ST04}, and the guarded negation fragment
with transitive relations outside guards \cite{ABBB16}, the decidability was shown only for general satisfiability, and its finite version is open.  
(Finite satisfiability was shown decidable only for the \emph{two-variable} guarded fragment with transitive guards \cite{KT18}).

We believe that moving from \UNFOEQ{} from \cite{DK18} to \UNFOTR{} is an important improvement. Besides the fact that 
 this requires strengthening our techniques and employing some new ideas, 
general transitive relations
have stronger motivations than equivalences. In particular, it opens natural connections to the realm of description logics, DLs.

\smallskip\noindent
{\bf \UNFO{} and expressive description logics.}
\UNFO{}, via the above-mentioned standard translation, embeds the DL \ALC{}, as well as its extension by inverse roles ($\mathcal{I}$)
and role intersections ($\sqcap$). Thus, having the ability of expressing conjunctive queries, we can use our results to
solve the so-called \emph{(finite) ontology mediated query answering} problem, (F)OMQA, for some DLs. This problem is 
a counterpart of (finite) open-world query answering: 
given a conjunctive query (or a union of conjunctive queries) and a knowledge base specified in a DL,
check whether the query holds in every (finite) model of this knowledge base. 

While there are quite a lot of results for OMQA, not much is known about FOMQA.
In particular, for DLs with transitive roles ($\mathcal{S}$) the only positive results we are aware of are the ones obtained recently in \cite{GYM18}, where the decidability and \TwoExpTime-completeness of FOMQA for the logics $\mathcal{SOI}$, $\mathcal{SIF}$ and $\mathcal{SOF}$ is shown. This is orthogonal to our results described above, since \UNFOTR{} captures neither nominals ($\mathcal{O}$) nor
functional roles ($\mathcal{F}$). On the other hand, we are able to express any positive boolean combinations of roles, including their intersection ($\sqcap$), which 
allows us to solve FOMQA, e.g., for the logic $\mathcal{SI}^\sqcap$. Moreover we can use non-trivially relations of arity greater than two. 

It is an interesting question if our decidability result can be extended to capture some more expressive DLs. Unfortunately, we cannot hope for \emph{number restrictions} ($\mathcal{Q}$ or $\mathcal{N}$)
or even \emph{functional roles} ($\mathcal{F}$), as satisfiability and finite satisfiability of \UNFO{}  (even without transitive relations) and two binary functional relations are undecidable. This is implicit in \cite{StC13} (see Appendix \ref{a:frund} for an explicit proof).  On the positive side, we show the decidability and \TwoExpTime-completeness of finite satisfiability of
\UNFOSOH, extending \UNFOTR{} by  constants (corresponding to \emph{nominals}
($\mathcal{O}$)) and inclusions of binary relations (capturing \emph{role hierarchies} ($\mathcal{H}$)). This is sufficient, in particular, to
imply the decidability of FOMQA for the description logic
$\mathcal{SHOI}^{\sqcap}$, which, up to our knowledge, is a new result.

\smallskip\noindent
{\bf Towards guarded negation fragment.} 
We propose also another decidable extension of our basic logic, the \emph{one-dimensional base-guarded negation fragment} with transitive relations on non-guard positions, \GNFOTRonedim{}. This is a non-trivial fragment of the already mentioned logic from \cite{ABBB16}. After some rather easy adjustments, our constructions  cover this bigger logic, however, it becomes undecidable when extended with inclusions of binary relations. 

\smallskip\noindent
{\bf Organization of the paper.}
The rest of this paper is organized as follows. Section \ref{s:prel} contains definitions, basic facts and a high-level description of
our decidability proof. As our constructions are rather complex, in the main body of the paper, Section \ref{s:2var},
we explicitly process the restricted, 
two-variable case of our logic, for which our ideas can be presented more transparently. In Section \ref{s:alltheother} we just formulate the remaining results, leaving the details for the Appendix, which also contains the missing proofs from Sections 2 and 3. In Section \ref{s:conclusions} we conclude
the paper.

\section{Preliminaries} \label{s:prel}

\subsection{Logics, structures, types and functions} \label{s:prela}

We employ standard terminology and notation from model theory.
We refer to structures using Fraktur capital letters,
and their domains using the corresponding Roman capitals.
For a structure $\str{A}$ and $A' \subseteq A$ we use 
$\str{A} \restr A'$ or $\str{A}'$ to denote the restriction of $\str{A}$ to $A'$.

The \emph{unary negation fragment of first-order logic}, \UNFO{} is defined  by the following 
grammar \cite{StC13}:
$\phi=B\bar{x} \mid x=y \mid \phi \wedge \phi \mid \phi \vee \phi \mid \exists x \phi \mid \neg \phi(x)$,
where, in the first clause, $B$ represents any relational symbol, and, in the last clause, $\phi$ has no free variables besides (at most) $x$. 
An example formula
not expressible in \UNFO{} is $x \not= y$.
We formally do not have universal quantification. However we allow
ourselves to use $\forall \bar{x} \neg \phi$ as an abbreviation for $\neg \exists \bar{x}  \phi$,
for an \UNFO{} formula $\phi$.
Note that frontier-one TGDs $\forall \bar{x} y (\psi(\bar{x}, y) \rightarrow \exists \bar{z} \psi'(y, \bar{z})   )$ are in \UNFO{}
as they can be rewritten as $\neg \exists \bar{x} y ( \psi(\bar{x}, y) \wedge \neg \exists \bar{z} \psi'(y, \bar{z})  )$.

We mostly work with purely relational signatures (admitting constants only in some extensions of our main results) of the form $\sigma=\sigma_{\cBase} \cup \sigma_{\cDist}$, where $\sigma_{\cBase}$ is the \emph{base signature}, and $\sigma_{\cDist}$ is the \emph{distinguished signature}. 
We assume that $\sigma_{\cDist} = \{{T}_1, \ldots, {T}_{2k} \}$, with all the ${T}_u$  binary, and intension
that ${T}_{2u}$ is interpreted as the inverse of ${T}_{2u-1}$. 
For every $1 \le u \le k$ we sometimes write ${T}_{2u}^{-1}$ for ${T}_{2u-1}$, and ${T}_{2u-1}^{-1}$ for ${T}_{2u}$. We say that a subset $\mathcal{E}$ of $\sigma_{\cDist}$ is \emph{closed under inverses} if, for every $1 \le u \le 2k$,
we have ${T}_u \in \mathcal{E}$ iff ${T}_u^{-1} \in \mathcal{E}$. Note that $\mathcal{E}$ is closed under inverses iff $\sigma_{\cDist} \setminus \mathcal{E}$ is closed under inverses. 
Given a formula $\phi$ we denote by $\sigma_\phi$ the signature induced by $\phi$, \emph{i.e.}, the minimal signature, with 
its distinguished part closed under inverses, containing all symbols from $\phi$.

The \emph{unary negation fragment with transitive relations}, \UNFOTR, is defined by the same grammar as \UNFO{}, however when satisfiability of its formulas is considered, we restrict the class of admissible models  to those that interpret all symbols from $\sigma_{\cDist}$ as transitive relations and, additionally,
 for each $u$, interpret ${T}_{2u}$ as the inverse of ${T}_{2u-1}$. The latter condition is intended to
simplify the presentation, and is imposed without loss of generality.
In our constructions we sometimes consider some auxiliary structures in which symbols from $\sigma_{\cDist}$ are not necessarily
interpreted as transitive relations (but the pairs ${T}_{2u-1}$, ${T}_{2u}$ are always interpreted as inverses of each other).

An \emph{(atomic) $k$-type} over a signature $\sigma$ is a maximal satisfiable set of literals (atoms and negated atoms) over $\sigma$ with variables $x_1, \ldots, x_k$. 
We often identify a $k$-type with the conjunction of its elements. We are mostly interested in $1$- and $2$-types.
Given a $\sigma$-structure $\str{A}$ and  $a, b \in A$ we denote by $\type{\str{A}}{a}$ the 
$1$-type \emph{realized} by $a$, that is the unique $1$-type $\alpha(x_1)$ such that $\str{A} \models \alpha(a)$,
and by $\type{\str{A}}{a,b}$ the unique $2$-type $\beta(x_1,x_2)$ such that $\str{A} \models \beta(a,b)$.

We use various functions in our paper. Given a function $f:A \rightarrow B$ we  denote by $\rng{f}$ its range, by $\dom{f}$ its domain,
and by $f \restr A_0$ the restriction of $f$ to  $A_0 \subseteq A$.

\subsection{Normal form, witnesses and basic facts}

We say that an \UNFOTR{} formula is in Scott-normal form if it is of the shape
\begin{eqnarray}
\forall x_1, \ldots, x_t \neg \phi_0(\bar{x}) \wedge \bigwedge_{i=1}^{m} \forall x \exists \bar{y} \phi_i(x,\bar{y})
\label{eq:nf}
\end{eqnarray}
where each  $\phi_i$ is a \UNFOTR{} quantifier-free formula and $\phi_0$ is additionally in negation normal form (NNF).
A similar normal form for \UNFO{} was introduced in the bachelor's thesis \cite{Dzi17}. 
By a straightforward adaptation of  Scott's translation for \FOt{} \cite{Sco62} one can translate in polynomial time
any \UNFOTR{} formula to a formula in normal form, in such a way that both are satisfiable over the same domains.
 This allows us, when dealing with decidability/complexity issues for \UNFOTR{}, or when considering 
the size of minimal finite models of formulas, to restrict attention to normal form formulas.

Given a structure $\str{A}$, a normal form formula $\phi$ as in (\ref{eq:nf}) and elements $a, \bar{b}$ of $A$ such that
$\str{A} \models \phi_i(a,\bar{b})$ we say that the elements of $\bar{b}$ are \emph{witnesses} for $a$ and $\phi_i$ and that
$\str{A} \restr \{a, \bar{b} \}$ is a \emph{witness structure} for $a$ and $\phi_i$.
Fix an element $a$. For every  $\phi_i$ choose a witness structure $\str{W}_i$. Then the structure $\str{W}=\str{A} \restr \{W_1 \cup \ldots \cup W_m \}$ is called a $\phi$-\emph{witness structure} for $a$.

We are going to present
a construction which given an arbitrary finite model of a
normal form \UNFOTR{} formula $\phi$ builds a finite model of $\phi$ of a bounded size.
The construction goes via several intermediate steps in which some tree-like models are produced. To argue that that they are still
 models
of $\phi$ we use the following basic observation 
(we recall that $t$ is the number of variables of the $\forall$-conjunct of $\phi$).

\begin{lemma} \label{l:homomorphisms}
Let $\str{A}$ be a model of a normal form \UNFOTR{} formula $\phi$. Let $\str{A}'$ be a structure
in which
all symbols from $\sigma_{\cDist}$ are interpreted as transitive relations, such that
\begin{enumerate}[({a}1)]
\item 
for every $a' \in A'$ there is a $\phi$-witness structure for $a'$ in $\str{A}'$,
\item for every tuple  $a'_1, \ldots, a'_t \in A'$ 
 there is a homomorphism $\fh: \str{A}' \restr \{a_1', \ldots, a_t' \} \rightarrow \str{A}$ which preserves $1$-types of elements.
\end{enumerate}
Then $\str{A}' \models \phi$.

\end{lemma}

\subsection{Plan of the small model construction} \label{s:stages}
Our main goal is to show that finite satisfiability of \UNFOTR{} formulas can be checked in \TwoExpTime{}. 
To this end we will introduce a natural notion of tree-like structures and a measure associating with  transitive paths of such structures
their so-called \emph{ranks}. Intuitively, for a transitive relation $T_i$ and a $T_i$-path $\pi$, the $T_i$-rank of $\pi$ is the number of one-directional $T_i$-edges in $\pi$ (a precise definition is given in Section \ref{s:ltreelike}).
Then we show that having
the following forms of models is equivalent for a normal form  formula $\phi$:
\begin{enumerate}[(f1)]
\item finite;
\item tree-like, with  bounded ranks of transitive paths;
\item tree-like, with ranks of transitive paths bounded doubly exponentially in $|\phi|$;
\item tree-like, with ranks of paths bounded doubly exponentially in $|\phi|$, and regular (with doubly exponentially many non-isomorphic subtrees);
\item finite of size triply exponential in $|\phi|$.
\end{enumerate}
We will make the  following steps: (f1) $\leadsto$ (f2), (f2) $\leadsto$ (f3), (f3) $\leadsto$ (f4), (f4) $\leadsto$ (f5).
The step closing the circle, (f5) $\leadsto$ (f1) is trivial. In the two-variable case, we will omit the form (f4) and
directly show (f3) $\leadsto$ (f5). Our \TwoExpTime-algorithm will look for models of the form (f3).
Showing transitions leading from (f3) to (f5) justifies that its answers coincide indeed with the existence of \emph{finite} models.

This scheme is similar to the one we used to show the finite model property for \UNFOEQ{} in \cite{DK18}. 
In the main part of the construction from  \cite{DK18} we build bigger and bigger substructures in which some equivalence relations are total.  The induction goes, roughly speaking, by the number of non-total equivalences in the substructure.
Here we extend this approach to handle one-way transitive connections. It may be useful to briefly
compare the case of \UNFOTR{} and the case of \UNFOEQ{}.  

First of all, if a given formula $\phi$ is from \UNFOEQ{} then we can start our constructions leading to a small finite model of $\phi$ from its arbitrary model, while if $\phi$ is in \UNFOTR{} we start from a \emph{finite} model of $\phi$.
A very simple step (f1) $\leadsto$ (f2) in both papers is, essentially, identical. 
The counterpart of step (f3) $\leadsto$ (f4) in the case of equivalences is slightly simpler, but the main differences
lie in steps (f2) $\leadsto$ (f3) and (f4) $\leadsto$ (f5). 
The former, clearly, is not present at all in \cite{DK18}.
While the general idea in this step is quite standard, as we just use a kind of tree pruning, the details are
rather delicate due to possible interactions among different transitive relations, and this step is, by no means, trivial.
We refine here, in particular, the apparatus of \emph{declarations} introduced in \cite{DK18}. 
 Regarding step (f4) $\leadsto$ (f5), the main construction there, in its single inductive step, has two phases: building the so-called
\emph{components} and then arranging them into a bigger structure. It is this first phase which is more complicated than
in the corresponding step in \cite{DK18}. Having components prepared we join them similarly as in \cite{DK18}.

\section{The two-variable case} \label{s:2var}

As in the case of unbounded number of variables we can restrict attention to normal form formulas, which in the two-variable case 
simplify to the standard Scott-normal form  \cite{Sco62}:
\begin{eqnarray}
\forall {xy} \neg \phi_0({x,y}) \wedge \bigwedge_{i=1}^{m} \forall x \exists {y} \phi_i(x,{y}),
\label{eq:lnf}
\end{eqnarray}
where all $\phi_i$ are quantifier-free \UNFOtTR{}
formulas (in this restricted case it is not important whether $\phi_0$ is in NNF or not). As is typical for two-variable logics we  
assume that formulas do not use relational symbols of arity greater than $2$ (cf.~\cite{GKV97}).

\subsection{Tree pruning in the two-variable case} \label{s:ltreelike}

We use a standard notion of a (finite or infinite) rooted tree and related terminology. 
Additionally, any set consisting of a node and all its children is called a \emph{family}.
Any node $b$, except for the root and the leaves, belongs to two families: the one containing its parent, and the one containing its children,
the latter called  the \emph{downward family of} $b$.

We say that a structure $\str{A}$ over a signature consisting of unary and binary symbols is a \emph{light tree-like structure} 
if its nodes can be arranged into a rooted tree  in such a way that
if $\str{A} \models Baa'$ for some non-transitive relation symbol $B$ 
then one of three conditions holds: $a=a'$, $a$ is the parent of $a'$ or $a$ is a child of $a'$, and
if $\str{A} \models {T}_uaa'$ for some ${T}_u$ then  either $a=a'$ or there is a sequence of distinct nodes 
$a=a_0, a_1, \ldots, a_k = a'$ such that $a_i$ and $a_{i+1}$ are joined by an edge of the tree and $\str{A} \models {T}_ua_i a_{i+1}$.
In other words, distant nodes in a light tree-like structure can be joined only by transitive connections, moreover, these
transitive connections are just the transitive closures of connections inside families.
For a light tree-like structure $\str{A}$
and $a \in A$ we denote by $A_a$ the set of all nodes
in the subtree rooted at $a$ and by $\str{A}_a$ the corresponding substructure.

Let $\str{A}$ be a light tree-like structure. A sequence of nodes $a_1, \ldots, a_N \in A$ is a \emph{downward path} 
in $\str{A}$ if for each $i$ 
$a_{i+1}$ is a child of $a_i$.
A \emph{downward}-${T}_u$-\emph{path}
is a downward path such that for each $i$ we have $\str{A}\models{T}_ua_ia_{i+1}$. 
The ${T}_u$-\emph{rank} of a downward-${T}_u$-path $\vectorize{a}$,
$\rank{\str{A}}{u}{\vectorize{a}}$,
is the cardinality of the set $\{i:\str{A}\models\neg{T}_ua_{i+1}a_{i}\}$. The ${T}_u$-\emph{rank} of an element $a\in A$ is defined as $\rank{\str{A}}{u}{a}=\sup\{\rank{\str{A}}{u}{\vectorize{a}}:\vectorize{a}=a,a_2,\ldots,a_N;\vectorize{a} \text{ is a downward-}{T}_u\text{-path}\}$.
For an integer $M$, we say that $\str{A}$ has downward-${T}_u$-paths \emph{bounded by} $M$  when for all $a\in A$ we have $\rank{\str{A}}{u}{{a}}\leq M$, 
and that $\str{A}$ has \emph{transitive paths bounded by}
$M$ if it has downward-${T}_u$-paths 
bounded by $M$ for all $u$. Note that a downward-${T}_u$-path bounded by $M$ may have more than $M$ nodes, 
 as the symmetric ${T}_u$-connections do not increase  the rank.

Given an arbitrary model $\str{A}$ of a normal form \UNFOtTR{} formula $\phi$ we can simply construct its light tree-like model
of degree bounded by $|\phi|$. We define a \emph{light}-$\phi$-\emph{tree-like unraveling} $\str{A}'$ of $\str{A}$ and an associated function $\fh:A' \rightarrow A$  in the following way. 
$\str{A}'$ is divided into levels $L_0, L_1, \ldots$. Choose an
arbitrary element $a \in A$ and add to level $L_0$ of $A'$  an element $a'$ such that $\type{\str{A}'}{a'}=\type{\str{A}}{a}$; set $\fh(a')=a$. 
The element $a'$ will
be the only element of $L_0$ and will become the root of $\str{A}'$.
Having defined $L_i$ repeat the following for every $a' \in L_i$. 
For every $j$, if $\fh(a')$ is not a witness for $\phi_j$ and itself then choose in $\str{A}$ a witness $b$ for $\fh(a')$ and $\phi_j$. 
Add a fresh copy $b'$ of $b$ to $L_{i+1}$, make 
$\str{A}' \restr \{a', b' \}$  isomorphic to $\str{A} \restr \{\fh(a'), b\}$ and set $\fh(b')=b$. Complete the definition
of $\str{A}'$ transitively closing all relations from $\sigma_\cDist$.

\begin{lemma}[(f1) $\leadsto$ (f2), light] \label{l:ltreelike}
Let $\str{A}$ be a finite model of a normal form \UNFOtTR{} formula $\phi$. Let $\str{A}'$ be a light-$\phi$-tree-like unraveling of $\str{A}$.
Then $\str{A}' \models \phi$ and $\str{A}'$ is a light tree-like structure of degree bounded by $|\phi|$, and transitive paths bounded by
$|A|$. 
\end{lemma}

Our next task is making the transition  (f2) $\leadsto$ (f3).
For this purpose we introduce a notion of light declarations. It is closely related to a notion of declarations which will
be used in the general case, but simpler than the latter. Fix a signature and let $\AAA$ be the set of $1$-types over this
signature.

For $\cT \subseteq \{{T}_1, \ldots, {T}_{2k}\}$ we write $\str{A} \models  \cT ab$ iff $\str{A} \models
{T}_u ab$ for all ${T}_u \in \cT$.
A \emph{light declaration} is a function of type $\cP(\{{T}_1, \ldots, {T}_{2k}\}) \rightarrow \cP(\AAA)$. 
Given a light tree-like structure $\str{A}$ and its node $a$ we say that $a$ \emph{respects} a light declaration $\fd$ if 
for every $\cT$, for every $\alpha \in \fd(\cT)$ there is \textbf{no} node $b \in A$ of $1$-type $\alpha$ such that $\str{A} \models \cT ab$.
We denote by $\ldec{\str{A}}{a}$ 
the maximal light declaration respected by $a$. Formally, for every $\cT \subseteq \{{T}_1, \ldots, {T}_{2k}\}$, $\ldec{\str{A}}{a}(\cT)=\{ \alpha: \text{for every node $b$ of type $\alpha$ we have $\neg \str{A} \models  \cT ab$} \}$.
Intuitively, 
$\ldec{\str{A}}{a}$ says, for any combination of transitive relations, 
which $1$-types have {\bf no}  realizations to which 
$a$ is connected by this 
combination in $\str{A}$. 
Note that if $a$ respects a light declaration $\fd$ then
for any $\cT$ we have $\fd(\cT) \subseteq \ldec{\str{A}}{a}(\cT)$.
We remark that it would be equivalent to define the light declarations without the negations, listing the 1-types that a given node {\bf is} connected with, however we choose a version with negations to make them uniform with the corresponding (more complicated) notion in the general case,
where negations are more convenient.

Now we define the \emph{local consistency conditions (LCCs)} for a system of light declarations $(\fd_a)_{a\in A}$ assigned to all nodes of a tree-like structure $\str{A}$. Let $F$ be the downward family of some node $a$. 
 We say that the system \emph{satisfies LCCs at $a$} if 
 for every $a_1, a_2 \in F$ and for every $\cT$ such that $\str{A} \models  \cT a_1 a_2$ 
the following two conditions hold:
{\color{black!70} \bf (ld1)} for every $\alpha \in \AAA$, if $\alpha \in \fd_{a_1}(\cT)$  then $\alpha \in \fd_{a_2}(\cT)$,
{\color{black!70} \bf (ld2)} $\type{\str{A}}{a_2} \not\in \fd_{a_1}(\cT)$. 
Given a light tree-like structure $\str{A}$ we say that a  system of light declarations $(\fd_a)_{a\in A}$ is \emph{locally consistent} if it satisfies LCCs at each $a\in A$ and is \emph{globally consistent} if $\fd_a(\cT) \subseteq\ldec{\str{A}}{a}(\cT)$ for each $a\in A$ and each $\cT$.
Note that the global consistency means that all nodes $a$ respect their light declarations $\fd_a$. It is not difficult to see that local
and global consistency play along in the following sense.

\begin{lemma}[Local-global, light] \label{l:llocglobl}
Let $\str{A}$ be a light tree-like structure. Then, (i) if a system of light declarations $(\fd_a)_{a\in A}$ is locally consistent then it is globally consistent; and (ii) the \emph{canonical} system of light declarations, $(\ldec{\str{A}}{a})_{a\in A}$, is locally consistent.
\end{lemma}

Given a light tree-like structure $\str{A}$, by the \emph{generalized type} of a node $a$ of $\str{A}$ we will mean a pair $(\ldec{\str{A}}{a}$, $\type{\str{A}}{a})$, and denote it as $\gtype{\str{A}}{a}$. 
We introduce  a concept of top-down tree pruning.
Let $\str{A}$ be a light tree-like structure. A \emph{top-down tree pruning process} on $\str{A}$ has countably many steps $0, 1, 2, \ldots$, each of them producing
a new light tree-like structure by removing some nodes from the previous one and naturally stitching together the surviving nodes.
We emphasise that the universes of all structures build in this process are  subsets of the universe of the original structure $\str{A}$.
More specifically, we take $\str{A}_0:=\str{A}$, and 
having constructed $\str{A}_i$, $i \ge 0$ construct $\str{A}_{i+1}$ as follows. For every node $a$ of $\str{A}_i$ of depth $i+1$ (we assume that the root has depth $0$) either leave the subtree rooted at $a$ untouched or replace it by a subtree rooted at some descendant $b$ of $a$ having in the original structure $\str{A}$ the same generalized type as $a$,  and then
transitively close all transitive relations. The result of the process is a naturally defined limit structure $\str{A}'$,
in which the pair of elements $a, b$, of depth $d_a$ and $d_b$ respectively, has its $2$-type taken from $\str{A}_{\max(d_a, d_b)}$. 
Note that this $2$-type is not modified in the subsequent structures, so the definition is sound.

\begin{lemma}[Tree-pruning, light] \label{l:ltreepruning}
Let $\str{A}$ be a light tree-like structure. 
Let $(\fd_a)_{a \in A}$ be the canonical system of light declarations on $\str{A}$, $\fd_a:=\ldec{\str{A}}{a}$.
Let $\str{A}'$ be the result of  a top-town tree pruning process on $\str{A}$.
Then (i) the system of light declarations $(\fd_a)_{a \in A'}$ (the canonical declarations
from $\str{A}$ of the nodes surviving the pruning process)
in $\str{A}'$ is locally consistent, (ii) for any pair of 
elements $a,a' \in A'$ there is a homomorphism $\str{A}\restr \{a, a' \} \rightarrow A$ preserving the $1$-types; it also follows that 
(iii) for a normal form $\phi$, if $\str{A}$ is a model of $\phi$  such that any node $a$ has all its witnesses in its downward family then $\str{A}' \models \phi$.
\end{lemma}

It is not difficult to devise a strategy of top-down tree pruning 
leading to a model with short transitive paths in a simple scenario
where only one transitive relation is present.
With several transitive relations, however,  a quite intricate strategy seems to be required. The main obstacle is
that when decreasing the ${T}_u$-rank of an element $a$, for some $u$, we may accidentally increase the  ${T}_v$-rank 
of $a$ for some $v \not= u$. Nevertheless, an appropriate strategy exists (see Appendix \ref{s:strategy}), which allows us to state:

\begin{lemma}[(f2) $\leadsto$ (f3), light] \label{l:lshortpaths2}
Let $\phi$ be  a normal form \UNFOtTR{} formula. Let  $\str{A} \models \phi$ be a light tree-like structure over  signature $\sigma_\phi$,
with transitive paths bounded by some natural number $M$, such that each element has all the required witnesses in its downward family. Then $\phi$ 
has a light tree-like model with transitive paths bounded doubly exponentially in
$|\phi|$. 
\end{lemma}

\subsection{Finite model construction in the two-variable case}

In this section we show the following small model property. To this end, in particular, we will make the transition (f3) $\leadsto$ (f5).

\begin{theorem} \label{t:lmain}
Every finitely satisfiable two-variable \UNFOTR{} formula $\phi$ has a finite model of size bounded triply exponentially 
in $|\phi|$. 
\end{theorem}

Let us fix a finitely satisfiable normal form \UNFOTR{} formula $\phi$
over a signature $\sigma_\phi=\sigma_{\cBase} \cup \sigma_{\cDist}$ for
$\sigma_{\cDist}=\{{T}_1, \ldots, {T}_{2k} \}$. Denote by $\AAA$ the set of $1$-types over this signature.
Fix a light tree-like model $\str{A} \models \phi$, with linearly bounded degree and doubly exponentially bounded 
transitive paths (in this section we denote this bound by $\hat{M}_\phi$), as guaranteed by Lemma \ref{l:lshortpaths2}. 
We show how to build a `small' finite model $\str{A}' \models \phi$. 
For a  set  
$\mathcal{E} \subseteq \sigma_{\cDist}$, closed under inverses, and $a \in A$ we denote by $[a]_{\mathcal{E}}$ the set consisting of $a$ and all elements $b \in A$ such that
$\str{A} \models {T}_u ab$ for all ${T}_u \in \mathcal{E}$. 
Note that $[a]_\mathcal{E}$ is either a singleton or each of the ${T}_u \in \mathcal{E}$ is total on $[a]_\mathcal{E}$, that is, for each $b_1, b_2 \in [a]_\mathcal{E}$ we have $\str{A} \models {T}_u b_1 b_2$ for
all ${T}_u \in \mathcal{E}$.  
We note that $[a]_\emptyset=A$.

In our construction we inductively produce finite fragments of $\str{A}'$ corresponding to 
some (potentially infinite) classes $[a]_{\mathcal{E}}$ of $\str{A}$. Essentially, the induction goes
downward on the size of $\mathcal{E}$. 
Intuitively, if a relation is total then it plays no important role, so 
we may forget about it during the construction.
Every such fragment will be obtained by an appropriate arrangement of 
some number of basic building blocks, called \emph{components}. Each of the components is obtained by some number of  applications of 
the inductive assumption to situations in which a new pair of relations ${T}_{2u-1}$, ${T}_{2u}$ is added 
to $\mathcal{E}$.

 Let us formally state our inductive lemma. In this statement we do not explicitly include any bound on the size of
promised finite models, but such a bound will be implicit in the proof and will be presented later. 
Recall that $\str{A}$ is the model fixed at the beginning of this subsection.

\begin{lemma}[Main construction, light]\label{l:lfin} Let $a_0\in A$ and let $\mathcal{E}_0\subseteq\sigma_{\cDist}$ be closed under inverses, let $\mathcal{E}_{tot}:=\sigma_{\cDist} \setminus\mathcal{E}_0$. 
Let $\str{A}_0=\str{A}_{a_0}\restr[a_0]_{\mathcal{E}_{tot}}$. Then there exist a finite structure $\str{A}_0'$, a function $\fp:A_0'\to A_0$ and an element $a_0'\in A_0'$, called the \emph{\origin} of $\str{A}_0'$, such that
	 \begin{enumerate}[(b1)] 
	 	\item $A_0'$ is a singleton 
		or every symbol from $\mathcal{E}_{tot}$ is interpreted as the total relation on $\str{A}_0'$.\label{blone} 
	 	
	 	\item $\fp(a_0')=a_0$.\label{bltwo} 
	 	
	 	\item For each $a'\in A_0'$ and each $i$, if $\fp(a')$ has a child being its witness  for $\phi_i$ in $\str{A}_0$ then 
		$a'$ has  a witness for $\phi_i$ in $\str{A}_0'$. Moreover, 
		$\type{\str{A}_0'}{a'}=\type{\str{A}_0}{\fp(a')}$.
		 \label{blthree}

	 	\item 
	 	For every pair $a', b' \in A_0'$ there  exists a homomorphism $\fh:\str{A}_0'\restr \{a', b'\} \rightarrow \str{A}$ preserving $1$-types such that $\fh(a')=\fp(a')$, and for any 1-type $\alpha$ and $\mathcal{T}\subseteq\{1,\ldots,2k\}$, 
	 	if $\str{A}_0'\models\mathcal{T}a'b'$ and $\alpha\not\in\ldec{\str{A}}{\fp(b')}(\mathcal{T})$ then $\alpha\not\in\ldec{\str{A}}{\fp(a')}(\mathcal{T})$.
	 	
	 	 \label{blfour}	 	
	 \end{enumerate}
\end{lemma}

Observe first that Lemma \ref{l:lfin} indeed
allows us to build a particular finite model of $\phi$. 
Apply it to $\mathcal{E}_0=\sigma_{\cDist}$ (which means that $\mathcal{E}_{tot}=\emptyset$ and $[a_0]_{\mathcal{E}_{tot}} = A$)  and  $a_0$ being the root of $\str{A}$ (which means that $\str{A}_0=\str{A}$) and
use Lemma \ref{l:homomorphisms} to see that the obtained structure $\str{A}_0'$ is a model of $\phi$. Indeed, Condition (a1) of Lemma \ref{l:homomorphisms} follows directly from Condition (b\ref{blthree}), as in this case $\fp(a')$ has all witnesses in $\str{A}_0$.
 Condition (a2) is directly implied by Condition (b\ref{blfour}).

The proof of Lemma \ref{l:lfin} goes by induction on $l$, where $l=|\mathcal{E}_0|/2$. 
In the base of induction, $l=0$, we have $\mathcal{E}_{tot}=\sigma_{\cDist}$. Without loss of generality we may assume that 
the classes $[a]_{\mathcal{E}_{tot}}$ are singletons for all $a \in A$. (If this is not the case, we just add artificial transitive relations ${T}_{2k+1}$ and ${T}_{2k+2}$ both interpreted as the identity in $\str{A}$.)
We simply take $\str{A}'_0:=\str{A}_0= \str{A}\restr\{a_0\}$ and set $\fp(a_0)=a_0$. It is
readily verified that the conditions (b\ref{blone})--(b\ref{blfour}) are then satisfied.

For the inductive step assume that Lemma \ref{l:lfin}  holds for arbitrary  $\mathcal{E}_0$ closed under inverses, of size $2(l-1)<2k$.
We show that then it holds for $\mathcal{E}_0$ of size $2l$. Take such $\mathcal{E}_0$, and assume, w.l.o.g., that $\mathcal{E}_0=\{{T}_1,\ldots,{T}_{2l}\}$.  In the next two subsections we present a construction of $\str{A}_0'$. We argue that it is correct in Appendix \ref{s:cor}.
Finally
we estimate the size of the produced models and establish the complexity of the finite satisfiability problem.

%THE CONSTRUCTION:
\subsubsection{Pattern components}

We plan to construct $\str{A}_0'$ out of basic building blocks called \emph{components}. Each component will be 
an isomorphic copy of some {pattern component}. 

Let $\GGG[A_0]$ be the set of the generalized types realized in $\str{A}_0$. 
For every $\gamma \in \GGG[A_0]$ we construct two pattern structures, a \emph{pattern component} $\str{C}^\gamma$ and an \emph{extended pattern component} $\str{G}^\gamma$. $\str{C}^\gamma$ is a finite structure whose universe
is divided into $2l$ \emph{layers}
$L_1,\ldots,L_{2l}$.
$\str{G}^\gamma$ extends $\str{C}^\gamma$ by an additional, \emph{interface layer}, denoted $L_{2l+1}$. See the left part of Fig.~\ref{f:llayers}.
We now define $\str{G}^\gamma$, obtaining then $\str{C}^\gamma$ just by the restriction of $\str{G}^\gamma$ to 
non-interface layers.

%PATTERN COMPONENTS: 
 
Each non-interface layer $L_i$ is further divided into  \emph{sublayers} $L_i^1,L_i^2,\ldots, L_i^{\hat{M}_\phi+1}$.
Additionally, in each sublayer $L_i^j$ its initial part $L_i^{j,init}$ is distinguished.
In particular, $L_1^{1,init}$ consists of a single element called the \emph{root}.
The interface layer $L_{2l+1}$ has no internal division but, for convenience, is sometimes
referred to as $L_{2l+1}^{1,init}$.
The elements of $L_{2l}$ are called \emph{leaves} and 
the elements of $L_{2l+1}$ are called \emph{interface elements}. See  Fig.~\ref{f:llayers}.

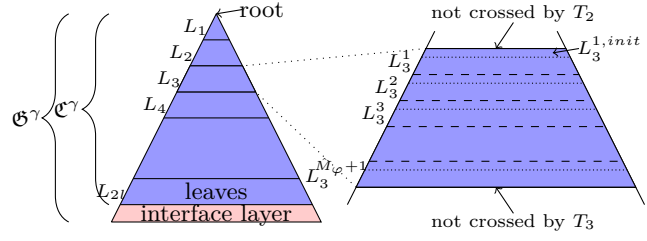
\begin{figure}
\begin{center}
\begin{tikzpicture}[scale=0.23]

\draw [decorate,decoration={brace,amplitude=10pt},xshift=-4pt,yshift=0pt]
(0,1) -- (0,12) node [black,midway,xshift=-0.55cm] 
{\footnotesize $\str{C}^\gamma$};

\draw [decorate,decoration={brace,amplitude=10pt},xshift=-4pt,yshift=0pt]
(-2.3,0) -- (-2.3,12) node [black,midway,xshift=-0.55cm] 
{\footnotesize $\str{G}^\gamma$};

%tlo
\fill[opacity=0.4, color = blue] (0.5,1)--(6,12)--(12,0) -- (11.5,1)--(0.5,1);
\fill[opacity=0.2, color = red] (0,0) -- (0.5,1) -- (11.5,1) -- (12,0) -- (0,0);

%czesc lewa
%triangles
	\draw (0,0) -- (6,12) -- (12,0) -- (0,0);

%layers lewe
    \draw (5.25,10.5) -- (6.75,10.5);%L1
		\draw (4.5,9) -- (7.5,9);%L2
		\draw (3.75,7.5) -- (8.25,7.5);%L3
		\draw (3,6) -- (9,6); %L4
		\draw (0.5,1) -- (11.5,1);
		\draw (1.25,2.5) -- (10.75,2.5);
		
		\draw[fill=blue!40] (3.75,7.5) -- (4.5,9) -- (7.5,9) -- (8.25,7.5) -- (3.75,7.5);

%podpisy

\coordinate [label=center:$\scriptstyle L_1$] (A) at ($(4.8,11.25)$); 
\coordinate [label=center:$\scriptstyle L_2$] (A) at ($(4.0,9.75)$); 
\coordinate [label=center:$\scriptstyle L_3$] (A) at ($(3.3,8.25)$); 
\coordinate [label=center:$\scriptstyle L_4$] (A) at ($(2.5,6.75)$); 
\coordinate [label=center:$\scriptstyle L_{2l}$] (A) at ($(0,1.75)$); 
\coordinate [label=center:{\small leaves}] (A) at ($(6,1.8)$); 
\coordinate [label=center:{\small interface layer}] (A) at ($(6,0.4)$);

\coordinate [label=center:{\small root}] (A) at ($(8.6,12.25)$); 
\draw[->] (7.5,12.25) -- (6,12);

%rozszerzenie szarego obszaru
\draw[dotted] (18,10) -- (7.5,9);
\draw[dotted] (14,2) -- (8.25,7.5);

\begin{scope}[shift={(-2,0)}]
%czesc prawa:
%triangles
	\draw (15.5,1) -- (20.5,11);
	\draw (32.5,1) -- (27.5,11);

%layers prawe

	\draw[fill=blue!40] (16,2) -- (20,10) -- (28,10) -- (32,2) -- (16,2);

   \draw (16,2) -- (32,2);
	\draw (20,10) -- (28,10);
	\draw[dashed] (19.25,8.5) -- (28.75,8.5);
	\draw[dashed] (18.5,7) -- (29.5,7);		
	\draw[dashed] (17.75,5.5) -- (30.25,5.5);
	\draw[dashed] (16.75,3.5) -- (31.25,3.5);
	
	\draw[densely dotted] (19.75,9.5) -- (28.25,9.5);
	\draw[densely dotted] (19,8) -- (29,8);
	\draw[densely dotted] (18.25,6.5) -- (29.75,6.5);
	\draw[densely dotted] (16.5,3) -- (31.5,3);

\coordinate [label=center:$\scriptstyle L_3^1$] (A) at ($(18.5,9.25)$); 
\coordinate [label=center:$\scriptstyle L_3^2$] (A) at ($(17.75,7.75)$); 
\coordinate [label=center:$\scriptstyle L_3^3$] (A) at ($(17,6.25)$); 
\coordinate [label=center:$\scriptstyle L_3^{\hspace*{-2pt}\hat{M}_\phi{+}1}$] (A) at ($(14.9,2.85)$);

\coordinate [label=center:$\scriptstyle L_3^{1,init}$] (A) at ($(30.5,10)$); 
\draw[->] (28.75,10) -- (27.25,9.75);

\coordinate [label=center:$\scriptstyle \text{not crossed  by } {T}_2$] (A) at ($(25,12)$); 
\draw[->] (25,11.5) -- (24,10);

\coordinate [label=center:$\scriptstyle \text{not crossed  by }{T}_3$] (A) at ($(25,0)$); 
\draw[->] (25,0.5) -- (24,2);
\end{scope}

\end{tikzpicture}
\caption{A schematic view of a component in the two-variable case.
}
\label{f:llayers}%
\end{center}
\end{figure}

$\str{G}^\gamma$ will have a shape resembling a tree, with  structures obtained by the inductive assumption as nodes, though
it will not be tree-like in the sense of Section \ref{s:ltreelike} (in particular, the internal structure of nodes may
be complicated). All elements of $\str{G}^\gamma$, except for the interface elements, will have appropriate witnesses (those required
by (b\ref{blthree})) provided.
The crucial property we want to enforce is that the root of $\str{G}^\gamma$ will not be joined to its interface elements 
by any
transitive path.

We remark that during the process of building a pattern component we do not yet apply the
transitive closure to the distinguished relations. Postponing this step is not important from the point of view
of the correctness of the construction, but will allow us for a more precise
presentation of the proof of this correctness.
Given a component  $\str{C}$ (extended component $\str{G}$) we will sometimes denote by $\str{C}_+$ ($\str{G}_+$) the structure obtained from $\str{C}$ ($\str{G}$) by
applying all the appropriate transitive closures. 

The role of every non-interface layer $L_u$ 
is, speaking informally, to \emph{kill}  ${T}_u$, that is to ensure that
there will be no ${T}_u$-connections from $L_u$  to $L_{u+1}$. 
See the right part of Fig.~\ref{f:llayers}.
The role of sublayers of $L_u$,
on the other hand, is to decrease the ${T}_u$-rank of 
the patterns of
elements. 
The purpose of the interface layer, 
$L_{2l+1}$, 
will be to connect the component with other components.

If $\gamma$ is the generalized type of
${a_0}$ then take $a:=a_0$; otherwise take as $a$ any element of $A_0$ of generalized type $\gamma$.
We begin the construction of $\str{G}^\gamma$ by defining $L_1^{1,init}=\{a'\}$ for a fresh $a'$, setting $\type{\str{G}^\gamma}{a'}=\type{\str{A}}{a}$ and $\fp(a')=a$.

\smallskip\noindent
\emph{Construction of a  layer}: Let $1\leq u \leq 2l$.
Assume we have defined layers $L_1,\ldots,L_{u-1}$, the initial part of sublayer $L_u^1$, $L_u^{1,init}$, and both the structure of $\str{G}^\gamma$ and the values of $\fp$ on $L_1 \cup \ldots \cup L_{u-1} \cup L_u^{1,init}$. We are going to kill ${T}_u$. 
We now expand $L_u^{1, init}$ to a full layer $L_u$.

\smallskip\noindent
%STEP 1 (SUBCOMPONENTS) 
\emph{Step 1: Subcomponents.} Assume that we have defined 
sublayers $L_u^{1}, \ldots,L_u^{j,init}$, and both the structure of $\str{G}^\gamma$ and the values of $\fp$ on
$L_1 \cup \ldots \cup L_{u-1} \cup L_u^{1} \cup \ldots \cup  L_u^{j,init}$. 
For each  $b\in L_u^{j,init}$ perform independently the following procedure. 
Apply the inductive assumption to $\fp(b)$ and the set $\mathcal{E}_0 \setminus \{ {T}_u, {T}_{u}^{-1}\}$ obtaining a structure $\str{B}_0$, its \origin{}
$b_0$ and a function $\fp_b:B_0\to A_{\fp(b)}\cap[\fp(b)]_{\mathcal{E}_{tot} \cup \{{T}_u, {T}_{u}^{-1}\}}\subseteq A_0$ with $\fp_b(b_0)=\fp(b)$. 
Identify $b_0$ with $b$ and add the remaining elements of $\str{B}_0$ to $L_u^j$, retaining the structure. 
Substructures $\str{B}_0$ of this kind will be called \emph{subcomponents} (note that all appropriate relations
are transitively closed in subcomponents).
Extend $\fp$ so that $\fp\restr B_0=\fp_b$. This finishes the
definition of $L_u^j$.

\smallskip\noindent
%STEP 2 (PROVIDING WITNESSES) 
\emph{Step 2: Providing witnesses.} For each $b\in L_u^j$ and $1 \le s \le m$ independently perform the following procedure. Let $\str{B}_0$ be the subcomponent created inductively in \emph{Step 1}, such that $b \in B_0$. 
If 
$\fp(b)$ has a witness  for $\phi_s(x,y)$ in $A_0$ then we want to reproduce such a witness for $b$.
Choose one such witness $c$ (being a child of $\fp(b)$) for $\fp(b)$.
Let us denote $\beta=\type{\str{A}}{\fp(b),c}$. 
If $\{{T}_u xy, {T}_{u}^{-1}xy \} \subseteq \beta$ then by Condition (b\ref{blthree}) of the inductive assumption $b$ already has an appropriate witness in 
the subcomponent $\str{B}_0$. So we do nothing in this case. 
If ${T}_u xy \in \beta$ and ${T}_{u}^{-1} xy \not\in \beta$
then we add a copy $c'$ of $c$ to $L_u^{j+1, init}$; 
if ${T}_u xy \not\in \beta$  then we add a copy $c'$ of $c$ to  $L_{u+1}^{1, init}$. 
We  join $b$ with $c'$ by $\beta$ and set $\fp(c')=c$.

An attentive reader may be afraid that when adding witnesses for elements of the last sublayer $L_{u}^{\hat{M}_\phi+1}$
of $L_u$ we may want to add
one of them to the non-existing layer $L_{u}^{\hat{M}_\phi+2}$. 
There is however no such danger, which follows from the following claim.

\begin{claim} \label{c:lnodanger} 
(i) Let $b\in L_u^{j,init}$ and let $\str{B}_0$ be the subcomponent created for $b$ in \emph{Step 1}. Then for all $b'\in\str{B}_0$ we have $\rank{\str{A}}{u}{\fp(b)}\geq\rank{\str{A}}{u}{\fp(b')}$.
(ii) Let $b\in L_u^j$ and let $c' \in L_u^{j+1}$ be a witness created for $b$ in \emph{Step} 2. Then  $\rank{\str{A}}{u}{\fp(b)}>\rank{\str{A}}{u}{\fp(c')}$. 
\end{claim}

Hence, when moving from $L_u^{j}$ to $L_u^{j+1}$ the ${T}_u$-ranks of pattern elements for the elements of these sublayers strictly decrease. Since these ranks are bounded by $\hat{M}_\phi$, then, even if the ${T}_u$-ranks of the patterns of some elements of $L_u^{1}$ are equal to $\hat{M}_\phi$, then, if $L_u^{\hat{M}_\phi+1}$ is non-empty,
the ${T}_u$-ranks of the patterns of its elements must be $0$, which means that they cannot have witnesses connected to them one-directionally by ${T}_u$.

The construction of  $\str{G}^\gamma$ is finished when layer $L_{2l}$ is fully processed. We have added some elements to the interface layer, $L_{2l+1}$. Recall that it has only
its `initial part'.

\subsubsection{Joining the components}
 
In this section we take some number of copies of pattern components and arrange them into the desired structure $\str{A}'_0$,
identifying interface elements of some components with the roots of some other.
Some care is needed in this process in order
to avoid any modifications of the internal structure
of closures $\str{C}_+$ of components $\str{C}$, which could potentially result from the transitivity of relations. 
In particular we need to ensure that if for some $u$ a pair of elements of a component $\str{C}$ is not connected by 
${T}_u$ inside $\str{C}$, then it will not become connected by a chain of ${T}_u$-edges external to $\str{C}$.

We create a pattern component $\str{C}^\gamma$ and its extension $\str{G}^\gamma$ for every $\gamma \in \GGG[A_0]$.
Let $\gamma_{a_0}$ be the generalized type of $a_0$.
Let $max$ be the maximal number of interface elements across all the $\str{G}^{\gamma}$. For each $\str{G}^{\gamma}$ 
arbitrarily number its interface elements from $1$ up to, maximally, $max$.

For each $\gamma$ we take copies $\str{G}^{\gamma,g}_{i,\gamma'}$ of $\str{G}^{\gamma}$ for $g\in\{0,1\}$, $1\leq i\leq max$ and $\gamma'\in\GGG[A_0]$. The parameter $g$ is sometimes called a \emph{color} (\emph{red} or \emph{blue}); it is convenient to  think that the non-interface elements of ${G}^{\gamma,g}_{i,\gamma'}$ are of color $g$, but
its interface elements have color $1-g$, cf.~the left part of Fig.~\ref{f:llayers}, as the latter will be later identified with the roots of some components of color $1-g$. 
We import the numbering of the interface elements to these copies.
We also take an additional copy $\str{G}^{\gamma_{a_0},0}_{\bot,\bot}$ of $\str{G}^{\gamma_{a_0}}$. Its root will become the \origin{} of the whole $\str{A}_0'$.
By  $\str{C}^{\gamma,g}_{i,\gamma'}$ we denote the restriction of $\str{G}^{\gamma,g}_{i,\gamma'}$ to its non-interface elements.

 For each $\gamma$, $g$ consider extended components of the form $\str{G}^{\gamma,g}_{\cdotp,\cdotp}$, where the placeholders $\cdotp$ can be substituted with any combination of proper indices. Perform the following
procedure for each $1 \le i \le max$.  
Let $b$ be the $i$-th interface element of any such extended component, let  $\gamma'$ be the generalized type of ${\fp(b)}$.
Identify the $i$-th interface elements of all $\str{G}^{\gamma,g}_{\cdotp,\cdotp}$
with the root $c_0$ of $\str{G}^{\gamma',1-g}_{i,\gamma}$. 
Note that the values of $\fp(c_0)$ and $\fp(b)$ 
may differ. However, by construction, they have identical generalized types $\gamma'$.
For the element $c^*$ obtained in this identification step we 
define $\fp(c^*)=\fp(c_0)$.

Define the graph of components used in the above construction, $G^{comp},$ by joining two components by an edge iff we identified an interface element of the extended version of one of them with the root of the other. 
Let $A_0^0$ be the union of the components accessible from $\str{C}^{\gamma_{a_0},0}_{\bot,\bot}$ in $G^{comp}$ and
let  $\str{A}_0^0$ be the induced structure.
Note that in $\str{A}_0^0$ we still do not take the transitive closures of relations. 
We define $\str{A}_0'$ by transitively closing all relations from $\sigma_\cDist$ 
in $\str{A}_0^0$. 
Finally, we choose as the \origin{} $a_0'$ of $\str{A}_0'$ the root of the pattern component $\str{C}^{\gamma_0,0}_{\bot,\bot}$. 

We remark that it is sufficient to take as the universe of $\str{A}_0'$ the union of the universes of some components $\str{C}^{\cdot, \cdot}_{\cdot,\cdot}$,
and not of their extended versions $\str{G}^{\cdot, \cdot}_{\cdot,\cdot}$  from which we started our construction, since the interface elements from these
extended components
were identified with some roots of other components.

For the correctness proof of our construction see Appendix \ref{s:cor}.
In this proof it is helpful to think about $\str{A}_0^0$ and $\str{A}_0'$ as the structures placed on a cylindrical surface and divided into $4l$ \emph{levels}, see Fig.~\ref{f:cylinder0}. 
 What is crucial, any transitive path in $\str{A}_0^0$ can cross at most one of the two borders between colors.

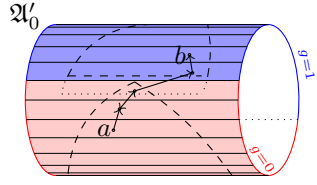
\begin{figure} 
\begin{center}
\begin{tikzpicture}

\begin{scope}[scale=0.8]

\node at (-4,1.4) {$\str{A}'_0$};

%kontury 
  \draw[red]  (0,-1.25)--(-3.5,-1.25);
	\draw[blue]  (0,1.25)--(-3.5,1.25);
  
	\draw[blue] (-3.5,1.25) arc(90:165:0.5 and 1.25);
	\draw[blue] (0,1.25) arc(90:165:0.5 and 1.25);
	\draw[blue, postaction={decorate, decoration={text along path, raise=3pt, text align={align=center}, text color=blue, text={${\scriptscriptstyle g=1}${} }}}] (0,1.25) arc(90:-15:0.5 and 1.25) node (bck) {}; 
	
	\path (0,1.25) arc(90:165:0.5 and 1.25) node (tmp) {} ;
	\draw[red, postaction={decorate, decoration={text along path, raise=3pt, text align={align=center}, text color=red, text={${\scriptscriptstyle g=0}${} }}}] (tmp.center) arc(165:345:0.5 and 1.25); 

	\path (-3.5,1.25) arc(90:165:0.5 and 1.25) node (tmp) {} ;
  \draw[red] (tmp.center) arc(165:270:0.5 and 1.25);

%niebieskie i czerwone tlo
	\fill [opacity=0.4, color = blue] (0,1.25) -- (-3.5,1.25) arc
  (90:165:0.5 and 1.25) -- ++(3.5,0)  arc(165:90:0.5 and 1.25) ;
	\fill [opacity=0.2, color = red] (0,-1.25) -- (-3.5,-1.25) arc
  (270:165:0.5 and 1.25) -- ++(3.5,0)  arc(165:270:0.5 and 1.25) ;
	
%linia wewnetrzna i nozyczki
\draw[dotted] (bck.center) -- +(-1,0);	
%\node at (0.7,-0.3) (scissors) {\Leftscissors} ;

%poziome linie	

\foreach \t in {90,105,...,270}
{
\path (0,1.25) arc(90:\t:0.5 and 1.25) node (a) {} ;
\draw[very thin] (a.center) -- +(-3.5,0);
};

%sciezka
\node at (-2.7,-0.5) {$a$};
\node at (-1.45,0.75) {$b$};

\draw (-2.55,-0.5) circle (0.02);
\draw (-2.45,-0.15) circle (0.02);
\draw (-2.2,0.15) circle (0.02);
\draw (-1.25,0.45) circle (0.02);
\draw (-1.3,0.75) circle (0.02);

\draw[->] (-2.55,-0.5) -- (-2.45,-0.15);
\draw[<->] (-2.45,-0.15) -- (-2.2,0.15);
\draw[->] (-2.2,0.15) -- (-1.25,0.45);
\draw[->] (-1.25,0.45) -- (-1.3,0.75);

%

%komponenty
\draw[dashed] (-3.3,0.4) to[bend left] (-2,1.25);
\draw[dashed] (-3.3,0.4) to (-1.1,0.4);
\draw[dashed] (-1,0.4) to[bend right=10] (-1,1.25);

\draw[dotted] (-3.3,0.4) to[bend right=10] (-3.4,0.1);
\draw[dotted] (-3.4,0.1) to (-1.1,0.1);
\draw[dotted] (-1,0.4) to[bend left=10] (-1.1,0.1);

%\draw[dotted] (0.7,1.3) -- (0.5, 1.0) -- (3.0, 1.0) -- (2.8,1.3);

\draw[dashed] (-2.2,0.3) to[bend right] (-3.3,-1.25);
\draw[dashed] (-2.2,0.3) to[bend left=10] (-0.6,-1.25);

 \end{scope}
\end{tikzpicture}
\caption{Viewing  $\str{A}^0_0$ and $\str{A}_0'$ as placed on a cylindrical surface.}
\label{f:cylinder0}
\end{center} 
\end{figure}

\subsubsection{Size of models and complexity}

By a rather routine calculation we can show that models produced in the proof of Lemma \ref{l:lfin} are
of size bounded triply exponentially in the length of input formulas.
This finishes the proof of Thm.~\ref{t:lmain}, which 
immediately gives the decidability of the finite satisfiability problem for \UNFOtTR{} and
 suggests a simple \ThreeNExpTime{}-procedure: guess a finite structure
of size bounded triply exponentially in the size of input $\phi$ and verify that it is indeed a model of $\phi$. We can however do better and show 
a doubly exponential upper bound matching 
the known complexity of the general satisfiability problem. For this we design an alternating exponential space 
algorithm searching for models of the form (f3).
The lower bound  can be 
obtained for the two-variable \UNFOtTR{} in the presence of one transitive relation by
a straightforward adaptation of the lower bound proof for \GFt{} with transitive guards \cite{Kie06}.

\begin{theorem}\label{t:lalgo}
The finite satisfiability problem for \UNFOtTR{} is \TwoExpTime-complete.
\end{theorem}

\section{The general case and its further extensions} \label{s:alltheother}

In Appendix \ref{s:general}
we generalize the ideas from Section \ref{s:2var} to show:

\begin{theorem}\label{t:algo}
The finite satisfiability problem for \UNFOTR{} is \TwoExpTime-complete.
\end{theorem}

We also obtain a triply exponential upper bound on the size of minimal finite models of finitely satisfiable formulas.
The structure of the proofs is similar to the two-variable case, though some details are more complicated. In particular,
we need to go through form (f4) of models: regular trees with bounded ranks of transitive paths. We also explain that in addition to  general transitive relations
we can use also equivalences and partial orders.

In Appendix \ref{s:edl} we further extend Thm.~\ref{t:algo} by considering an extension, \UNFOSOH{}, of \UNFOTR{} by constants and inclusion of binary relations
of the form $B_1 \subseteq B_2$, interpreted in a natural way: $\str{A} \models B_1 \subseteq B_2$ iff $\str{A} \models
\forall xy (B_1xy \rightarrow B_2xy)$.
\begin{theorem} \label{t:full}
The finite satisfiability problem for \UNFOSOH{} is \TwoExpTime-complete.
\end{theorem}

As mentioned in the Introduction, \UNFOSOH{} captures several interesting description logics. This implies that we
can solve FOMQA problem for them. In particular, we have the following corollary, which, up to our knowledge is the first decidability result for FOMQA
in the case of a description logic with both transitive roles and role hierarchies.   
\begin{corollary}\label{c:fomqa}
Finite ontology mediated query answering, FOMQA, for the description logic $\mathcal{SHOI}^{\sqcap}$ is decidable and \TwoExpTime-complete.
\end{corollary}

$\mathcal{SHOI}^{\sqcap}$ and some related logics are considered, e.g., in \cite{GK08}. For more about FOMQA for description logics with transitivity see \cite{GYM18}. For more about OMQA for description logics see, e.g., references in \cite{GYM18}.

Somewhat orthogonally to the extensions motivated by description logics in Appendix \ref{s:tgn}
we consider  
 the \emph{base-guarded negation fragment with transitivity}, \GNFOTR{}, for which the general satisfiability problem
 was shown decidable in \cite{ABBB16}.
We do not solve its finite satisfiability problem here, but, analogously to the extension with equivalence relations,
\UNFOEQ{} \cite{DK18}, we are able to lift our results to its one-dimensional restriction, \GNFOTRonedim{},
admitting only formulas in which every maximal block of quantifiers leaves at most one variable free.

\begin{theorem}\label{t:gnfo}
The finite satisfiability problem for \GNFOTRonedim{} is \TwoExpTime-complete.
\end{theorem}

Surprisingly, in contrast to \UNFOTR{}, \GNFOTRonedim{} becomes undecidable when extended by inclusions of binary relations.

\section{Conclusions} \label{s:conclusions}

We proved that the finite satisfiability problem for 
the unary negation fragment with transitive relations, 
\UNFOTR{}, is decidable and
\TwoExpTime-complete, complementing this way the analogous result for the general satisfiability problem for this logic implied
by two other papers.  Further, we identified some decidable extensions of our base logic capturing the concepts of nominals
and role hierarchies from description logics.
We noted that our work has some interesting implications on the finite query answering problem both 
under the classical (open-world) database scenario as well as  in the description logics setting. 

One open question is the decidability of the finite satisfiability problem for the full logic \GNFOTR{} from \cite{ABBB16}.
We made a step in this direction here, by solving this problem for the one-dimensional restriction of that logic.
Another question is if our techniques can be adapted to a setting in which we do not assert that some distinguished relations
are transitive but where we can talk about the transitive closure of the binary relations, or, more generally, to the extension
of \UNFO{} with regular path expressions from \cite{JLM18}.

We finally remark that we do not know if our small model construction, producing finite models of size bounded triply
exponentially in the size of the input formulas,  is optimal  with 
respect to the size of models. The best we can do for the lower bound is to enforce models of 
 doubly exponential size (actually, this can be done in UNFO even without transitive relations).

\bibliography{mybib}

\begin{thebibliography}{10}

\bibitem{ABBB16}
A.~Amarilli, M.~Benedikt, P.~Bourhis, and M.~{Vanden Boom}.
\newblock Query answering with transitive and linear-ordered data.
\newblock In {\em Proceedings of the 25th International Joint Conference on
  Artificial Intelligence, {IJCAI} 2016}, pages 893--899, 2016.

\bibitem{ABN98}
H.~Andr{\'e}ka, J.~van Benthem, and I.~N\'{e}meti.
\newblock Modal languages and bounded fragments of predicate logic.
\newblock {\em Journal of Philosophical Logic}, 27:217--274, 1998.

\bibitem{BLM10}
J.-F. Baget, M.~LeClere, and M.-L. Mugnier.
\newblock Walking the decidability line for rules with existential variables.
\newblock In {\em Proceedings of the 12th International Conference on
  Principles of Knowledge Representation and Reasoning, {KR} 2010}, pages
  466--476, 2010.

\bibitem{BLL09}
J.{-}F. Baget, M.~Lecl{\`{e}}re, M.{-}L. Mugnier, and E.~Salvat.
\newblock Extending decidable cases for rules with existential variables.
\newblock In {\em Proceedings of the 21st International Joint Conference on
  Artificial Intelligence, {IJCAI} 2009}, pages 677--682, 2009.

\bibitem{BtCS15}
V.~B{\'{a}}r{\'{a}}ny, B.~ten Cate, and L.~Segoufin.
\newblock Guarded negation.
\newblock {\em J. {ACM}}, 62(3):22, 2015.

\bibitem{CGK13}
A.~Cal{\`{\i}}, G.~Gottlob, and M.~Kifer.
\newblock Taming the infinite chase: Query answering under expressive
  relational constraints.
\newblock {\em J. Artif. Intell. Res.}, 48:115--174, 2013.

\bibitem{DK18}
D.~Danielski and E.~Kiero\'nski.
\newblock Unary negation fragment with equivalence relations has the finite
  model property.
\newblock In {\em 33rd Annual ACM/IEEE Symposium on Logic in Computer Science,
  {LICS} 2018}, pages 285--294, 2018.

\bibitem{DK18b}
D.~Danielski and E.~Kieronski.
\newblock Unary negation fragment with equivalence relations has the finite
  model property.
\newblock {\em CoRR}, abs/1802.01318, 2018.
\newblock \href {http://arxiv.org/abs/1802.01318} {\path{arXiv:1802.01318}}.

\bibitem{Dzi17}
M.~Dziecio{\l}owski.
\newblock Satisfability issues for unary negation logic.
\newblock Bachelor's thesis, University of Wroc\l{}aw, 2017.

\bibitem{GMV99}
H.~Ganzinger, Ch. Meyer, and M.~Veanes.
\newblock The two-variable guarded fragment with transitive relations.
\newblock In {\em 14th Annual IEEE Symposium on Logic in Computer Science,
  {LICS 1999}}, pages 24--34, 1999.

\bibitem{GK08}
B.~Glimm and Y.~Kazakov.
\newblock Role conjunctions in expressive description logics.
\newblock In {\em Logic for Programming, Artificial Intelligence, and
  Reasoning, 15th International Conference, {LPAR} 2008}, pages 391--405, 2008.

\bibitem{GYM18}
T.~Gogacz, Y.~A. Ib{\'{a}}{\~{n}}ez{-}Garc{\'{\i}}a, and F.~Murlak.
\newblock Finite query answering in expressive description logics with
  transitive roles.
\newblock In {\em Principles of Knowledge Representation and Reasoning:
  Proceedings of the Sixteenth International Conference, {KR} 2018.}, pages
  369--378, 2018.

\bibitem{Gra99}
E.~Gr{\"a}del.
\newblock On the restraining power of guards.
\newblock {\em J. Symb. Log.}, 64(4):1719--1742, 1999.

\bibitem{GKV97}
E.~Gr{\"a}del, P.~Kolaitis, and M.~Y. Vardi.
\newblock On the decision problem for two-variable first-order logic.
\newblock {\em Bulletin of Symbolic Logic}, 3(1):53--69, 1997.

\bibitem{GOR99}
E.~Gr{\"a}del, M.~Otto, and E.~Rosen.
\newblock Undecidability results on two-variable logics.
\newblock {\em Archiv f{\"{u}}r Mathematische Logik und Grundlagenforschung},
  38(4-5):313--354, 1999.

\bibitem{JLM18}
J.~Ch. Jung, C.~Lutz, M.~Martel, and T.~Schneider.
\newblock Querying the unary negation fragment with regular path expressions.
\newblock In {\em International Conference on Database Theory, {ICDT} 2018},
  pages 15:1--15:18, 2018.

\bibitem{Kaz06}
Y.~Kazakov.
\newblock {\em Saturation-based decision procedures for extensions of the
  guarded fragment}.
\newblock PhD thesis, Universit\"at des Saarlandes, Saarbr\"ucken, Germany,
  2006.

\bibitem{Kie05}
E.~Kiero\'nski.
\newblock Results on the guarded fragment with equivalence or transitive
  relations.
\newblock In {\em Computer Science Logic}, volume 3634 of {\em LNCS}, pages
  309--324. Springer, 2005.

\bibitem{Kie06}
E.~Kiero\'nski.
\newblock On the complexity of the two-variable guarded fragment with
  transitive guards.
\newblock {\em Inf. Comput.}, 204(11):1663--1703, 2006.

\bibitem{KT18}
E.~Kiero\'{n}ski and L.~Tendera.
\newblock Finite satisfiability of the two-variable guarded fragment with
  transitive guards and related variants.
\newblock {\em ACM Trans. Comput. Logic}, 19(2):8:1--8:34, 2018.

\bibitem{P-H18}
I.~Pratt-Hartmann.
\newblock The finite satisfiability problem for two-variable, first-order logic
  with one transitive relation is decidable.
\newblock {\em Mathematical Logic Quarterly}, 2018.

\bibitem{P-HST19}
I.~Pratt-Hartmann, W.~Szwast, and L.~Tendera.
\newblock The fluted fragment revisited.
\newblock {\em Journal of Symbolic Logic}, Forthcoming, 2019.

\bibitem{P-HT19}
I.~Pratt-Hartmann and L.~Tendera.
\newblock The fluted fragment with transitivity.
\newblock In {\em 44th International Symposium on Mathematical Foundations of
  Computer Science, {MFCS} 2019}, pages 15:1--15:15, 2019.

\bibitem{Q69}
W.~V. Quine.
\newblock On the limits of decision.
\newblock In {\em Proceedings of the 14th International Congress of
  Philosophy}, volume III, pages 57--62, 1969.

\bibitem{Sco62}
D.~Scott.
\newblock A decision method for validity of sentences in two variables.
\newblock {\em Journal Symbolic Logic}, 27:477, 1962.

\bibitem{ST04}
W.~Szwast and L.~Tendera.
\newblock The guarded fragment with transitive guards.
\newblock {\em Annals of Pure and Applied Logic}, 128:227--276, 2004.

\bibitem{StC13}
B.~ten Cate and L.~Segoufin.
\newblock Unary negation.
\newblock {\em Logical Methods in Comp. Sc.}, 9(3), 2013.

\end{thebibliography}

\begin{appendix}

\newpage

\section{\UNFO{} with functional restrictions} \label{a:frund}

The following result is implicit in \cite{StC13}. Here we prove it directly, by a simple reduction from domino tilings.
\begin{theorem} \label{t:cund}
The satisfiability and the finite satisfiability problems for  \UNFO{} with four variables (even without
transitive relations) and two functional binary relations and some unary relations are undecidable.
\end{theorem}

\begin{proof} 
We axiomatize a class of grid structures, with $H$ being the horizontal successor relation and $V$ being the vertical
successor relation, as follows. We assert that $H$ and $V$ are functional
$$\forall x (\exists^{=1} y Hxy \wedge \exists^{=1} y Vxy)  $$
and then define grids using the following \UNFO{} formula
$$
 \forall x (\exists yzt (Hxy \wedge Vxz \wedge Vyt \wedge Hzt))
$$

Let $\lambda$ be the conjunction of the two above formulas. Clearly, the standard grids on $\N \times \N$ and on the $t \times t$ tori, for $t \in \N$, are models of  $\lambda$. Conversely, any model of $\lambda$ homomorphically embeds the standard grid. Having $\lambda$ with
such  properties, 
reducing the undecidable domino tiling problem: \emph{given a domino system verify if it tiles} $\N \times \N$ (\emph{some} $\Z_t \times \Z_t$) 
to satisfiability (finite satisfiability) is routine. 
\end{proof}

\section{Normal form} \label{s:nf}

\begin{lemma}[Scott-normal form] \label{l:nf}
For every \UNFOTR{} sentence $\phi$ one can compute in polynomial time a normal form \UNFOTR{} sentence
$\phi'$ over signature extended by some fresh unary symbols, such that 
any model of $\phi'$ is a model of $\phi$ and any model of $\phi$ can be
expanded to a model of $\phi'$ by an appropriate interpretation of the additional
unary symbols. 
\end{lemma}

\begin{proof} (Sketch)
Take any \UNFOTR{} sentence $\phi$. Recall that it uses no universal quantifiers. First we convert it to its UN-normal form,
in which each maximal block of quantifiers leaves at most one variable free. This can be done as described in \cite{StC13}. 
Then we consider an innermost subformula of $\phi$, starting with a block of quantifiers, $\exists \bar{y} \psi(x, \bar{y})$,
replace it by a fresh unary predicate $P(x)$, and add two auxiliary conjuncts $\forall x \exists \bar{y} (\neg P(x) \vee \psi(x, \bar{y}))$
and $\forall x \bar{y} \neg(\psi(x, \bar{y})\vee\neg P(x))$, whose conjunction is equivalent to $\forall x (P(x)\leftrightarrow\exists\bar{y}\psi(x,\bar{y}))$.
Moving up the original formula $\phi$ we repeat this procedure for subformulas that are now innermost, and so forth. The formula
obtained in this process has, up to trivial logical transformations, the desired shape and properties.
\end{proof}

\section{Proof of Lemma \ref{l:homomorphisms}}

\begin{proof}
Due to (a1) all elements of $\str{A}'$ have the required witness structures for all $\forall\exists$-conjuncts. It remains
to see that the $\forall$-conjunct is not violated. But since $\str{A}  \models \neg \phi_0(\fh(a_1), \ldots, \fh(a_t))$ 
and $\phi_0$ is a quantifier-free formula in which only unary atoms may be negated, it is straightforward, using (a2). 
\end{proof}

\section{Missing proofs and a strategy description from Section \ref{s:2var}}

\subsection{Proof of Lemma \ref{l:ltreelike}}
\begin{proof}
It is readily verified that $\str{A}'$ meets the properties required by Lemma \ref{l:homomorphisms}. In particular function $\fh$ associated with the given unravelling is the
required homomorphism. That $\str{A}'$ is tree-like and has an appropriately bounded degree is also straightforward. 
For the last condition
assume to the contrary that  there exist $u$ and a downward-${T}_u$-path $(a_i)_{i=0}^N$
in $\str{A}'$  with rank bigger than $|A|$. Then there are indices $i_0,\ldots,i_{|A|}$ such that 
$\str{A}\models{T}_ua_{i_j}a_{i_j+1}\wedge\neg{T}_ua_{i_j+1}a_{i_j}$. Since  $\fh$ preserves the connections
between elements and their witnesses
we have $\str{A}_0\models{T}_u\fh(a_{i_j})\fh(a_{i_j+1})\wedge\neg{T}_u\fh(a_{i_j+1})\fh(a_{i_j})$. By the pigeonhole principle there exist $x<x'$ such that $\fh(a_{i_x})=\fh(a_{i_{x'}})$. This gives, by transitivity of ${T}_u$, that $\str{A}_0\models{T}_u\fh(a_{i_x+1})\fh(a_{i_x})$. Contradiction. 
\end{proof}

\subsection{Proof of Lemma \ref{l:llocglobl}}
\begin{proof}
(i)
Assume to the contrary that the given system is locally consistent but not globally consistent. This means that 
for some node $a$, for some $\alpha \in \AAA$, and some $\cT$ we have that $\alpha \in \fd_a(\cT)$ but
$\alpha \not\in \ldec{\str{A}}{a}(\cT)$, that is there is a node $b$, of atomic type $\alpha$ such that $\str{A} \models \cT ab$.
Thus, there exists a sequence of distinct nodes
$a=a_0, a_1, \ldots, a_N=b$ such that $a_i$ is either a child or the parent of $a_{i+1}$
and $\str{A} \models \cT a_i a_{i+1}$. 
Observe that it must be $a \not= b$ since otherwise Condition (ld2) would not be satisfied at $a$. 
By induction, using Condition (ld1), we can show that $\alpha \in \fd_{a_i}(\cT)$ for all
$i$, in particular for $i=N-1$. We now get a contradiction with (ld2) at $a_{N-1}$ or $a_N$ (depending on
which of them is the parent of the other). Part (ii) is straightforward.
\end{proof}

\subsection{Proof of Lemma \ref{l:ltreepruning}}

\begin{proof}
(i) Follows from the fact that for every $a \in A'$ its  downward family in $\str{A}'$ is an isomorphic copy of the downward family of $a$ in $\str{A}$. Moreover, due to the requirement that a subtree with root $a$ is replaced by a subtree with the root of the same generalized type as the type of $a$, this copy also preserves declarations.
(ii)  Consider now any pair of elements $a, a' \in A'$. Assume that
they have $1$-types, resp., $\alpha$ and $\alpha'$. 
 If $a$ is a child of $a'$ or $a'$ is a child of $a$ then the edge that joins them is an isomorphic copy of an
edge from $\str{A}$. Otherwise, due to the definition of light tree-like structure, $a$ and $a'$ may be joined only by some transitive 
relations. Let $\cT$ be the set of all transitive relations ${T}_u$ such that $\str{A}' \models {T}_u aa'$. By part (i) of 
this lemma the system of declarations is locally consistent on $\str{A}'$. By Lemma \ref{l:llocglobl} it is also globally consistent. In particular
$\fd(a)(\cT) \subseteq \ldec{\str{A}'}{a}(\cT)$. Since $\alpha' \not\in \ldec{\str{A}'}{a}(\cT)$ it follows that 
$\alpha' \not\in \fd(a)(\cT)=\ldec{\str{A}}{a}(\cT)$. Thus, there is a realization $b$ of $\alpha'$ in $\str{A}$ such
that $\str{A} \models \cT ab$, and hence the function mapping $a$ to itself and $a'$ to $b$ 
is the required homomorphism. 
(iii) 
To see that all nodes of $\str{A}'$ have the required witnesses again just note that for every $a \in A'$ its  downward family in $\str{A}'$
is an isomorphic copy of the downward family of $a$ in $\str{A}$. That $\str{A}' \models \phi$ follows now from part (ii) of this lemma and from Lemma \ref{l:homomorphisms}.
\end{proof}

\subsection{The pruning strategy} \label{s:strategy}

\smallskip\noindent

Let $\str{A}$ be a light-$\phi$-tree-like unravelling of a finite model of $\phi$. Let $(\fd_a)_{a \in A}$, $\fd_a := \ldec{\str{A}}{a}$ be the canonical
system of light declarations on $\str{A}$. 
During a top-down pruning process we define a function $\fs$ assigning to the surviving nodes a permutation 
of the set $\{1,\ldots,2k\}$.
In each $\str{A}_i$ this function is partial and defined for all nodes of depth at most $i$
(and is not modified in the subsequent structures for these nodes). Its values can be then transfered to $\str{A}'$, where it becomes total.
The purpose of $\fs$ is to define some order of shortening  
paths at a given node. Intuitively,  for $\fs(a)=\tau$, if  $v<v'$ then we  prefer to shorten ${T}_{\tau(v)}$ over ${T}_{\tau(v')}$. 

In $\str{A}_0$, let $\fs$ assigns an arbitrary permutation to the root. Assume that we have constructed $\str{A}_i$, for $i \ge 0$, and we
have assigned the values of $\fs$ to all its nodes of depth at most $i$.
Consider a node $a$ of $\str{A}_i$ of depth $i+1$. Denote its  parent by $a'$.
 Our task is to choose a descendant $b$ of $a$ whose subtree will replace the subtree of $a$ (or decide that this subtree is left untouched).
To make our choice we will look at permutation $\tau=\fs(a')$ assigned to the parent of $a$ and at three sets of indices, $K$, $S$, $D$,  whose definition depends on the connection between $a'$ and $a$, 
as follows. 

We say that ${T}_u$ is (i) \emph{killed at} $a$ (or: \emph{at the edge} $(a',a)$)
if $\str{A} \models\neg{T}_{\tau(v)}a' a$,
(ii) \emph{sustained at}  $a$ if $\str{A}\models{T}_{u}a'a\wedge{T}_{u}aa'$, and
(iii) \emph{diminished at}  $a$ if $\str{A}\models{T}_{u}a'a\wedge\neg{T}_{u}aa'$.
Let $K=\{v: {T}_{\tau(v)} \text{ is killed at $a$}\}$, 
$S:=\{v: {T}_{\tau(v)} \text{ is sustained at $a$}\}$,
$D:=\{v: {T}_{\tau(v)} \text{ is diminished at $a$}\}$. Note that the above sets contain not the numbers of
transitive relations but rather their positions in the permutation $\tau$.

If $D=\emptyset$ then we just choose $b$ to be $a$, that is we decide to leave the subtree of $a$ as it is.
If $D\neq\emptyset$ then let $v_D=\min D$ and choose  
$b$ to be a node of $ A_{a}$ such that (i) 
$\gtype{\str{A}}{a}=\gtype{\str{A}}{b}$
(the standard requirement in the pruning process),
 (ii) for all $v<v_D$, $v\in S$: $\rank{\str{A}}{\tau(v)}{b}  \le \rank{\str{A}}{\tau(v)}{a}$ (where the ranks may be equivalently
computed in $\str{A}_i$), and (iii) $\rank{\str{A}}{\tau(v_D)}{b}$ is the lowest possible. 
Note that such an element exists (however, it may happen that $b=a$) and $\rank{\str{A}}{\tau(v_D)}{b} \le \rank{\str{A}}{\tau(v_D)}{a} <\rank{\str{A}}{\tau(v_D)}{a'}$.

It remains to define $\fs(b)$.
If $K\not=\emptyset$ then let $v_K=\min K$
 and set $\fs(b):=\tau \circ (v_K, v_{K}+1, \ldots, 2k-1, 2k)$, where $\circ$ denotes permutation composition, and the second argument
is a cyclic permutation. 
In other words, we cyclically move the elements on positions $v_K, \ldots, v_{2k}$ in $\tau$ by one position to the left. This way the 
relation with the biggest priority among the relations that are killed in the current step  now gets the lowest priority. 
If $K = \emptyset$ then set  $\fs(b):=\tau$.

\begin{lemma} \label{l:lshortpaths} Let  $\str{A} \models \phi$ be a light tree-like structure over  signature $\sigma$, with transitive paths bounded by some natural number $M$. Then the result $\str{A}'$ of any top-down pruning process respecting our pruning strategy is a light tree-like structure with transitive paths bounded doubly exponentially in $|\sigma|$.
\end{lemma}

\begin{proof} 
Let $M_\phi$ be the number of the generalized types realized in $\str{A}$ increased by $2$.
Clearly, $M_\phi$ is bounded doubly exponentially in $|\phi|$. 
Let us first make an auxiliary estimation.
\begin{claim} \label{c:lshortpaths}
Let $v_0$ and a downward-${T}_u$-path $\vectorize{a}=(a_i)_{i=1}^N$ in $\str{A}'$ be such that for all $i,i'$ and $v\leq v_0$ we have $\fs(a_i)(v)=\fs(a_{i'})(v)$ (in this case, slightly abusing  notation, we write $\fs(\vectorize{a})(v)=\fs(a_i)(v)$) and let $u:=\fs(\vectorize{a})(v_0)$.
Let $Sum=\sum_{v< v_0}
\rank{\str{A}'}{\fs(\vectorize{a})(v)}{\vectorize{a}}$. 
Then $\rank{\str{A}'}{u}{\vectorize{a}} \leq M_\phi \cdot Sum + Sum + M_\phi$. 
\end{claim}
\begin{proof}
Consider first the case $v_0=1$ and take a downward-${T}_u$-path $\vectorize{a}=(a_i)_{i=1}^N$ in $\str{A}'$ such that 
$\fs(\vectorize{a})(1)=u$. $Sum=0$ in this case, so we need to show that $\rank{\str{A}'}{u}{\vectorize{a}} \leq M_\phi$. 
Observe first that the ${T}_u$-rank of elements, computed in $\str{A}$, is non-increasing along $\vectorize{a}$.
More precisely, for $1 \le i < N$, if $\str{A}' \models T_ua_i a_{i+1} \wedge \neg T_u a_{i+1} a_i$  
then  $\rank{\str{A}}{u}{a_i} > \rank{\str{A}}{u}{a_{i+1}}$ and if $\str{A}' \models T_ua_i a_{i+1} \wedge T_u a_{i+1} a_i$  
then  $\rank{\str{A}}{u}{a_i} \ge \rank{\str{A}}{u}{a_{i+1}}$. Both properties follow from our strategy: the former from
condition (iii) (note that in this case $v_D=1$) and the latter from condition (ii).
Assume now to the contrary
that $\rank{\str{A}'}{u}{\vectorize{a}} >  M_\phi$. This means that there are at least $M_\phi+1$ elements $a_{i+1}$ such 
that $\str{A}' \models T_ua_i a_{i+1} \wedge \neg T_u a_{i+1} a_i$. Thus, by the pigeonhole principle, there are at least two such elements,
$a_{x+1}$ and $a_{x'+1}$, say $x < {x'}$, having the same
generalized types in $\str{A}$. By the observation above, $\rank{\str{A}}{u}{a_{x+1}} > \rank{\str{A}}{u}{a_{x'+1}}$. But then, 
condition (iii) of our strategy requires us to use $a_{x'+1}$ instead of $a_{x+1}$ when looking for a child of $a_x$. Contradiction.

We now show that the Claim is true for arbitrary $2 \le  v_0 \le 2k$.   
Take a downward-${T}_u$-path $\vectorize{a}=(a_i)_{i=1}^N$ in $\str{A}'$ such that 
$\fs(a_i)(v)$ is constant on $\vectorize{a}$ for all $v \le v_0$ and $\fs(\vectorize{a})(v_0)=u$.
Note that none of $1, \ldots, v_0$ belongs to any of the sets $K$ computed during the construction of $\vectorize{a}$ 
(since if $v \in K$ then $\fs(a_i)(v)$ changes). Thus, relations
${T}_{\fs(\vectorize{a})(1)}$, $\ldots$, ${T}_{\fs(\vectorize{a})(v_0)}(={T}_u)$ are either diminished or sustained along $\vectorize{a}$.
Assume to the contrary that  $\rank{\str{A}'}{u}{\vectorize{a}}>M_\phi \cdot Sum + Sum + M_\phi$. Consider the edges of $\vectorize{a}$ such that ${T}_u$ is diminished on them. The number of such edges on
which additionally some of ${T}_{\fs{(\vectorize{a})(v)}}$ for $v<v_0$ is diminished is bounded by $Sum$ (by the definition of ranks). Thus at more than $M_\phi \cdot Sum + M_\phi$ 
edges we chose $v_D=v_0$ along the considered path. Let $\cal{Q}$ be the set of such edges. 

We now divide $\vectorize{a}$ into fragments containing $M_\phi$ edges from $\cal{Q}$ (a suffix of $\vectorize{a}$ with less then $M_\phi$
edges may be left). There are at least $Sum +1$ such fragments.
It follows, by the pigeonhole principle, that in at least one of them, call it $\vectorize{a}^*$, all of the ${T}_{\fs{(\vectorize{a})(v)}}$,
for $v<v_0$ are sustained. By arguments similar to those given in the case $v_0=1$ we see that
the ranks $\rank{\str{A}}{\fs(\vectorize{a})(v)}{a}$ are non-increasing along $\vectorize{a}^*$ for $v \le v_0$, and
$\rank{\str{A}}{\fs(\vectorize{a})(v_0)}{a}$ decreases at least $M_\phi$ times. The latter happens, again by the pigeonhole principle,
at least two times for edges leading to elements with the same generalized types in $\str{A}$, so, as in the case of $v_0=1$, we get
a contradiction with our strategy.
\end{proof}

The above claim allows us in particular to compute recursively
a (uniform) doubly exponential bound  on 
$\rank{\str{A}'}{u}{\vectorize{a}}$ for all $v_0$, $u$ and $\vectorize{a}$ as in assumption. Denote this bound by $\overline{M}_\phi$.

Consider now any downward-${T}_u$-path $\vectorize{a}=(a_i)_{i=1}^N$ in $\str{A}'$.
For each node $a_i$ from $\vectorize{a}$ let $v_{u}(a_i)$ be such
that $\fs(a_i)(v_u)=u$. Due to the strategy that we use to define $\fs$ the value of $v_u$ is non-increasing along $\vectorize{a}$. Indeed, when moving from $a_i$ to $a_i+1$ the value of $v_u$ is either unchanged or decreases by $1$; the only chance of increasing it would be 
to change it to $2k$ but this happens only when ${T}_u$ is killed. Let us divide $\vectorize{a}$ into fragments $\vectorize{a}_1, \vectorize{a}_2, \ldots$
on which $v_u$ is constant. The number of such fragments  is obviously bounded by $2k$. On each of such fragments
$\vectorize{a}_i$ for all $v \le  v_u$ we have that $\fs(\vectorize{a}_i)(v)$ is constant. So we can apply Claim \ref{c:lshortpaths} to bound 
$\rank{\str{A}'}{u}{\vectorize{a}_i}$ 
by $\overline{M}_\phi$. This gives the desired doubly exponential bound $\hat{M}_\phi=2k \overline{M}_\phi$ on $\rank{\str{A}'}{u}{\bar{a}}$
and finishes the proof of Lemma \ref{l:lshortpaths}. 
\end{proof}

Our strategy, together with Lemma \ref{l:ltreepruning} gives Lemma \ref{l:lshortpaths2}.

\subsection{Proof of Claim \ref{c:lnodanger}}

\begin{proof}  
	(i) 
	Take any $b' \in B_0$. Note that $\fp \restr B_0$ goes into $A_{\fp({b})} \cap [\fp(b)]_{\mathcal{E}_{tot} \cup \{{T}_u, {T}_{u}^{-1}\}}$, so $\fp(b')$ belongs
	to the subtree of $\fp(b)$ and is connected to it by by both ${T}_u$ and ${T}_u^{-1}$. It follows that extending any downward-${T}_u$-path starting at $\fp(b')$ by
	the path from $\fp(b)$ to $\fp(b')$ does not change its ${T}_u$-rank. Hence $\rank{\str{A}}{u}{\fp(b))\geq\rank{\str{A}}{u}{\fp(b'}}$.
	(ii) By our strategy of choosing witnesses and assigning layers to them we know that $\fp(c')$ is a child of $\fp(b)$ and $\str{A} \models {T}_u \fp(b) \fp(c') \wedge
	\neg {T}_u \fp(c') \fp(b)$. Thus, extending any downward-${T}_u$-path starting at
	$\fp(c')$ by the edge from $\fp(b)$ to $\fp(c')$ increases its rank by $1$. The claim thus follows.
\end{proof}

\subsection{Correctness of the construction in the proof of Lemma \ref{l:lfin}}\label{s:cor}

Recall Fig.~\ref{f:cylinder0}. We naturally divide $\str{A}_0^0$ and $\str{A}_0'$ into $4l$ \emph{levels}. For $g=0,1$ and $1 \le i \le 2l$, level $2lg+i$ is the union of layers $L_i$ of all 
components of color $g$.

%1: 
\noindent
(b\ref{blone}) 
Assume that $\str{A}_{a_0} \restr [a_0]_{\mathcal{E}_{tot}}$ is not a singleton
(in this case also $\str{A}_0'$ is not a singleton). 
Thus each of the ${T}_u \in \mathcal{E}_{tot}$ is total on it, in particular it is reflexive. By the inductive assumption it is total on subcomponents (for singleton subcomponents it follows from (b\ref{blfour}), using the fact that $\fh$ preserves the $1$-types, and thus, in particular, reflexivity of the ${T}_u)$. When components are formed out of subcomponents we always use $2$-types from $[a_0]_{\mathcal{E}_{tot}}$. It is thus straightforward that, after taking
a connected fragment of the graph of components and applying the transitive closures to get $\str{A}_0'$ all relations from $\mathcal{E}_{tot}$ become
total.

%2: 
\noindent
(b\ref{bltwo}) As $a_0'$ we take the root of $\str{C}^{\gamma_{a_0},0}_{\bot,\bot}$. Recall that we explicitly map
the root of the pattern component $\str{C}^{\gamma_{a_0}}$  by $\fp$ to $a_0$.

\noindent
(b\ref{blthree})
If we prove that an element and its pattern have the same 1-types, then the existence of witnesses is easy to show. Indeed,
we explicitly take care of this when building components in Step 2 (\emph{Providing witnesses}). In each component, every
element from  layer $L_i^j$ has its witnesses in $L_i^j \cup L_i^{j+1, init} \cup L_{i+1}^{1, init}$. 
Every interface element is identified with the root of some other component so it also has its witnesses.
So, the only potential danger is that some $1$-types are enlarged. While we initially explicitly copy the $1$-types from the original model,
it is probably not completely obvious that they remain the same after taking the transitive closures: the potential danger is that
we may possibly form a ${T}_u$-cycle from an element $a'$, such that  $\str{A} \models \neg {T}_u \fp(a')\fp(a')$. 

To see that this cannot happen, as well as to prepare ourselves for a proof of (b\ref{blfour}) 
we now spend a while on understanding transitive paths in $\str{A}_0^0$ and thus transitive connections in $\str{A}_0'$.
First, 
 observe that each 2-type in $\str{A}_0^0$ is either a copy of a 2-type between an element and its witness (Step 2), was set when putting a subcomponent (Step 1), or is trivial, that is, it makes true only some unary atoms.
We say that a sequence $a_1,\ldots,a_N\in A_0^0$ satisfying $a_i\neq a_{i+1}$ for all $1\leq i\leq N-1$ is a \emph{path} in $\str{A}_0^0$ if for all $1\leq i\leq N-1$ the elements $a_i,a_{i+1}$ either belong to the same subcomponent or one is put as a witness for the other, and a ${T}_u$-\emph{path} 
if for all $1\leq i\leq N-1$ we have $\str{A}_0^0\models{T}_ua_ia_{i+1}$.
Observe that in particular every ${T}_u$-path in $\str{A}_0^0$ is a path and every path in $\str{A}_0^0$ is automatically a ${T}_u$-path for all $u>2l$.

Consider a pair of elements $a',b'$ (possibly $a'=b'$). We are interested in  the $2$-type of $(a',b')$ in $\str{A}_0'$, in particular
 in the ${T}_u$-paths joining $a', b'$ in $\str{A}_0^0$. 
Assuming that this $2$-type contains some binary symbol other than
the symbols from $\mathcal{E}_{tot}$, we argue, that it is identical to
a $2$-type of some pair in some simplifications of $\str{A}_0'$. (The case where $a', b'$ are joined only by relations from $\mathcal{E}_{tot}$ is simple
and will be treated separately.)

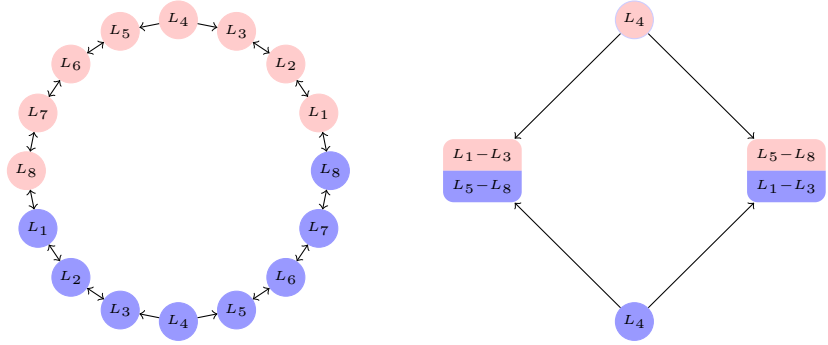
\begin{figure}
\begin{center}
\begin{tikzpicture}[scale=1]

\begin{scope}
\foreach \a in {1,2,...,8}{
\draw ({360/16*(\a)}: 2) node[circle, draw=red, minimum width=0.5cm, opacity=0.2, fill=red, label=center:$\scriptscriptstyle L_{\a}$] (A\a) {};
}

\foreach \a in {1,2,...,8}{
\draw ({360/16*(\a+8)}: 2) node[circle, draw=blue, minimum width=0.5cm, opacity=0.4, fill=blue, label=center:$\scriptscriptstyle L_{\a}$] (B\a) {};
}

\draw (A1) edge[<->] (A2);
\draw (A2) edge[<->] (A3);
\draw (A3) edge[<-] (A4);
\draw (A4) edge[->] (A5);
\draw (A5) edge[<->] (A6);
\draw (A6) edge[<->] (A7);
\draw (A7) edge[<->] (A8);
\draw (A8) edge[<->] (B1);
\draw (B1) edge[<->] (B2);
\draw (B2) edge[<->] (B3);
\draw (B3) edge[<-] (B4);
\draw (B4) edge[->] (B5);
\draw (B5) edge[<->] (B6);
\draw (B6) edge[<->] (B7);
\draw (B7) edge[<->] (B8);
\draw (B8) edge[<->] (A1);
\end{scope}

\begin{scope}[shift={(6,0)}]
\draw ({360/16*(4)}: 2) node[circle, draw=blue, minimum width=0.5cm, opacity=0.2, fill=red, label=center:$\scriptscriptstyle L_{4}$] (A24) {};
\draw ({360/16*(4+8)}: 2) node[circle, draw=blue, minimum width=0.5cm, opacity=0.4, fill=blue, label=center:$\scriptscriptstyle L_{4}$] (B24) {};
%}
\draw ({360/16*(8)}: 2) node[rectangle, rounded corners, minimum width=1cm, minimum height=1.5cm,  rectangle split, rectangle split part fill={red!20,blue!40},
 rectangle split parts=2] (A28) {$\scriptscriptstyle {L_1-L_3}$
\nodepart{second}
  {$\scriptscriptstyle {L_5-L_8}$}
};

\draw ({360/16*(16)}: 2) node[rectangle, rounded corners, minimum width=1cm, minimum height=1.5cm,  rectangle split, rectangle split part fill={red!20,blue!40},
 rectangle split parts=2] (B28) {$\scriptscriptstyle {L_5-L_8}$
\nodepart{second}
  {$\scriptscriptstyle {L_1-L_3}$}
};

\draw (A24) edge[->] (A28);
\draw (A24) edge[->] (B28);
\draw (B24) edge[->] (A28);
\draw (B24) edge[->] (B28);

\end{scope}
\end{tikzpicture}
\end{center}
\caption{${T}_3$ and ${T}_4$-connections (recall: ${T}_4={T}_3^{-1}$) between levels of $\str{A}_0'$.  The absence of an arrow from a level (group of levels) $A$ to level (group of levels) $B$ means that there are no ${T}_3$-connections from elements of  $A$ to elements of $B$. Groups of levels in the right figure contain $2l-1$ levels each.
} \label{f:lpaths}
\end{figure}

\smallskip\noindent
\emph{Reduction 1}. Note that for any ${T}_u$-path, $1 \le u \le 2l$, there is a set of at most $2l$ consecutive levels
of our cylindrical structure in which all elements of this path are contained.
Consider, e.g., the case of a ${T}_3$-path, see Fig.~\ref{f:lpaths}.
This path must be contained either in the union of 
red levels $L_1$--$L_3$, blue levels $L_5$--$L_{2l}$, and at most one (red or blue) level $L_4$, or, symmetrically, 
in the union of blue levels $L_1$--$L_3$, red levels $L_5$--$L_{2l}$, and at most one level $L_4$. Analogously for the other ${T}_u$.
This becomes evident when looking at the graph from the right part of Fig.~\ref{f:lpaths} whose nodes are strongly connected
components (consisting of at most $2l-1$ levels) of the graph from the left part of this figure.

Assume that $a'$ and $b'$ are connected by a
${T}_u$-path, for some $1 \le u \le 2l$ or by a non-transitive connection crossing one of the borders between colors, that is using an edge between a leaf of color $g$ and a root of 
color $1-g$, for some $g=0,1$. 
Then for any $1 \le v \le 2l$, including $v=u$,
any ${T}_v$-path
from $a'$ to $b'$ cannot cross the other border since the minimal set of consecutive levels containing the levels of $a'$, $b'$ and the levels
adjacent to that other border would have cardinality greater than $2l$. 
Thus we can cut all the connections
between leaves of color
$1-g$ and roots of color $g$ (that is, make any atom containing a pair of such elements false).
See Fig.~\ref{f:cylinder}, where this step is illustrated for $g=1$. Let $\str{D}_0^0$ be the structure so-obtained and $\str{D}_0'$ its transitive closure, call it \emph{a 
transection of} $\str{A}_0'$. 
By the discussion above, the inclusion map $\iota:\str{A}_0'\restr\{a,b\}\to\str{D}_0'$  
is a homomorphism 

\smallskip\noindent
\emph{Reduction 2.} Take $g$ from the previous reduction. By our strategy of joining the components,
there exists a generalized type $\gamma$ such that any ${T}_u$-path joining  $a'$ and $b'$ is contained in components of 
the forms $\str{C}^{g,\gamma}_{\cdotp,\cdotp}$ and $\str{C}^{1-g,\cdotp}_{\cdotp,\gamma}$. 
Denote $\str{E}_0^0$ the restriction of $\str{D}_0^0$ to these components  and $\str{E}_0'$ its transitive closure. 
Choose any component $\str{C}^{\fp}$ of the form $\str{C}^{g,\gamma}_{\cdotp,\cdotp}$. Recall that all of them are isomorphic copies of
the pattern component $\str{C}^{\gamma}$. Now, define another auxiliary structure $\str{F}_0^0$ obtained by restricting $\str{E}_0^0$ to the union of $C^{\fp}$ and the domains of the components of the form $\str{C}^{1-g,\cdotp}_{\cdotp,\gamma}$. Let $\str{F}_0'$ be its transitive closure. There is a natural projection $\pi:E_0^0\to\str{F}_0^0$, which maps the elements of the components of the form $\str{C}^{g,\gamma}_{\cdotp,\cdotp}$ into the corresponding elements of component $\str{C}^{\fp}$ (being the identity on the other elements). Observe that $\pi:\str{E}_0^0\to\str{F}_0^0$ is a homomorphism, and as we can apply it to transitive paths, we obtain that also $\pi:\str{E}_0'\to\str{F}_0'$ is a homomorphism. See Fig.~\ref{f:lred2}.

$\str{F}_0^0$ looks like a single component but is twice as high. It can be viewed as a tree $\tau$ defined by taking the subcomponents used to build the components of $\str{F}_0^0$ (Step 1) as the nodes and connecting two of them iff one of them contains a witness for some element of the other (Step 2). Note that this way the
leaves of component $\str{C}^\fp$ are parents of some roots of the remaining components of $\str{F}_0^0$.

Now we are ready to come back to the proof of (b\ref{blthree}). What remains is to show that for every $a'\in A_0'$ we have $\type{\str{A}_0'}{a'}=\type{\str{A}_0}{\fp(a')}$. Recall that 
the only possible reason for not being so is that the 1-type of $a'$ was enlarged by the application of the transitive closure to $\str{A}_0^0$. That is, there exists, for some $1 \le u \le 2l$, a ${T}_u$-path that connects $a'$ with itself in $\str{A}_0^0$. 
Taking $b'=a'$, we use Reduction 1 and 2 (if $a'$ is an element of a component of color $g$, then we choose $\str{C}^{\fp}$ to be the
component of $a'$), and after an application of $\pi$, we have a  ${T}_u$-path connecting $a$ with itself in $\str{F}_0^0$.
Due to the tree shape of $\str{F}_0^0$ either there is such a path in the subcomponent of $a'$ or in the $2$-element substructure joining $a'$ 
with one of its witnesses, denote it $w'$.
It follows that ${T}_uxx\in\type{\str{A}_0}{\fp(a')}$. Indeed, in the former case we get it by the fact that the subcomponents
are closed transitively and by the  1-type assumption in (b\ref{blthree}) for the subcomponent of $a'$. In the latter case we just recall that
the $2$-type for the pair $a', w'$ was copied 
to $\str{A}_0^0$
from the original structure, which was transitively closed.

\begin{figure} 
\begin{center}
\begin{tikzpicture}

\begin{scope}[scale=1.4]

\node at (-4,1.4) {$\str{A}'_0$};

%kontury 
  \draw[red]  (0,-1.25)--(-3.5,-1.25);
	\draw[blue]  (0,1.25)--(-3.5,1.25);
  
	\draw[blue] (-3.5,1.25) arc(90:165:0.5 and 1.25);
	\draw[blue] (0,1.25) arc(90:165:0.5 and 1.25);
	\draw[blue, postaction={decorate, decoration={text along path, raise=3pt, text align={align=center}, text color=blue, text={${\scriptscriptstyle g=1}${} }}}] (0,1.25) arc(90:-15:0.5 and 1.25) node (bck) {}; 
	
	\path (0,1.25) arc(90:165:0.5 and 1.25) node (tmp) {} ;
	\draw[red, postaction={decorate, decoration={text along path, raise=3pt, text align={align=center}, text color=red, text={${\scriptscriptstyle g=0}${} }}}] (tmp.center) arc(165:345:0.5 and 1.25); 

	\path (-3.5,1.25) arc(90:165:0.5 and 1.25) node (tmp) {} ;
  \draw[red] (tmp.center) arc(165:270:0.5 and 1.25);

%niebieskie i czerwone tlo
	\fill [opacity=0.4, color = blue] (0,1.25) -- (-3.5,1.25) arc
  (90:165:0.5 and 1.25) -- ++(3.5,0)  arc(165:90:0.5 and 1.25) ;
	\fill [opacity=0.2, color = red] (0,-1.25) -- (-3.5,-1.25) arc
  (270:165:0.5 and 1.25) -- ++(3.5,0)  arc(165:270:0.5 and 1.25) ;
	
%linia wewnetrzna i nozyczki
\draw[dotted] (bck.center) -- +(-1,0);	
\node at (0.7,-0.3) (scissors) {\Leftscissors} ;

%poziome linie	

\foreach \t in {90,105,...,270}
{
\path (0,1.25) arc(90:\t:0.5 and 1.25) node (a) {} ;
\draw[very thin] (a.center) -- +(-3.5,0);
};

%sciezka
\node at (-2.7,-0.5) {$a$};
\node at (-1.45,0.75) {$b$};

\draw (-2.55,-0.5) circle (0.02);
\draw (-2.45,-0.15) circle (0.02);
\draw (-2.2,0.15) circle (0.02);
\draw (-1.25,0.45) circle (0.02);
\draw (-1.3,0.75) circle (0.02);

\draw[->] (-2.55,-0.5) -- (-2.45,-0.15);
\draw[<->] (-2.45,-0.15) -- (-2.2,0.15);
\draw[->] (-2.2,0.15) -- (-1.25,0.45);
\draw[->] (-1.25,0.45) -- (-1.3,0.75);

%

%komponenty
\draw[dashed] (-3.3,0.4) to[bend left] (-2,1.25);
\draw[dashed] (-3.3,0.4) to (-1.1,0.4);
\draw[dashed] (-1,0.4) to[bend right=10] (-1,1.25);

\draw[dotted] (-3.3,0.4) to[bend right=10] (-3.4,0.1);
\draw[dotted] (-3.4,0.1) to (-1.1,0.1);
\draw[dotted] (-1,0.4) to[bend left=10] (-1.1,0.1);

%\draw[dotted] (0.7,1.3) -- (0.5, 1.0) -- (3.0, 1.0) -- (2.8,1.3);

\draw[dashed] (-2.2,0.3) to[bend right] (-3.3,-1.25);
\draw[dashed] (-2.2,0.3) to[bend left=10] (-0.6,-1.25);

 \end{scope}
 %%%%%%%%%%%%%%%%%%%%%%%%%%%%%%%%%%%%%%%%%%%%%%%%%%%%%%%%%%%%%%%%%%%%%%%%%%%%%%%%%%%%%%%%%%
\begin{scope}[scale=1.4, shift={(1.5,-1.25)}]

\node at (-0.4, 2.7) {$\str{D}'_0$};

%kontury
\draw[red]  (0,-0.3)--(3.5,-0.3)--(3.5,1.25)  (0,1.25) -- (0,-0.3);
\draw[blue] (0,1.25) -- (0,2.8) -- (3.5, 2.8) -- (3.5,1.25); 

%tlo
\fill[opacity=0.2, color = red] (0,-0.3)--(3.5,-0.3)--(3.5,1.25) -- (0,1.25) -- (0,-0.3);
\fill[opacity=0.4, color = blue] (0,1.25) -- (0,2.8) -- (3.5, 2.8) -- (3.5,1.25);

%poziome kreski
\foreach \t in {0.35, 0.65, ..., 2.15}
{
\draw[very thin] (0,\t) -- (3.5,\t);
}

\draw[very thin] (0,0) -- (3.5,0);
\draw[very thin] (0,2.5) -- (3.5,2.5);

%numery layerow
\node[scale=0.8] at (3.35,1.1) {$L_{1}$};
\node[scale=0.8] at (3.35,0.8) {$L_{2}$};
\node[scale=0.8] at (3.35,0.5) {$L_{3}$};
\node[scale=0.8] at (3.35,-0.15) {$L_{2l}$};

\node[scale=0.8] at (3.35,1.4) {$L_{2l}$};
\node[scale=0.8] at (3.35,2.65) {$L_{1}$};

%trzykropki

\node (tk) at (1.75,2.4) {$\vdots$};
\node (tk) at (1.75,0.25) {$\vdots$};

%punkty sciezki

\node at (1.45,0.5) {$a$};
\node at (2.05,1.7) {$b$};

\draw (1.6,0.5) circle (0.02);
\draw (1.65,0.8) circle (0.02);
\draw (1.75,1.1) circle (0.02);
\draw (2.3,1.4) circle (0.02);
\draw (2.2,1.7) circle (0.02);

%strzalki sciezki
\draw[->] (1.6,0.5) -- (1.65,0.8);
\draw[<->] (1.65,0.8) -- (1.75,1.1); 
\draw[->] (1.75,1.1) -- (2.3,1.4);
\draw[->] (2.3,1.4) -- (2.2,1.7);

%komponenty
\draw[dashed] (0.7,1.3) -- (1.75,2.7) -- (2.8,1.3) -- (0.7,1.3);
\draw[dotted] (0.7,1.3) -- (0.5, 1.0) -- (3.0, 1.0) -- (2.8,1.3);
\draw[dashed] (0.7,-0.2) -- (1.75,1.2) -- (2.8,-0.2) -- (0.7,-0.2);

\end{scope}
\end{tikzpicture}
\caption{Reduction 1: From cylindrical structure $\str{A}'_0$ to its transection  $\str{D}'_0$.}
\label{f:cylinder}
\end{center} 
\end{figure}
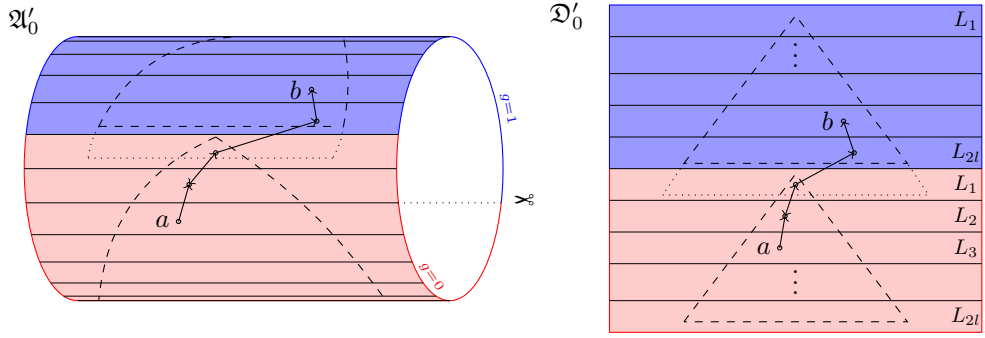

\begin{figure}  

\begin{tikzpicture}[scale=0.9]

\begin{scope}[scale=1.4]

%kontury
\draw[red]  (-0.1,-0.1)--(4.8,-0.1)--(4.8,1.4)  (-0.1,1.4) -- (-0.1,-0.1);
\draw[blue] (-0.1,1.4) -- (0.75,2.9) -- (4.8,2.9) -- (4.8,1.4);  

%tlo
\fill[opacity=0.2, color = red] (-0.1,-0.1)--(4.8,-0.1)--(4.8,1.4)--(-0.1,1.4) -- (-0.1,-0.1);
\fill[opacity=0.4, color = blue]  (-0.1,1.4) -- (0.75,2.9) -- (4.8,2.9) -- (4.8,1.4);

\foreach \x/\y in {0/0, 1.6/0, 3.2/0, 0.8/1.5, 2.4/1.5}{
\draw (\x,\y) -- +(1.5,0) -- +(0.75,1.3) -- (\x, \y);

\draw[dotted] (4.8,2.7) -- (4.75,2.8) -- (4.0,1.5) -- (4.8,1.5);
\draw[dotted] (0,1.5) -- (0.7,1.5) -- (0.33,2.15);
}

%punkty
\foreach \x/\y in {0.65/0.8, 0.75/1.1, 
                   1/1.6, 1.2/1.9, 1.4/1.9, 1.6/1.6, 1.8/1.9, 2.0/1.6,
									 2.35/1.1,
									 3.2/1.6, 3.4/1.9, 3.6/1.6, 3.0/1.9, 2.8/1.9, 2.6/1.6,
									 3.95/1.1, 4.05/0.8, 3.95/0.5}{
\draw (\x,\y) circle (0.02); %
}

%strzalki
\draw[->] (0.65,0.8) -- (0.75,1.1);
\draw[->] (0.75,1.1) -- (1,1.6);
\draw[->] (1,1.6) -- (1.2,1.9);
\draw[->] (1.2,1.9) -- (1.4,1.9);
\draw[->] (1.4,1.9) -- (1.6,1.6);
\draw[<->] (1.6,1.6) -- (2.35,1.1);
\draw[<->] (2.35,1.1) -- (3.2,1.6);
\draw[->] (3.2,1.6) -- (3.4,1.9);
\draw[->] (3.4,1.9) -- (3.6,1.6);
\draw[->] (3.6,1.6) -- (3.95,1.1);
\draw[->] (3.95,1.1) -- (4.05,0.8);
\draw[->] (4.05,0.8) -- (3.95,0.5);

\draw[->, dotted] (1.6,1.6) -- (1.8,1.9);
\draw[->, dotted] (1.8,1.9) -- (2.0,1.6);
\draw[->, dotted] (2.6,1.6) -- (2.8,1.9);
\draw[->, dotted] (2.8,1.9) -- (3.0,1.9);
\draw[->, dotted] (3.0,1.9) -- (3.2,1.6);

%izomorfizmy

\node at (2.35,2) {$\cong$};

% podpisy

\node at (0.1,2.7) {$\str{D}'_0 {\rightarrow} \str{E}'_0$};

\node at (0.8,0.8) {$a$};
\node at (4.1,0.5) {$b$};

\end{scope}
%%%%%%%%%%%%%%%%%%%%%%%%%%%%%%%%%%%%%%%%%%%%%%%%%%%%%%%%%%%%%%%%%%%%%%%%%%%%%%%%%%%%%%

\begin{scope}[scale=1.4, shift={(5.5,0)}]

%kontury
\draw[red]  (-0.1,-0.1)--(4.8,-0.1)--(4.8,1.4)  (-0.1,1.4) -- (-0.1,-0.1);
\draw[blue] (-0.1,1.4) -- (0.75,2.9) -- (4.8,2.9) -- (4.8,1.4);  

%tlo
\fill[opacity=0.2, color = red] (-0.1,-0.1)--(4.8,-0.1)--(4.8,1.4)--(-0.1,1.4) -- (-0.1,-0.1);
\fill[opacity=0.4, color = blue]  (-0.1,1.4) -- (0.75,2.9) -- (4.8,2.9) -- (4.8,1.4);

\foreach \x/\y in {0/0, 1.6/0, 3.2/0, 0.8/1.5}{
\draw (\x,\y) -- +(1.5,0) -- +(0.75,1.3) -- (\x, \y);
}

\foreach \x/\y in {2.4/1.5}{
\draw[dotted] (\x,\y) -- +(1.5,0) -- +(0.75,1.3) -- (\x, \y);
}

%punkty
\foreach \x/\y in {0.65/0.8, 0.75/1.1, 
                   1/1.6, 1.2/1.9, 1.4/1.9, 1.6/1.6, 1.8/1.9, 2.0/1.6,
									 2.35/1.1,
									 3.2/1.6, 3.4/1.9, 3.6/1.6, 3.0/1.9, 2.8/1.9, 2.6/1.6,
									 3.95/1.1, 4.05/0.8, 3.95/0.5}{
\draw (\x,\y) circle (0.02); %
}

%strzalki
\draw[->] (0.65,0.8) -- (0.75,1.1);
\draw[->] (0.75,1.1) -- (1,1.6);
\draw[->] (1,1.6) -- (1.2,1.9);
\draw[->] (1.2,1.9) -- (1.4,1.9);
\draw[->] (1.4,1.9) -- (1.6,1.6);
\draw[<->] (1.6,1.6) -- (2.35,1.1);
\draw[<->, dotted] (2.35,1.1) -- (3.2,1.6);
\draw[->, dotted] (3.2,1.6) -- (3.4,1.9);
\draw[->, dotted] (3.4,1.9) -- (3.6,1.6);
\draw[->, dotted] (3.6,1.6) -- (3.95,1.1);
\draw[->] (3.95,1.1) -- (4.05,0.8);
\draw[->] (4.05,0.8) -- (3.95,0.5);

\draw[->] (1.6,1.6) -- (1.8,1.9);
\draw[->] (1.8,1.9) -- (2.0,1.6);
\draw[->, dotted] (2.6,1.6) -- (2.8,1.9);
\draw[->, dotted] (2.8,1.9) -- (3.0,1.9);
\draw[->, dotted] (3.0,1.9) -- (3.2,1.6);

\draw[->] (2.0,1.6) -- (3.95,1.1);

%izomorfizmy

\node at (2.35,2) {$\cong$};

%wygieta strzlka
\draw[bend right=10, ->,  dashed] (2.8,2.4) to (1.9,2.4);
\node at (2.35,2.6) {$\pi$};

% podpisy

\node at (1.55,2.5) {$\str{C}^{\fp}$};
\node at (0.1,2.7) {$\str{E}'_0{\rightarrow}\str{F}'_0$};

\node at (0.8,0.8) {$a$};
\node at (4.1,0.5) {$b$};

\end{scope}
\end{tikzpicture}
\caption{Reduction 2: from transection $\str{D}'_0$ to $\str{E}'_0$ and from $\str{E}'_0$  to a tree of components $\str{F}'_0$. }
\label{f:lred2}
\end{figure}
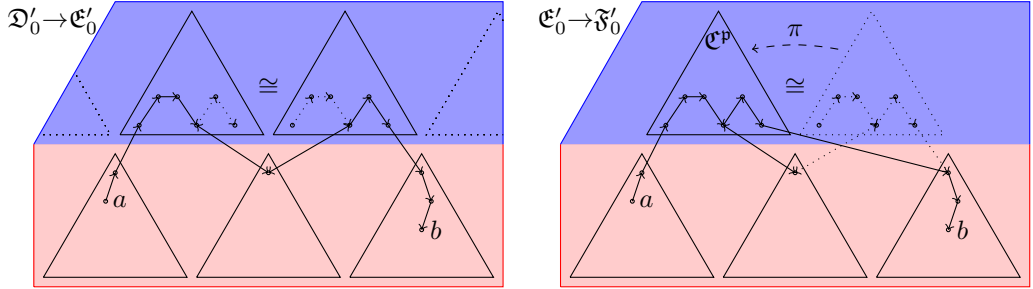

%%%%%%

(b4) %\emph{Reduction 0} 
Suppose that $a',b'\in A_0'$ are not connected by any relation except for those from $\mathcal{E}_{tot}$.
Then it suffices to take $\fh:=\fp\restr\{a',b'\}$, use (b3) for the 1-type preservation and recall that $\fp(A_0')\subseteq[a_0]_{\cE_{tot}}$. From this the
other required condition also follows.
From now we assume that $a',b'$ are connected by some binary relation not belonging to $\mathcal{E}_{tot}$.

Assume that $\str{A}_0'\models\mathcal{T}a'b'$, that is for any ${T}_u\in\mathcal{T}$ there exists a ${T}_u$-path 
in $\str{A}_0^0$ that joins $a'$ and $b'$. Observe that such paths may differ for different ${T}_u$ and, firstly, we want to find a common path for all the ${T}_u\in\mathcal{T}$.
As before, apply Reductions 1 and 2 to get that $\pi(a'),\pi(b')$ are connected in $\str{F}_0^0$ by ${T}_u$-paths for all ${T}_u\in\mathcal{T}$.
 If we prove the thesis for $\pi(a'),\pi(b')$ (in particular, looking for an appropriate partial 
homomorphism from $\str{F}_0$ into $\str{A}$), then it holds also for $a',b'$,
as $\pi$ is a homomorphism from $\str{A}_0$ to $\str{F}_0$,  $\fp\circ\pi=\fp$ and $\pi,\iota$ preserve the 1-types.

Now, using  standard tree reasoning and the fact that the subcomponents are transitively closed, one may observe that the shortest path (in $\str{F}_0^0$) connecting $\pi(a')$ with $\pi(b')$ is a ${T}_u$-path for any ${T}_u\in\mathcal{T}$. Denote this path $\pi(a')=a_1,\ldots,a_N=\pi(b')$. 

Take a 1-type $\alpha\not\in\ldec{\str{A}}{\fp(b')}$. We prove by induction on $j=N,N-1,\ldots,1$ that $\alpha\not\in\ldec{\str{A}}{\fp(a_j)}$. In particular we prove the second part of condition (b\ref{blfour}).
The thesis for $j=N$ clearly holds, since we explicitly assumed it.
Now consider $a_j$ and $a_{j+1}$. Then either they appear together in some subcomponent or a copy of a witness structure ($a_j$ is a witness for $a_{j+1}$ or vice versa). In the first case, use the condition (b4) for such subcomponent, which holds by the inductive assumption of Lemma \ref{l:lfin}. In the second, just use the fact that (ld1) holds for $\fp(a_j)$ and $\fp(a_{j+1})$.

Finally, we build the required homomorphism $\fh$.
If $\pi(a')=\pi(b')$, then it suffices to put $\fh(\pi(a'))=\fp(\pi(a'))$ and recall (b3).
If $\pi(a')\neq\pi(b')$ and they are connected by some not-transitive atom, then $N=2$ (that is, the path is $\pi(a')=a_1,a_2=\pi(b')$). Then if $\{\pi(a'),\pi(b')\}$ is a copy of a pair consisting of an element and its witness, then put $\fh:=\fp\restr\{\pi(a'),\pi(b')\}$; and if they appear in the same subcomponent, use a homomorphism guaranteed by condition (b\ref{blfour}) of the inductive assumption.
So now assume that $\pi(a')\neq\pi(b')$ and $\mathcal{T}=\{{T}_u:\str{F}_0'\models{T}_u\pi(a')\pi(b')\}$ is non-empty. First, observe that $\ldec{\str{A}}{\fp(a_{N-1})}\not\ni\alpha$, where $\alpha=\type{\str{\str{F}_0'}}{a_N}$. Indeed, if $a_{N-1},a_N$ are an element and its witness (or vice versa), then it suffices to use the fact that $\fp$ preserves the 1-types (that is, a part of (b3)). If $a_{N-1},a_N$ belong to the same subcomponent, then by the condition (b4) of the inductive assumption of
Lemma \ref{l:lfin} there exists in a homomorphism $\fh$, which gives us that $\ldec{\str{A}}{\fp(a_{N-1})}\not\ni\type{\str{A}}{\fh(a_N)}(=\alpha)$.
Thus, applying the induction just like one paragraph above, we get that $\ldec{\str{A}}{\fp(\pi(a'))}\not\ni\alpha$. By the maximality of such declarations, there exists $b\in\str{A}$ such that $\type{\str{A}}{b}=\alpha$ and $\str{A}\models\mathcal{T}\fp(\pi(a'))b$. Put $\fh(\pi(a')):=\fp(\pi(a')), \fh(\pi(b')):=b$.

\subsection{Size of models constructed in the proof of Lemma \ref{l:lfin}}\label{s:size}

\begin{claim}
The construction in the proof of Lemma \ref{l:lfin} produces models of size at most triply exponential in the size
of the given formula.
\end{claim}

\begin{proof}
The following routine estimation shows that $|A_0'|$ is triply exponential in $n=|\phi|$, regardless of the choice of the initial 
tree-like model $\str{A}$.
We 
calculate a bound $ S_{2l}$
on the size of the structure obtained in the proof of Lemma \ref{l:lfin} for $|\mathcal{E}_0|=2l$. We are interested in $ S_{2k+2}$, which is the desired bound on the size of $\str{A}_0'$ (we use $ S_{2k+2}$ here, rather than $S_{2k}$, because we may potentially
introduce the auxiliary identity relations in the base step of induction).
Recall that any pattern component may be viewed as a (rooted)
tree of subcomponents consisting of at most $2l(\hat{M}_{\phi}+1)$ sublayers,
which is automatically a bound on the height of this tree.
The root (the only vertex of depth $1$) corresponds to a single subcomponent, that
contains at most $ S_{2l-2}$ elements.
Each of these elements may then require at most $n$ witnesses, each of which is added to the initial part either of the next sublayer of the first layer or of the first sublayer of the next layer 
(in either case: to a subcomponent of depth $2$). 
Thus we get that we put at most $S_{2l-1}n$ elements, 
each of them then gives rise to a subcomponent (being itself its \origin), so we
have at most $S_{2l-2}^2n$ elements 
corresponding to the vertices of the tree of depth $2$.
Iterating, we have at most $ S_{2l-2}^in^{i-1}$ elements 
corresponding to the vertices of the tree of depth $i$,
which leads to an estimate $( S_{2l-2}n)^{2l(\hat{M}_{\phi}+1)+1}$ both on the number of elements
in a pattern component, and the number of interface elements in its extended version.
In the joining phase we thus use at most $2 \cdot |\GGG[A]|^2 \cdot ( S_{2l-2}n)^{2l(\hat{M}_{\phi}+1)+1}$ components.
Finally,
estimating $l$ in the exponents by $2n$, 
we get a bound 
$ S_{2l}=2|\GGG[A]|^2( S_{2l-2}n)^{8n(\hat{M}_{\phi}+1)+2}$.
Solving this recurrence relation, and recalling that $\hat{M}_{\phi}$ and $|\GGG[A]|$ are doubly exponential in $|\phi|$, we obtain a triply exponential bound on $ S_{2k+2}$. 
\end{proof}

\subsection{Proof of Theorem \ref{t:lalgo}}
\begin{proof}
The lower bound  can be 
obtained for the two-variable \UNFOtTR{} in the presence of one transitive relation. The proof is
a straightforward adaptation of the lower bound proof for \GFt{} with transitive relations in guards \cite{Kie06}.

For the upper bound, we design an \AExpSpace{} algorithm. Given $\phi$ in \UNFOtTR{} this algorithm converts it
to normal form $\phi'$ and then looks for a model of $\phi'$ of form (f3), that is, a light tree-like model of with doubly exponentially bounded transitive paths (as in Lemma \ref{l:lshortpaths2}). As we proved, $\phi'$ has such a model iff $\phi$ has a finite model. We first calculate the bound $\hat{M}$ on the rank
of transitive paths, and let  $M:=(\hat{M}+1)^{2k} \cdot |\GGG|$, 
where $\GGG$ is the set of all the generalized types over the signature of 
$\phi'$. Both values are bounded doubly exponentially in $|\phi|$.

In an alternating fashion we construct a single downward path of a model.
We start from the root, by guessing its generalized type, and setting the values of variables $r_1, \ldots, r_{2k}$ to $0$. The intention is
that $r_u$ is the maximal rank of a downward ${T}_u$-path ending at the current  node.
Then, having constructed an element $a$,  we universally choose a $\forall\exists$-conjunct of $\phi'$, and, if $a$ is not a witness for itself, add a witness
$a'$ for $a$ and this conjunct. 
We guess a $2$-type joining $a$ and $a'$, guess the generalized type of $a'$ and update the values
of the $r_u$, in accordance with the guessed $2$-type (each value may stay unchanged, increase by $1$, or reset to $0$).
We check that $a'$ is indeed a witness for $a$ and the considered conjunct and verify the LCC conditions between the generalized types of $a$ and $a'$, rejecting if any of these fails.
We reject also if any of the $r_u$ exceeds $\hat{M}$. We keep an additional counter, measuring the depth of the current node
in the path constructed so far, and accept if this depth exceeds $M$.

It is clear that the described algorithm can be implemented in \AExpSpace{}: we only need to store information about a pair of nodes, values of the $r_u$ plus a counter. These can be written using exponentially many bits. 
 
It is also not difficult to see that the above-described algorithm accepts its input $\phi$ iff $\phi$ has a finite model. If $\phi$ has a finite
model then its normal form $\phi'$ also has such a model. Let $\str{A}\models \phi'$ be a light tree-like model promised in Lemma \ref{l:lshortpaths2}. The algorithm can then accept by making all the guesses in accordance with $\str{A}$. In the opposite direction,
if the algorithm has an 
accepting run, then from this run we can naturally infer a tree-like structure $\str{A}^*$ consisting of at most $M{+}1$ levels.
Note that on each path of length $M+1$ from the root of to a leaf  in $\str{A}^*$ there is a pair of nodes for which the guessed generalized
types, and the calculated values of all the $r_u$ are identical. 
Cut each branch at the first position on which the above parameters reappear and make a link from this point to
the their first occurrence  on the considered branch. Naturally unravel so-obtained structure into an infinite tree-like structure $\str{A}$.
It should be clear that the guessed generalized types still respect LCCs in $\str{A}$ and that the values of the $r_u$ copied to $\str{A}$ from $\str{A}^*$
remain correct. Thus $\str{A}$ is indeed a tree-like model of $\phi'$ with appropriately bounded transitive paths.
\end{proof}

\section{The general case} \label{s:general}

In this section we consider \UNFOTR{} with arbitrarily many variables. The material here is self-contained, we will, however,
make some references to Section \ref{s:2var}, in which we dealt with the simplified two-variable case. In such references we
will emphasise the differences and similarities between both variants.

\subsection{Tree pruning} \label{s:treelike} 

\subsubsection{Tree-like unravellings}
We first naturally generalize the notion of light tree-like structures used in the two-variable case.
Recall that any set consisting of a node $b$ of a tree and all its children is called a \emph{family},
or the \emph{downward family of} $b$.
We say that $\str{A}$ is a \emph{tree-like structure} 
if its nodes can be arranged into a rooted tree  in such a way that
if $\str{A} \models B\bar{a}$ for some non-transitive relation symbol $B$, 
then 
$\bar{a}$ is contained in some family, and
if $\str{A} \models T_uaa'$ for some transitive $T_u$, then  either $a=a'$ or there is a sequence of distinct nodes 
$a=a_0, a_1, \ldots, a_k = a'$ such that $a_i$ and $a_{i+1}$ belong to the same family and $\str{A} \models T_ua_ia_{i+1}$.
A slight difference, compared to light tree-like structures is that now we admit direct (omitting the parent) transitive connections between some children of a
node.  
For a tree-like structure $\str{A}$
and $a \in A$ we denote by $A_a$ the set of all nodes
in the subtree rooted at $a$ and by $\str{A}_a$ the corresponding substructure.

Let us recall  some furhter definitions, related to tree-like structures, used in the two-variable case.
Let $\str{A}$ be a tree-like structure. A sequence of nodes $a_1, \ldots, a_N \in A$ is a \emph{downward path} 
in $\str{A}$ if for each $i$ 
$a_{i+1}$ is a child of $a_i$.
A \emph{downward}-${T}_u$-\emph{path}
is a downward path such that for each $i$ we have $\str{A}\models{T}_ua_ia_{i+1}$. 
The ${T}_u$-\emph{rank} of a downward-${T}_u$-path $\vectorize{a}$,
$\rank{\str{A}}{u}{\vectorize{a}}$,
is the cardinality of the set $\{i:\str{A}\models\neg{T}_ua_{i+1}a_{i}\}$. The ${T}_u$-\emph{rank} of an element $a\in A$ is defined as $\rank{\str{A}}{u}{a}=\sup\{\rank{\str{A}}{u}{\vectorize{a}}:\vectorize{a}=a,a_2,\ldots,a_N;\vectorize{a} \text{ is a downward-}{T}_u\text{-path}\}$.
For an integer $M$, we say that $\str{A}$ has downward-${T}_u$-paths \emph{bounded by} $M$  when for all $a\in A$ we have $\rank{\str{A}}{u}{{a}}\leq M$, 
and that $\str{A}$ has \emph{transitive paths bounded by} $M$ if it has downward-${T}_u$-paths 
bounded by $M$ for all $u$. Note that a downward-${T}_u$-path bounded by $M$ may have more than $M$ nodes, 
 as the symmetric ${T}_u$-connections do not increase  the rank.

Given an arbitrary model $\str{A}$ of a normal form \UNFOTR{} formula $\phi$ we can simply construct its tree-like model
of degree bounded by $|\phi|$. Essentially, the construction works as in the two-variable case. We define a $\phi$-\emph{tree-like unraveling} $\str{A}'$ of $\str{A}$, together with an associated function $\fh:A' \rightarrow A$  in the following way. 
$\str{A}'$ is divided into levels $L_0, L_1, \ldots$. Choose an
arbitrary element $a \in A$ and add to level $L_0$ of $A'$  an element $a'$ such that $\type{\str{A}'}{a'}=\type{\str{A}}{a}$; set $\fh(a')=a$. 
The element $a'$ will be the only element of $L_0$ and will become the root of $\str{A}'$.
Having defined $L_i$ repeat the following for every $a' \in L_i$. Choose in $\str{A}$ a $\phi$-witness structure for $\fh(a')$. Assume it consists of $\; \fh(a'), a_1, \ldots, a_s$. Add a fresh copy $a_j'$  of every $a_j$ to $L_{i+1}$, make 
$\str{A}' \restr \{a', a_1', \ldots, a_s' \}$  isomorphic to $\str{A} \restr \{\fh(a'), a_1, \ldots, a_s\}$ and set $\fh(a_j')=a_j$. Complete the definition
of $\str{A}'$ transitively closing all relations from $\sigma_\cDist$.

As in the two-variable case we easily get the following fact.

\begin{lemma}[(f1) $\leadsto$ (f2)] \label{l:treelike}
Let $\str{A}$ be a model of a normal form \UNFOTR{} formula $\phi$. Let $\str{A}'$ be a $\phi$-tree-like unraveling of $\str{A}$.
Then $\str{A}' \models \phi$ and $\str{A}'$ is a tree-like structure of degree bounded by $|\phi|$ and transitive paths
bounded by $|A|$.
\end{lemma}
\begin{proof}
It is readily verified that $\str{A}'$ meets the properties required by Lemma \ref{l:homomorphisms}. In particular $\fh$ is the
required homomorphism. That $\str{A}'$ is tree-like and has an appropriately bounded degree is also straightforward. 
For the last condition
assume to the contrary that  there exist $u$ and a downward-${T}_u$-path $(a_i)_{i=0}^N$
in $\str{A}'$  with rank bigger than $|A|$. Then there are indices $i_0,\ldots,i_{|A|}$ such that 
$\str{A}\models{T}_ua_{i_j}a_{i_j+1}\wedge\neg{T}_ua_{i_j+1}a_{i_j}$. Since  $\fh$ preserves the connections
between elements and their witnesses
we have $\str{A}_0\models{T}_u\fh(a_{i_j})\fh(a_{i_j+1})\wedge\neg{T}_u\fh(a_{i_j+1})\fh(a_{i_j})$. By the pigeonhole principle there exist $x<x'$ such that $\fh(a_{i_x})=\fh(a_{i_{x'}})$. This gives, by transitivity of ${T}_u$, that $\str{A}_0\models{T}_u\fh(a_{i_x+1})\fh(a_{i_x})$. Contradiction. 
\end{proof}

We often work with tree-like models $\str{A}$ of normal form $\phi$ in which the downward family of every element $a$ forms a
$\phi$-witness structure. In such case we call this downward family \emph{the} $\phi$-witness structure for $a$ even if
some other $\phi$-witness structures for $a$ exist in $\str{A}$.

\subsubsection{Declarations}

Our next task is making the transition (f2) $\leadsto$ (f3).
To this end we introduce an apparatus of \emph{declarations} that allows us to perform some
surgery on tree-like models of  normal form formulas. Its main purpose is dealing with
their universal conjuncts $\forall \bar{x} \neg \phi_0(\bar{x})$. 
For a normal form $\phi$,  a $\phi$-declaration is a description of some patterns of connections
taking into account the literals of $\phi_0$ (and, for technical reasons, some additional
transitive atoms, equalities and inequalities). In particular, it may describe some dangerous
patterns, leading to a violation of $\phi_0$. 

We remark that, while declarations are a counterpart of light-declarations
from the previous section, they are much more complicated than the latter. One difference, worth pointing out,
is that light declarations were defined as independent from the given formula: they were only dependent on the signature.
In the current scenario, which patterns are considered by declarations depends on the literals of $\phi_0$. Making declarations independent from $\phi$ would be possible, but would not allow us to obtain tight complexity bounds.

Let us turn to a formal definition, for which recall that $\phi_0=\phi_0(x_1, \ldots, x_t)$ is in NNF,
and that $2k$ 
is the number of transitive relations, 
 
\begin{definition}
Let $\phi$ be a normal form \UNFOTR{} formula. 
Let $\mathcal{R}$ be the set consisting of all non-transitive literals (atoms or negated atoms) that appear in $\phi_0$ (recall that
only atoms which have at most one variable may be negated), $\mathcal{T}=\{1,\ldots,2k\}\times\{1,\ldots,t\}^2$ 
and $\mathcal{Q}=\{1,\ldots,t\}$. 
A \emph{$\phi$-declaration}  is a set consisting of some triples $(R,T,Q)$ such that $R\subseteq\mathcal{R}$, $T\subseteq\mathcal{T}$ and $Q\subseteq\mathcal{Q}$.

\end{definition} 

A triple $d=(R,T,Q)$ may be alternatively viewed as a formula describing a pattern of connections on a tuple consisting of $t+1$ (not necessarily distinct) elements: $\psi_{d}(x_1,\ldots,x_t,y)=\bigwedge_{r\in R}r(\bar{x})\wedge\bigwedge_{(u,j,j')\in T}T_ux_jx_{j'}\wedge\bigwedge_{i\in Q}x_i=y\wedge\bigwedge_{i\in\mathcal{Q}\setminus Q}x_i\neq y$. In the sequel we often identify the triple and 
the formula that  it represents.

Let $\str{A}$ be a tree-like structure  and $a\in A$. 
We say that $a$ \emph{respects} a $\phi$-declaration $\fd$ if for each $\psi\in\fd$ we have $\str{A}_a\models\neg\exists\bar{x}\psi(\bar{x},a)$. Given an element $a$ in $\str{A}$, we denote  by $\dec{\str{A}}{\phi}{a}$ the (unique) maximal declaration respected by $a$.
Note that if $a_0$ is the root of $\str{A}$ then, knowing $\dec{\str{A}}{\phi}{a_0}$, we can determine if $\str{A} \models \forall \bar{x} \neg \phi_0(\bar{x})$.

Let us now give some intuitions and describe how we are going to use declarations. 
We work with tree-like structures with $\phi$-declarations assigned to all its nodes.
Assigning $\fd$ to a node $a$ of a structure $\str{A}$ may be treated as \emph{a promise that \textbf{none} of the patterns described by $\fd$
appear in $\str{A}_a$}. 
Note that we do not require that $\fd$ equals $\dec{\str{A}}{\phi}{a}$. 
Given a system of declarations assigned to all nodes of $\str{A}$ we formulate some natural local conditions such that their violation at a node $a$ breaks the promise of $a$ (i.e., some forbidden pattern occurs), and, the other way round, they are sufficient to guarantee that for every node $a$ the declaration $\fd$ assigned to $a$ is a subset of $\dec{\str{A}}{\phi}{a}$,
which means that $a$ respects $\fd$, that is, fulfills its promise. This allows us to proceed as follows:
Take a tree-like model $\str{A} \models \phi$, perform on it some surgery, obtaining a new tree-like structure $\str{A}'$. Assign to
the nodes of $\str{A}'$ a system of $\phi$-declarations in such a way that (i) the root of $\str{A}'$ gets the
declaration $\dec{\str{A}}{\phi}{a_0}$ where $a_0$ is the root of $\str{A}$, and (ii) for a node $a'$ its downward family $F'=\{a', a_1', \ldots, a'_s\}$ gets the declarations
$\dec{\str{A}}{\phi}{a}, \dec{\str{A}}{\phi}{a_1}, \ldots, \dec{\str{A}}{\phi}{a_s}$, where $F=\{a, a_1, \ldots, a_s\}$ is the downward family
of some node $a$ from $\str{A}$, and the structures on $F$ and $F'$ are isomorphic. This guarantees that the system of declarations satisfies the local conditions and thus that its promises
are fulfilled. Due to the declaration of the root we have that $\str{A}'$ satisfies the universal conjunct of $\phi$.

We are now ready for the details. 
Let $F=\{a, a_1,\ldots,a_s\}$ be the downward family of a node  $a$. We say that a function $\ff:\{1,\ldots,t\}\to\{a,a_1,A^-_{a_1},\ldots,a_s,A^-_{a_s}\}$ is a \emph{fitting} (to $F$). We think that a fitting describes a distribution of elements of a $t$-element tuple of nodes of $\str{A}_a$ among the downward
family of $a$ and the subtrees rooted
at the  children of $a$.
With a fitting $\ff$ 
we associate a function $\bar{\ff}:\{1,\ldots,t\}\to\{a,a_1,\ldots,a_s\}$ defined as follows: $\bar{\ff}(i)=a$ iff $\ff(i)=a$ and $\bar{\ff}(i)=a_j$ iff $\ff(i)\in\{a_j,A^-_{a_j}\}$. 
Let $b_1,\ldots,b_t=\bar{b}\subseteq A_a$. The fitting $\ff$ \emph{induced by} $\bar{b}$ is defined naturally: $\ff(i)=a$ iff $b_i=a$, $\ff(i)=a_j$ iff $b_i=a_j$ and $\ff(i)=A^-_{a_j}$ iff $b_i\in A_{a_j}\setminus \{a_j\}$. 

Let  $\ff$ be a fitting, $(R,T,Q)$  a tuple belonging to some declaration and $F=\{a, a_1, \ldots, a_s \}$ the downward family of some $a$. If  $r$ is a literal from $R$ (resp.~a tuple $(u,j,j') \in T)$ then we denote by $\var{r}$ the set of the indices of variables of $r$ (resp.~the set $\{j, j'\}$). If $r$ is a literal from $R$ or a literal $T_ux_jx_{j'}$ then we say that $r$ is \emph{fully fitted} (to $F$) 
if $\ff(\var{r})\subseteq F$.

Now we define the \emph{local consistency conditions (LCCs)} for a system of declarations. Consider declarations $\fd,\fd_1,\ldots,\fd_s$ assigned to the elements of some family $F=\{a,a_1,\ldots,a_s\}$. We say that they \emph{satisfy LCCs at $a$} if for each fitting $\ff$ and $\psi_{(R,T,Q)}\in\fd$ at least one of the following conditions holds.

\begin{enumerate}[(l1)]

\item Some $R$-conjunct $r\in R$ is not fully fitted, and
 $a\in\ff(\var{r})$ or $a_j,a_{j'}\in\ff(\var{r})$ for some $j\neq j'$.\label{lone}
\item Some $R$-conjunct $r\in R$ is fully fitted 
but $\str{A}\models\neg r(\ff(1),\ldots,\ff(t))$.

\item Some $T$-conjunct $T_ux_jx_{j'}$ is fully fitted but $\str{A}\models\neg T_u\ff(j)\ff(j')$.\label{lthree}

\item For some $Q$-conjunct $x_i=a$ we have $\ff(i)\neq a$.\label{lfour}

\item For some $(\mathcal{Q}{\setminus}Q)$-conjunct $x_i\neq a$ we have $\ff(i)=a$.\label{lfive}

\item Some $R$-conjunct $r\in R$ is not fully fitted and
is `distributed over several subtrees', that is $|\bar{\ff}(\var r)\cap\{a_1,\ldots,a_s\}|\geq2$. 
\label{lsix}

\item Two elements of two different subtrees cannot be transitively joined due to the structure on $F$, that is for 
 some $T$-conjunct $T_ux_jx_{j'}$  we have $\bar{\ff}(j)\neq\bar{\ff}(j')$ but $\str{A}\models\neg T_u\bar{\ff}(j)\bar{\ff}(j')$.\label{lseven}

\item All elements are fitted to a single subtree, that is for some  $i$ and all $j$ we have $\bar{\ff}(j)=a_i$, and the promise is propagated to this subtree: $(R,T,\ff^{-1}(a_i))\in\fd_i$.\label{leight}

\item There exists $i$ such that $a_i \in \rng\bar{\ff}$ and $\fd_i$ contains $(R',T',Q')$ defined as follows: fix some 
$h \in \{1,\ldots,t\}\setminus\bar{\ff}^{-1}(a_i)$ and (i) $R':=\{r\in R:\var r\subseteq\bar{\ff}^{-1}(a_i)\}$, (ii) $T'$ is the minimal set such that for  $(u,j,j') \in T$  if $\bar{\ff}(j)=\bar{\ff}(j')=a_i$ then $(u,j,j')\in T'$ and if $\bar{\ff}(j)=a_i$ (resp. $\bar{\ff}(j)\neq a_i$) and $\bar{\ff}(j')\neq a_i$ (resp. $\bar{\ff}(j')=a_i$) then 
$(u,j,h)\in T'$ (resp. $(u,h,j')\in T')$, and (iii) $Q'=\ff^{-1}(a_i)\cup \{h \}\cup(\mathcal{Q}\setminus\bar{\ff}^{-1}(a_i))$.\label{lnine}
\end{enumerate}

Note that the above conditions are of two sorts.
Conditions (l\ref{lone})--(l\ref{lseven}) describe situations in which we immediately, just looking at the structure on $F$,
observe that any tuple of elements corresponding to the given fitting does not
break the promise of $a$. Conditions (l\ref{leight})--(l\ref{lnine}), on the other hand,
describe situations in which, intuitively speaking, we need to relegate such observation to one
of the children of $a$. 

Given a structure $\str{A}$ we say that a  system of declarations $(\fd_a)_{a\in A}$ is \emph{locally consistent} if it satisfies LCCs at each $a\in A$ and is \emph{globally consistent} if $\fd_a\subseteq\dec{\str{A}}{\phi}{a}$ for each $a\in A$.
Note that the global consistency means that the promises of all nodes are fulfilled. Conditions (l\ref{lone})--(l\ref{lnine}) are tailored so that local
and global consistency play along in the following lemma, which is a natural counterpart of Lemma \ref{l:llocglobl}.

\begin{lemma}[Local-global]\label{l:locglob}
Let $\str{A}$ be a tree-like structure. Then (i) if a system of declarations $(\fd_a)_{a\in A}$ is locally consistent then it is globally consistent (ii) the canonical system of declarations $(\dec{\str{A}}{\phi}{a})_{a\in A}$ is locally consistent.
\end{lemma}

\begin{proof} 
(i)  Assume to the contrary that there exist $a\in A$, $\psi\in\fd_a$ and $\bar{b}\subseteq A_a$ such that $\str{A}_a\models\psi(\bar{b},a)$. Take the fitting $\ff$ to the downward family $F=\{a,a_1,\ldots,a_s\}$ of $a$ induced by $\bar{b}$. By the choice of $\bar{b}$, none of (l\ref{lone})--(l\ref{lseven}) holds. Thus $\bar{b}\not\subseteq F$ and there exist some $a_i \in F \setminus \{a \}$ and $\psi'\in\fd_{a_i}$ such that $\str{A}_{a_i}\models\psi'(\bar{b}',a_i)$ where $\bar{b}'=b_1',\ldots,b_t'$ is defined as follows: $b_j'=b_j$ if $\bar{\ff}(j)=a_i$, otherwise $b_j'=a_i$. Denote by $\text{depth}(\bar{b})$ the maximal level of $\str{A}$ inhabited by an element of $\bar{b}$. Obviously
$\text{depth}(\bar{b}') \le \text{depth}(\bar{b})$. 
Thus after finitely many steps we get $a^*$, $\psi^{*}\in\fd_{a^{*}}$ and $\bar{b}^{*}$ contained in the downward family $F^*$ of $a^*$ such that $\str{A}_{a^*}\models\psi(\bar{b}^*,a^*)$. But this cannot happen, since neither (l\ref{leight}) nor (l\ref{lnine}) can hold for the fitting to $F^*$ induced by $\bar{b}^*$.

\smallskip\noindent
(ii) Follows from a careful inspection of the definition of LCCs. Basically, if for some $a$, $\ff$ and $\psi\in\dec{\str{A}}{\phi}{a}$ none of (l\ref{lone})--(l\ref{lnine}) holds then we can find $\bar{b}\subseteq A_a$ such that $\str{A}\models\psi(\bar{b},a)$. Indeed, use 
non-satisfaction of (l\ref{leight}) and (l\ref{lnine}) to find fragments of $\bar{b}$ belonging to the respective subtrees ($\bar{b}\cap\str{A}_{a_i}$)
(i.e., to find appropriate $b_j\in A_{a_i}$ for $j\in\bar{\ff}^{-1}(a_i)$); non-satisfaction of  (l\ref{lone})--(l\ref{lseven}) implies that they are connected so that $\psi(\bar{b},a)$ holds. 
\end{proof}

\subsubsection{Shortening transitive paths}

To make transition (f2) $\leadsto$ (f3) we proceed, essentially, as in the two-variable case, that is we perform the same
tree pruning process, with the same pruning strategy, with $\phi$-declarations playing now the role of light declarations. 
Correctness of this approach is justified mainly by Lemma \ref{l:locglob}, being a counterpart of Lemma \ref{l:llocglobl}. The only difference is the way we 
argue that a model obtained in the pruning process  satisfies the universal conjunct of the given formula. In the two-variable
case it was done by constructing an appropriate homomorphism into the original model. In the presence of more than two
variables it is not always possible. Instead, it is sufficient to use the fact that the system of declarations on
the produced model is globally consistent and look at the declaration of its root.

A reader understanding the tree pruning process in the two-variable case may thus safely skip this section, noting only 
the above-mentioned difference in dealing with the $\forall$-conjunct in the correctness proof.

For completeness, we give here a detailed description, presenting it, however, in a slightly more compact way than
in the two-variable case.

\begin{lemma}[(f2) $\leadsto$ (f3)]
\label{l:short}
If a normal form \UNFOTR{} formula $\phi$ has a tree-like model of degree bounded linearly and transitive paths bounded by some natural number then  it has a tree-like model of degree bounded linearly  and transitive paths bounded doubly exponentially in $|\phi|$.

\end{lemma}
\begin{proof} 
Let $\str{A} \models \phi$ be a tree-like model with bounded transitive paths. Let $M_\phi:=|\AAA| \cdot |D_\phi| +2$, where $\AAA$ is the set of atomic $1$-types over the signature of $\phi$ and $D_\phi$ is 
the set of $\phi$-declarations. Clearly, $M_\phi$ is bounded doubly exponentially in $|\phi|$.

Consider a mapping: $A\ni a\mapsto_{\fg}(\type{\str{A}}{a},\dec{\str{A}}{\phi}{a})$. 
Observe that $|\rng\fg| \le M_\phi -2$. 
We construct a tree-like model $\str{A}'$ having levels $L_0', L_1', \ldots$.
During our construction we maintain
a pattern function $\fp:A'\to A$ and
a function $\fs:A'\to\mathrm{Sym}(\{1,\ldots,2k\})$ whose purpose is to define some order of shortening  paths at a given node. Intuitively,  for $\fs(a)=\tau$, if  $v<v'$ then we  prefer to shorten ${T}_{\tau(v)}$ over ${T}_{\tau(v')}$.

Let $L_0'$ consist of $a_0'$---a copy of the root $a_0$ of $\str{A}$ (i.e. $\type{\str{A}'}{a_0'}=\type{\str{A}}{a_0}$). Put $\fp(a_0')=a_0$ and set $\fs(a_0')$ arbitrarily. 
Suppose that we have defined $L_i'$.
For each $a'\in L_i'$ let $\{\fp(a'),a_1,\ldots,a_s\}$ be the downward family of $\fp(a')$ in $\str{A}$ and let $\fs(a')=\tau$. Take fresh copies $a_j'$ of $a_j$ and make $\str{A'}\restr\{a',a_1',\ldots,a_s'\}$ isomorphic to $\str{A}\restr\{\fp(a'),a_1,\ldots,a_s\}$.

Presently we set the $\fp(a_j')$ and $\fs(a_j')$. Let 
$K(a_j')=K:=\{v:\str{A}\models\neg{T}_{\tau(v)}\fp(a')a_j\}$ (the ${T}_{\tau(v)}$ \emph{killed} at $a_j'$), 
$S(a_j')=S:=\{v:\str{A}\models{T}_{\tau(v)}\fp(a')a_j\wedge{T}_{\tau(v)}a_j\fp(a')\}$ (the ${T}_{\tau(v)}$ \emph{sustained} at $a_j'$) and 
$D(a_j')=D:=\{v:\str{A}\models{T}_{\tau(v)}\fp(a')a_j\wedge\neg{T}_{\tau(v)}a_j\fp(a')\}$ (the ${T}_{\tau(v)}$ \emph{diminished} at $a_j'$).
In this paragraph we refer to them without the argument, as it is always the same.
If $D\neq\emptyset$ then let 
$v_D(a_j')=v_D:=\min D$ and take as $\fp(a_j')$ a $b_j\in A_{a_j}$ such that (i) $\fg(b_j)=\fg(a_j)$ (ii) for all $v<v_D$, $v\in S$: $\rank{\str{A}}{\tau(v)}{b_j}\leq\rank{\str{A}}{\tau(v)}{\fp(a')}$ (iii) $\rank{\str{A}}{\tau(v_D)}{b_j}$ is the lowest possible
among the elements satisfying (i) and (ii). 
Note that such an element exists 
($a_j$ satisfies (i) and (ii); for (iii), the ranks in $\str{A}$ are bounded, in particular finite) 
and $\rank{\str{A}}{\tau(v_D)}{b_j}<\rank{\str{A}}{\tau(v_D)}{\fp(a')}$. If $D=\emptyset$ then let $\fp(a_j')=a_j$.
If $K\not=\emptyset$ then let $v_K=\min K$
 and set $\fs(a_j'):=\tau'$ where $\tau'$ is defined as follows: 
for $v<v_K$ let $\tau'(v)=\tau(v)$, for $ v_K \le v <2k$ let $\tau'(v)=\tau(v+1)$ and let $\tau'(2k)=\tau(v_K)$. 
Otherwise (that is $K=\emptyset$) 
put $\fs(a_j')=\tau$.    
To finish the construction, transitively close all the appropriate relations in $\str{A}'$. 

We claim that $\str{A}'$ constructed as above is a model of $\phi$ and has the desired properties.
$\forall\exists$-conjuncts are satisfied since for all $a'\in A'$ the structure on the downward family of $a'$ in $\str{A}'$ is  isomorphic to the structure on the downward family
of  $\fp(a')$ in $\str{A}$ and the latter is the $\phi$-witness structure for $\fp(a')$.

For the universal conjunct of $\phi$ consider the system of declarations $(\dec{\str{A}}{\phi}{\fp(a')})_{a'\in A'}$. 
Note that in this system the declarations on the downward family of any node $a'$ in $A'$
are copies of the declarations on the downward family of $\fp(a')$ in $\str{A}$ in the canonical system of declarations on $\str{A}$. 
This canonical system on $\str{A}$ is locally consistent by part (ii) of Lemma \ref{l:locglob}. This in turn gives that
the system we have defined on $\str{A}'$ is also locally consistent. By part (i) of Lemma \ref{l:locglob} this system is also
globally consistent.
 In particular, since $\phi_0$ is equivalent to $\phi_0^1\vee\ldots\vee\phi_0^s$ where the $\phi_0^j$ are conjunctions of some $\mathcal{R}$ and $\mathcal{T}$ formulas, for the root $a_0'$ of $\str{A}'$,  for each $j$ and $Q\subseteq\mathcal{Q}$ we have $\str{A}'=\str{A}_{a_0'}'\models\neg\exists\bar{x}(\phi_0^j(\bar{x})\wedge\bigwedge_{i\in Q}x_i=a_0\wedge\bigwedge_{i\in \mathcal{Q}\setminus Q}x_i\neq a_0)$, so $\str{A}'\models\forall\bar{x}\neg\phi_0(\bar{x})$.

That the degree of nodes in $\str{A}'$ is bounded linearly in $\phi$ follows from the fact that it was so bounded in $\str{A}$. 

It remains to show that the transitive  paths in $\str{A'}$ are doubly exponentially bounded. 
Let us first make an auxiliary estimation.
\begin{claim} \label{c:shortpaths}
Let $v_0$ and a downward-${T}_u$-path $\vectorize{a}=(a_i)_{i=1}^N$ in $\str{A}'$ be such that for all 
$v\leq v_0$ we have that the map $i\mapsto\fs(a_i)(v)$
is constant
(in this case, slightly abusing  notation, we write $\fs(\vectorize{a})(v)=\fs(a_1)(v)$).
Let $u:=\fs(\vectorize{a})(v_0)$. 
Then $\rank{\str{A}'}{u}{\vectorize{a}}\leq (M_\phi+1)(\sum_{v< v_0}\rank{\str{A}'}{\fs(\vectorize{a})(v)}{\vectorize{a}})+M_\phi$. 
\end{claim}
\begin{proof}
Assume to the contrary that there is a downward-${T}_u$-path $\vectorize{a}$ in $\str{A}'$  meeting the required conditions such that $\rank{\str{A}'}{u}{\vectorize{a}}>(M_\phi+1)(\sum_{v<v_0}\rank{\str{A}'}{\fs(\vectorize{a})(v)}{\vectorize{a}})+M_\phi$. Then there are more than $M_\phi(1+\sum_{v<v_0}\rank{\str{A}'}{\fs(\vectorize{a})(v)}{\vectorize{a}})$ indices $i$ such that $v_0=v_D(a_i)$. So there exist indices $i_1,\ldots,i_{M_\phi}$ such that for all $j$ it holds that $v_D(a_{i_j})=v_0$ and for all $i_1\leq i\leq i_{M_\phi}$ and $v<v_0$ we have $v\not\in D(a_i)$ (and thus $v\in S(a_i)$). 
It follows that for all $v\leq v_0$ the function  $i\mapsto\rank{\str{A}}{\fs(\vectorize{a})(v)}{\fp(a_i)}$ is non-increasing 
(on $\{i:i_1\leq i\leq i_{M_{\phi}}\}$)
and the function $j\mapsto\rank{\str{A}}{\fs(\vectorize{a})(v_0)}{\fp(a_{i_j+1})}$ is strictly decreasing. By the pigeonhole principle there exist $x<x'<M_\phi$ satisfying $\fg(\fp(a_{i_x+1}))=\fg(\fp(a_{i_{x'}+1}))$. This contradicts the choice of $\fp(a_{i_x+1})$. 
\end{proof}

The above claim allows us in particular to compute a (uniform) doubly exponential bound  on $\rank{\str{A}'}{u}{\vectorize{a}}$ for all $v_0$, $u$ and $\vectorize{a}$ as in assumption. Denote this bound by $\overline{M}_\phi$.

Consider now any downward-${T}_u$-path $\vectorize{a}=(a_i)_{i=1}^N$ in $\str{A}'$.
For each node $a_i$ from $\vectorize{a}$ let $v_{u}(a_i)$ be such
that $\fs(a_i)(v_u)=u$. Due to the strategy that we use to define $\fs$ the value of $v_u$ is non-increasing along $\vectorize{a}$. Indeed, when moving from $a_i$ to $a_i+1$ the value of $v_u$ is either unchanged or decreases by $1$; the only chance of increasing it would be 
to change it to $2k$ but this happens only when ${T}_u$ is killed. Let us divide $\vectorize{a}$ into fragments $\vectorize{a}_1, \vectorize{a}_2, \ldots$
on which $v_u$ is constant. The number of such fragments  is obviously bounded by $2k$. On each of such fragments
$\vectorize{a}_i$ for all $v \le  v_u$ we have that $\fs(\vectorize{a}_i)(v)$ is constant. So we can apply Claim \ref{c:shortpaths} to bound 
$\rank{\str{A}'}{u}{\vectorize{a}_i}$ 
by $\overline{M}_\phi$. This gives the desired doubly exponential bound $\hat{M}_\phi=2k \overline{M}_\phi$ on $\rank{\str{A}'}{u}{\vectorize{a}}$. 
This finishes the proof of Lemma \ref{l:short}. 
\end{proof}

\subsubsection{Regular tree-like models}

We conclude this section by showing that for finitely satisfiable formulas we can always construct \emph{regular} tree-like models with bounded transitive paths. We recall that this step was unnecessary in the two-variable case and was ommited there.

Let us introduce a tool, which  allows us to verify the property of having bounded transitive paths looking only at some local conditions.
Let $\str{A}$ be a tree-like structure with root $a_0$. Then a function $\mathcal{S}:A\to\{0,\ldots,M\}$ is 
the 
$({T}_u,M)$-\emph{stopwatch labeling} if: 
$\mathcal{S}(a_0)=0$;
for every $a\in A$ and its child $b$:
(i) if $\str{A}\models{T}_uab\wedge{T}_uba$ then $\mathcal{S}(b)=\mathcal{S}(a)$,
(ii) if $\str{A}\models{T}_uab\wedge\neg{T}_uba$ then $\mathcal{S}(b)=\mathcal{S}(a)+1$ (in particular $\mathcal{S}(a)<M$)
(iii) if $\str{A}\models\neg{T}_uab$ then $\mathcal{S}(b)=0$.

It is easy to see that the value of the ($T_u,M$)-stopwatch labeling in $a$ is equal to the maximal rank of a downward-$T_u$-path ending in $a$, therefore
the
(${T}_u,M$)-stopwatch labeling exists iff the structure has downward-${T}_u$-paths bounded by $M$.
\begin{lemma}
[(f3)$\leadsto$(f4)]
\label{l:reg}
If $\phi$ has a tree like model with linearly bounded degree and doubly exponentially bounded transitive paths, then it has a regular such model, that is a model with doubly exponentially many non-isomorphic subtrees.
\end{lemma}
\begin{proof} 
Let  $\str{A}$ be a tree-like model of $\phi$ with linearly bounded degree and transitive paths bounded doubly exponentially by $\hat{M}_\phi$. For each ${T}_u$ take 
the 
$({T}_u,\hat{M}_\phi)$-stopwatch labeling $\mathcal{S}_u$ of $\str{A}$. 
Consider the 
mapping $A\ni a\mapsto_{\fg} (\type{\str{A}}{a},\dec{\str{A}}{\phi}{a},(\mathcal{S}_u(a))_{u=1}^{2k})$. Note that $|\rng\fg|$ is bounded
doubly exponentially in $|\phi|$. 
We rebuild $\str{A}$ into a regular model $\str{A}'$.  

For each $p\in\rng\fg$ choose a representative $\fc(p)\in\fg^{-1}(p)$. Add to $L_0'$ an element $a_0'$ such that $\type{\str{A}'}{a_0'}=\type{\str{A}}{a_0}$, where $a_0$ is the root of $\str{A}$, and let $\fp(a_0')=\fc(\fg(a_0))$. Having defined $L_i'$, for $i \ge 0$, repeat the following for all $a'\in L_i'$. Denoting  $\fp(a'),a_1,\ldots,a_s$  the downward family of $\fp(a')$ in $\str{A}$, add a fresh copy $a_i'$ of each $a_i$ to $L_{i+1}'$ and make $\str{A}'\restr\{a',a_1',\ldots,a_s'\}$ isomorphic to $\str{A}\restr\{\fp(a'),a_1,\ldots,a_s\}$. Set $\fp(a_i'):=\fc(\fg(a_i))$. 
Finally, transitively close all transitive relations in $\str{A}'$. 

The proof that $\str{A}' \models \phi$ is similar to the corresponding proof in Lemma \ref{l:short}:
we observe that all elements have appropriate $\phi$-witness structures copied from $\str{A}$ and 
then use the apparatus of declarations to argue that the $\forall$-conjunct of $\phi$ is respected.
 By construction $\str{A}'$ is a regular tree-like model with the number of different subtrees bounded by $|\rng\fg|$. 
To see that $\str{A}'$ has transitive paths bounded by $\hat{M}_\phi$ create stopwatch labelings for $\str{A}'$ just by transferring them
 from $\str{A}$ using $\fp$. It is not difficult to see that they meet the conditions from the definition of stopwatch labelings.  
\end{proof}

\subsubsection{A few words for  automata fans}
The content of Section \ref{s:treelike} could be alternatively  presented with help of (B\"uchi) tree automata, however we decided not to present it this way. The reasons will be explained in a moment, after explaining how automata could be used. This comment should not be treated as a formal description (and does not bother with details), but rather as a glossary of terms one can refer to.
Let  $\Sigma$ be the alphabet consisting of the set of possible structures on downward families. There are three automata involved: $\mathcal{A}_1$ that checks whether the $\forall\exists$ conjuncts holds, $\mathcal{A}_2$ for the $\forall$ conjunct and $\mathcal{A}_3(M)$ for checking whether the paths are bounded by $M$. 
The construction of $\mathcal{A}_1$ is straightforward,  
$\mathcal{A}_3$ is also simple and corresponds to our stopwatch labellings. The only problem is the construction of $\mathcal{A}_2$. We solve the corresponding task in our approach using the declarations. Equivalently, we could 
design $\mathcal{A}_3$ with the set of states being the set of all possible declarations (the accepting states) plus the Black Hole (BH) (the rejecting state), and transitions described by the LCCs, plus two additional ones: anyone can go to the BH, and no one escapes the BH. The starting state can be defined using the description in the "For the universal conjunct\ldots " paragraph in the proof of Lemma \ref{l:short}. Now, Lemma \ref{l:locglob} is a counterpart of the statement that $\mathcal{A}_2$ accepts exactly the models of the $\forall$ conjunct. Lemma \ref{l:short} can be stated as: if there exists $M$ such that  $\mathcal{L}(\mathcal{A}_1\cap\mathcal{A}_2\cap\mathcal{A}_3(M))\neq\emptyset$ then there exists
such $M$  doubly exponential in $|\phi|$.  
Finally, Lemma \ref{l:reg} and Theorem \ref{t:algo} can be treated as some standard properties of automata.
So, why we decided not to use automata? Firstly, since it was not the way we originally thought about the problem. Secondly, and mainly, it would not allow us to omit the main obstacles, namely the definitions of the declarations, and the proofs of Lemmas \ref{l:locglob} and \ref{l:short}. Furthermore, we wanted to keep our presentation uniform,
and not to alternate between
logic and automata for this  (not the most complicated) task.
And, finally, we would not need the  full power of B\"uchi automata.

\subsection{Building small finite models}\label{s:main}

In this section we show the following small model property. To this end we make the missing transition (f4) $\leadsto$ (f5). Generally, our constructions here are quite similar to the constructions from the two-variable case. The main differences are: bigger components, using 
the isomorphism types of subtrees of a regular tree-like model
instead of generalized types, and a more complicated way of building witness structures. Also, the proof of correctness of the construction becomes now
more complicated, in particular, it uses non-trivially the regularity of tree-like models from the previous subsection.

\begin{theorem} \label{t:maintr}
Every finitely satisfiable \UNFOTR{} formula $\phi$ has a finite model of size bounded triply exponentially 
in $|\phi|$. 
\end{theorem}

Let us fix a finitely satisfiable normal form \UNFOTR{} formula $\phi$
over a signature $\sigma_{\cBase} \cup \sigma_{\cDist}$ for
$\sigma_{\cDist}=\{T_1, \ldots, T_{2k} \}$. 
Recall that we consider structures that for each $1\leq u\leq k$ interpret the transitive symbol $T_{2u}$ as the inverse of $T_{2u-1}$ and we sometimes write $T_{2u}^{-1}$ for $T_{2u-1}$ and $T_{2u-1}^{-1}$ for $T_{2u}$.
For a  set  
$\mathcal{E} \subseteq \sigma_{\cDist}$, closed under inverses, and $a \in A$ we denote by $[a]_{\mathcal{E}}$ the set consisting of $a$ and all elements $b \in A$ such that
$\str{A} \models T_u ab$ for all $T_u \in \mathcal{E}$. 
Note that $[a]_\mathcal{E}$ is either a singleton or each of the $T_u \in \mathcal{E}$ is total on $[a]_\mathcal{E}$, that is, for each $b_1, b_2 \in [a]_\mathcal{E}$ we have $\str{A} \models T_u b_1 b_2$ for
all $T_u \in \mathcal{E}$.  
We assume that $[a]_\emptyset=A$.
Fix a regular tree-like model $\str{A} \models \phi$, with linearly bounded degree, doubly exponentially bounded 
transitive paths (in this section we denote this bound by $\hat{M}_\phi$) and doubly exponentially many non-isomorphic subtrees, as in Lemma \ref{l:reg}.

We show how to build a `small' finite model $\str{A}' \models \phi$.
In our construction we inductively produce fragments of $\str{A}'$ in which all relations from some subset of $\sigma_{\cDist}$, that is closed under inverses, are total (or, for some technical reasons, such a fragment may be a singleton and then these relations do not need to be total). 
The induction 
is, essentially, over the number of the non-total $T_u$. Intuitively, if a relation is total then it plays no important role, so 
we may forget about it during the construction.
On the $l$-th  level of induction we produce a substructure for every isomorphism type of a subtree of $\str{A}$ 
and any (closed under taking inverses) combination of $2l$ non-total transitive relations, so that its each element has provided its \emph{partial $\phi$-witness structure}.
This partial $\phi$-witness structure is an isomorphic copy of the restriction of the $\phi$-witness structure (in $\str{A}$) for some $a\in A$ to the set $[a]_{\cE_{tot}}$ where $\cE_{tot}$ is the set of current total transitive relations.
Every substructure in the $l$-th level of induction is constructed by an appropriate arrangement of 
some number of basic building blocks, called \emph{components}. Each of the components is obtained by some number of  applications of 
the inductive assumption to situations in that one of the non-total transitive relations and its inverse 
are added to the set of total ones.

An important property of the substructures created during our inductive process is that they admit some partial homomorphisms to the pattern tree-like model $\str{A}$ which restricted to (partial) witness structures act as isomorphisms into the corresponding parts of the $\phi$-witness structures in $\str{A}$.  We impose that every homomorphism respects the
required condition using directly the structure of $\str{A}$.
To this end we introduce further fresh (non-transitive) binary symbols $W^i$ whose purpose is to relate elements to their witnesses.
We number
the elements of the $\phi$-witness structures in $\str{A}$ arbitrarily (recall that each element is a member of its own $\phi$-witness structure) and interpret $W^i$ in $\str{A}$ so that
for each $a,b\in A$, $\str{A}\models W^iab$ iff $b$ is the  $i$-th element of the $\phi$-witness structure for $a$ (from now, for short, we refer to the element $b$ satisfying $W^iab$ as the $i$-th witness for $a$).
Do this in such a way that if two subtrees of $\str{A}$ were isomorphic before interpreting the $W^i$ then they still are after such expansion of the structure.
Now, if we mark $b$ as the $i$-th witness for $a$ during the construction (that is set $\str{A}'\models W^iab$), then for any 
homomorphism $\fh$ we have $\str{A}\models W^i\fh(a)\fh(b)$.
To quickly sum up this operation, we encode the tree-like structure of $\str{A}$ with the relations $W^i$ in such a way, that the regularity of $\str{A}$ is respected (and thus, preserved).

 Let us formally state our inductive lemma. In this statement we do not explicitly include any bound on the size of
promised finite models, but such a bound will be implicit in the proof and will be presented later.

\begin{lemma}[Main construction]\label{l:finitetr} Let $a_0\in A$ and $\mathcal{E}_0\subseteq\sigma_{\cDist}$
be closed under inverses. 
Define $\mathcal{E}_{tot}=\sigma_{\cDist}\setminus\mathcal{E}_0$. 
Let $\str{A}_0=\str{A}_{a_0}\restr[a_0]_{\cE_{tot}}$. Then there exist a finite structure $\str{A}_0'$, a function $\fp:A_0'\to A_0$ and an element $a_0'\in A_0'$, called the \emph{\origin} of $\str{A}_0'$, such that
	 \begin{enumerate}[(b1)] 
	 	\item $\str{A}_0'$ has one element or every relation in $\cE_{tot}$ is total on it.\label{aone}
	 	
	 	\item $\fp(a_0')=a_0$.\label{atwo}
	 	
	 	\item For each $a'\in A_0'$ and each $i$, if the $i$-th witness for $\fp(a')$ lies in $A_0$ (that is $\str{A}_0\models\exists y\; W^i\fp(a')y$) then there exists a unique element $b'\in A_0'$ such that $\str{A}_0'\models W^i a'b'$. Otherwise there exists no such element. Denote $W_{a'}=\{b':\exists i\;\str{A}_0'\models W^i a'b'\}$ and for a tuple $\bar{a}$ let $W_{\bar{a}}=\bigcup_{a\in\bar{a}}W_{a}$. \label{atwohalf}
	 	
	 	\item For each $\bar{a}\subseteq A_0'$ satisfying $|\bar{a}|\leq t$ there exists a homomorphism $\fh:\str{W}_{\bar{a}}\to\str{A}_0$ such that for each $a\in\bar{a}$ we have $\str{A}_{\fp(a)}\cong\str{A}_{\fh(a)}$ and $\fh\restr W_a$ is an isomorphism (onto its image). Moreover, if $a_0'\in\bar{a}$ then we can choose $\fh$ so that $\fh(a_0')=a_0$.\label{athree}
	 	
	 	\item For each $a\in A_0'$ we have $\str{W}_a\cong\str{W}\restr A_0$ where $\str{W}$ is the $\phi$-witness structure for $\fp(a)$.
	 	 (Note that, by the definition of the $W^i$, each such isomorphism sends $a$ to $\fp(a)$.)
	 	\label{afour}
	 	
	 \end{enumerate}
\end{lemma}

Before we prove Lemma \ref{l:finitetr} let us observe that it indeed
allows us to build a particular finite model of $\phi$. 
Apply Lemma \ref{l:finitetr} to $\mathcal{E}_0=\mathcal{E}$ (which means that $\mathcal{E}_{tot}=\emptyset$)
and  $a_0$ being the root of $\str{A}$ (which means that $\str{A}_0=\str{A}$).
We use Lemma \ref{l:homomorphisms} to see that the obtained structure $\str{A}_0'$ is a model of $\phi$. Indeed, Condition (a1) of Lemma \ref{l:homomorphisms} follows directly from Condition (b\ref{afour}), as the structures $\str{W}_a$ from (b\ref{afour}) are full $\phi$-witness structures in this case, Condition (a2) is implied by Condition (b\ref{athree}) (since  $\bar{a}\subseteq\str{W}_{\bar{a}}$ and $\fh\restr W_a$ is an isomorphism, $\fh\restr\bar{a}$ preserves 1-types).

The proof of Lemma \ref{l:finitetr} goes by induction on  $l=|\mathcal{E}_0|/2$. 
In the base of induction, $l=0$, we have $\mathcal{E}_{tot}=\sigma_{\cDist}$. Without loss of generality we may assume that for each $a\in A$ the set $[a]_{\sigma_{\cDist}}$ has cardinality 1. If this is not the case, we simply add artificial transitive relations $T_{2k+1}$ and $T_{2k+2}$ and
interpret them as the identity in $\str{A}$ (this way they also satisfy the requirement of being each other's inverses).
We simply take $\str{A}'_0:=\str{A}_0= \str{A}\restr\{a_0\}$ and set $\fp(a_0)=a_0$. It is
readily verified that the conditions (b\ref{aone})--(b\ref{afour}) are then satisfied.

For the inductive step assume that Lemma \ref{l:finitetr}  holds for arbitrary closed under inverses set  
of size $2(l-1)$.
We show that then it holds for (closed under inverses) $\mathcal{E}_0$ of size $2l$. Without loss of generality 
we assume that $\mathcal{E}_0=\{T_1,\ldots,T_{2l}\}$. 
In the next two subsections we present a construction of $\str{A}_0'$ and then, in the following subsection, we argue that it is correct. Finally
we estimate the size of the produced models and establish the complexity of the finite satisfiability problem.

%THE CONSTRUCTION:
\subsubsection{Pattern components}

We plan to construct $\str{A}_0'$ out of basic building blocks called \emph{components}. Each component will be 
an isomorphic copy of some \emph{pattern component}. 
Let $\GGG[A_0]$ be the set of isomorphism types of subtrees of $\str{A}$ rooted at $A_0$.
Note that this way we overload the notation $\GGG[A_0]$ which in the two-variable case meant the set of 
generalized types realized in $\str{A}_0$.
For every $\gamma \in \GGG[A_0]$ we construct
a pattern component $\str{C}^\gamma$ as well as the \emph{extended pattern component} $\str{G}^{\gamma}$.
An important difference, compared to the two-variable case, is that components in this section will have more layers.

%PATTERN COMPONENTS: 
The extended pattern component, $\str{G}^\gamma$, is a finite structure whose universe
is divided into $2l(2t+1)$ \emph{inner layers}, 
$L_1,\ldots,L_{2l(2t+1)}$, and a single \emph{interface layer}, denoted $L_{2l(2t+1)+1}$.
The structure $\str{C}^{\gamma}$ is obtained by the restriction of $\str{G}^{\gamma}$ to its inner layers.
Each inner layer $L_i$ is further divided into  \emph{sublayers} $L_i^1,L_i^2,\ldots, L_i^{\hat{M}_\phi+1}$.
Additionally, in each sublayer $L_i^j$ its initial part $L_i^{j,init}$ is distinguished.
In particular, $L_1^{1,init}$ consists of a single element called the \emph{root}.
The interface layer $L_{2l(2t+1)+1}$ has no internal division but, for convenience, is sometimes
referred to as $L_{2l(2t+1)+1}^{1,init}$.
 The elements of $L_{2l(2t+1)}$ are called \emph{leaves} and 
the elements of $L_{2l(2t+1)+1}$ are called \emph{interface elements}. The structure of components is similar to 
the structure of components for the two-variable case, as depicted in Fig.~\ref{f:llayers}, just recall that the
number of layers is now bigger.

$\str{G}^\gamma$ (and $\str{C}^{\gamma}$) will have a shape resembling a tree, with  structures obtained by the inductive assumption as nodes, though
it will not be tree-like in the sense of Section \ref{s:treelike} (in particular, the internal structure of nodes may
be very complicated). All elements of $\str{G}^\gamma$, except for the interface elements, will have appropriate partial $\phi$-witness structures provided.

We remark that during the process of building a pattern component we do not yet apply the
transitive closure to the distinguished 
relations. Postponing this step is not important from the point of view
of the correctness of the construction, but will allow us for a simpler presentation of the proof of its correctness.
Given a pattern component $\str{C}$ we will sometimes denote by $\str{C}_+$ the structure obtained from $\str{C}$ by
applying all the appropriate transitive closures. 

The crucial property we want to enforce is that the root of $\str{C}^\gamma$ will be \emph{far} from its leaves in the following sense. 
Denote by $G_l(\str{S})$,  for a $\sigma$-structure $\str{S}$, the  Gaifman graph of the structure obtained by removing from $\str{S}$ 
the relations belonging to $\cE_{tot}$.
Then
there is no connected induced subgraph of $G_l(\str{C}^\gamma_+)$ of size $t$ containing an element of one of the first $2l$ layers
and, simultaneously, an element of one of the last $2l$  
layers of $\str{C}^\gamma$.

The role of every inner layer $L_i$ is, speaking informally, to \emph{kill} one of the ${T}_u$, that is to cause that
there will be no ${T}_u$-connections from $L_i$  to $L_{i+1}$. 
See the right part of Fig.~\ref{f:llayers}.
The role of sublayers, on the other hand, is to decrease the ${T}_u$-rank of elements. 
The purpose of the interface layer, $L_{2l(2t+1)+1}$, is to join the connect the component with other components.

Now we proceed to the construction of $\str{G}^{\gamma}$.
If $\gamma=\gamma_{a_0}$ is the type of
$\str{A}_{a_0}$
then take $a=a_0$; otherwise take any element $a \in A_0$ that is the root of a subtree of type $\gamma$.
Define $L_1^{1,init}=\{a'\}$ for a fresh $a'$, setting $\type{\str{G}^\gamma}{a'}=\type{\str{A}}{a}$ and $\fp(a')=a$. 

\smallskip\noindent
%ONE LAYER
\emph{Construction of an inner layer}: Let $1\leq i\leq 2l(2t+1)$. Assume we have defined layers $L_1,\ldots,L_{i-1}$, the initial part of sublayer $L_i^1$, $L_i^{1,init}$, and both the structure of $\str{G}^\gamma$ and the values of $\fp$ on $L_1 \cup \ldots \cup L_{i-1} \cup L_i^{1,init}$. 
Let $v=1+(i-1\mod 2l)$.
We are going to kill ${T}_v$. 
We now expand $L_i^{1, init}$ to full layer $L_i$.

\smallskip\noindent
%STEP 1 (SUBCOMPONENTS) 
\emph{Step 1: Subcomponents.} Assume that we have defined 
sublayers $L_i^{1}, \ldots,L_i^{j,init}$, and both the structure of $\str{G}^\gamma$ and the values of $\fp$ on
$L_1 \cup \ldots \cup L_{i-1} \cup L_i^{1} \cup \ldots \cup  L_i^{j,init}$. 
For each  $b\in L_i^{j,init}$ perform independently the following procedure. 
Apply the inductive assumption to $\fp(b)$ and the set $\mathcal{E}_0 \setminus\{T_v,T_v^{-1}\}$ obtaining a structure $\str{B}_0$, its \origin{}
$b_0$ and a function $\fp_b:B_0\to A_{\fp(b)}\cap[\fp(b)]_{\cE_{tot}\cup\{T_v,T_v^{-1}\}}\subseteq A_0$ with $\fp_b(b_0)=\fp(b)$. 
Identify $b_0$ with $b$ and add the remaining elements of $\str{B}_0$ to $L_i^j$, retaining the structure. 
Substructures $\str{B}_0$ of this kind will be called \emph{subcomponents} (note that all appropriate relations
are transitively closed in subcomponents).
Extend $\fp$ so that $\fp\restr B_0=\fp_b$. This finishes the
definition of $L_i^j$.

\smallskip\noindent
%STEP 2 (PROVIDING WITNESSES) 
\emph{Step 2: Providing witnesses.} For each $b\in L_i^j$ independently perform the following procedure. Let $\str{B}_0$ be the subcomponent created inductively in \emph{Step 1}, such that $b \in B_0$. Let $\str{W}$ be the $\phi$-witness structure for $\fp(b)$ in $\str{A}$. Let $\str{E}=\str{W}\restr[\fp(b)]_{\cE_{tot}}$ and $\str{F}=\str{W}\restr[\fp(b)]_{\cE_{tot}\cup\{T_v,T_v^{-1}\}}$. 
Note that $\str{F}$ is a substructure of $\str{E}$. By 
(b\ref{afour}) $b$ has the partial $\phi$-witness structure $\str{F}'$, isomorphic to $\str{F}$, provided in $\str{B}_0$. 
Extend $\str{F}'$ in $\str{G}^\gamma$ to an isomorphic copy $\str{E}'$ of $\str{E}$. 
The structure $\str{E}'$ will be the structure $\str{W}_b$ in $\str{G}^{\gamma}$ and then in $\str{A}_0'$.
The elements of $E' \setminus F'$ are fresh, and are assigned their
sublayers as follows.
For $c\in E'\setminus F'$ if $\str{E}'\models{T}_vbc$ (observe that in 
this case $\str{E}' \models \neg {T}_vcb$) then add $c$ to $L_i^{j+1,init}$, otherwise add $c$ to $L_{i+1}^{1,init}$. 
See Fig.~\ref{f:witnessestr}.
Take as the values of $\fp\restr (E'\setminus F')$ the corresponding elements of $E\setminus F$.

\begin{figure}
\begin{center}
\begin{tikzpicture}[scale=0.45]
\draw (6,8) -- (10,8) -- (11,4) -- (10,2) -- (6,2) -- (5,4) -- (6,8);
\draw[fill=gray!20] (5,4) -- (6,8) -- (7,4) -- (5,4);
\draw[fill=gray!10] (5,4) -- (0,4) -- (6,8) -- (5,4);

\fill[black] (4,4) circle (0.12);
\fill[black] (3,4) circle (0.12);
\fill[black] (2,4) circle (0.12);
\fill[black] (1,4) circle (0.12);

\fill[black] (6,8) circle (0.12);
\fill[black] (5.5,4) circle (0.12);
\fill[black] (6.5,4) circle (0.12);

\draw[->, thick] (6,8) -- (3.1,4.1);
\draw[->, thick] (6,8) -- (4.1,4.1);

%drzewo
\draw[fill=gray!20] (19,4) -- (20,8) -- (21,4) -- (19,4);
\draw[fill=gray!10] (19,4) -- (14,4) -- (20,8) -- (19,4);
\draw (20,8) -- (26,4);
\draw (26,4) -- (19,4);
\draw[dotted] (26,4)--(28,0);
\draw[dotted] (14,4)--(12,0);

\fill[black] (18,4) circle (0.12);
\fill[black] (17,4) circle (0.12);
\fill[black] (16,4) circle (0.12);
\fill[black] (15,4) circle (0.12);

\fill[black] (20,8) circle (0.12);
\fill[black] (19.5,4) circle (0.12);
\fill[black] (20.5,4) circle (0.12);

\draw[->, thick] (20,8) -- (17.1,4.1);
\draw[->, thick] (20,8) -- (18.1,4.1);

%wygieta strzlka
\draw[bend left=20, ->, very thin] (6.1,8.1) to (19.9,8.1);

\draw [
    %thick,
    decoration={
        brace,
        mirror,
        raise=0.5cm
    },
    decorate
] (0.6,4.7) -- (2.4,4.7);

\coordinate [label=center:$L_{i+1}^{1,init}$] (A) at ($(1.7,1)$);
\draw[->] (1.5,3.3) -- (1.5,2);

\draw [
    %thick,
    decoration={
        brace,
        mirror,
        raise=0.5cm
    },
    decorate
] (2.6,4.7) -- (4.4,4.7);
\draw[->] (3.5,3.3) -- (3.5,-0.4);
\coordinate [label=center:$L_i^{j+1,init}$] (A) at ($(4,-1.3)$);

\draw [
    %thick,
    decoration={
        brace,
        mirror,
        raise=0.5cm
    },
    decorate
] (5,2.5) -- (11,2.5);

\coordinate [label=center:in $L_i^j$] (A) at ($(8,-0.3)$);

%Napisy

\coordinate [label=center:${E}{\setminus}F$] (A) at ($(16.3,4.6)$); 

\coordinate [label=center:${{E}'{\setminus}{F}'}$] (A) at ($(2.3,4.6)$); 

\coordinate [label=center:$F$] (A) at ($(20,4.6)$); 

\coordinate [label=center:$F'$] (A) at ($(6,4.6)$); 

\coordinate [label=center:$\str{B}_0$] (A) at ($(9,4)$); 

\coordinate [label=center:$\str{W}$] (A) at ($(24,6.5)$);

\coordinate [label=center:$b$] (A) at ($(6,8.7)$); 
\coordinate [label=center:$\fp(b)$] (A) at ($(20,8.7)$); 

\coordinate [label=center:$\str{A}_{\fp(b)}$] (A) at ($(25,2.5)$);

\end{tikzpicture}
\caption{Providing witnesses. Thick arrows denote ${T}_u$-connections.}%
\label{f:witnessestr}%
\end{center}
\end{figure}
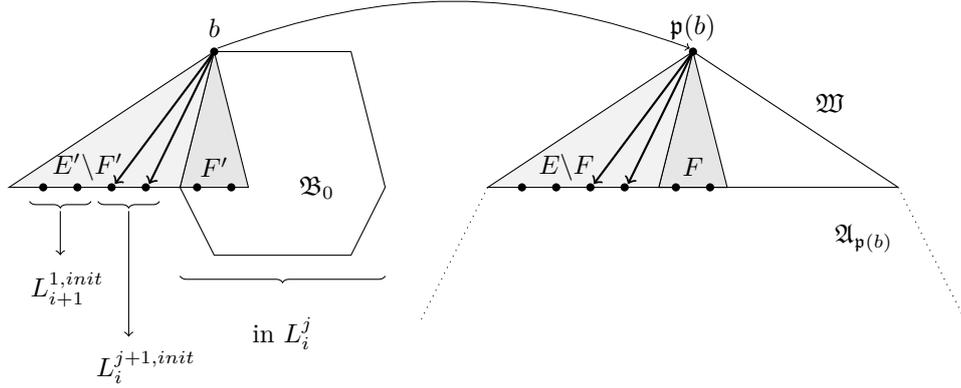

An attentive reader may be afraid that when adding witnesses for elements of the last sublayer $L_{i}^{\hat{M}_\phi+1}$
of $L_i$ we may want to add
one of them to the non-existing layer $L_{i}^{\hat{M}_\phi+2}$. 
There is however no such danger, which follows from the following claim. 

\begin{claim} \label{c:nodanger}
	(i) Let $b\in L_i^{j,init}$ and let $\str{B}_0$ be the subcomponent created for $b$ in \emph{Step 1}. Then for all $b'\in\str{B}_0$ we have $\rank{\str{A}}{v}{\fp(b)}\geq\rank{\str{A}}{v}{\fp(b')}$.
	(ii) Let $b\in L_i^j$ and let $\str{E}',\str{F}'$ be the partial $\phi$-witness structures for $b$ considered in \emph{Step 2}. Then for any $c\in E'\setminus F'$ such that $\str{E}'\models{T}_vbc$ (so $c\in L_i^{j+1}$) the inequality $\rank{\str{A}}{v}{\fp(b)}>\rank{\str{A}}{v}{\fp(c)}$ holds. 
\end{claim}

\begin{proof} 
	(i) By the inductive assumption applied to $\str{B}_0$, $\fp(b')=\fp_b(b')\in\str{A}_{\fp(b)}\restr[\fp(b)]_{\cE_{tot}\cup\{T_v\cup T_v^{-1}\}}$.
	This means that any downward-$T_v$-path (in $\str{A}$) starting in $\fp(b')$ can be extended to one starting in $\fp(b)$.
	Therefore $\rank{\str{A}}{v}{\fp(b)}\geq\rank{\str{A}}{v}{\fp(b')}$.
	
	\noindent
	(ii) By the choice of $\str{E}'$ and $\str{F}'$ we have $\str{E}'\models{T}_vbc\wedge\neg{T}_vcb$, thus by the choice of $\fp$ we have $\str{A}\models{T}_v\fp(b)\fp(c)\wedge\neg{T}_v\fp(c)\fp(b)$ and finally $\fr_v(\fp(b))>\fr_v(\fp(c))$. 
	\end{proof}
	
	Hence, when moving from $L_i^{j}$ to $L_i^{j+1}$ the ${T}_v$-ranks of pattern elements for the elements of these sublayers strictly decrease. Since these ranks are bounded by $\hat{M}_\phi$, then, even if the ${T}_v$-ranks of the patterns of some elements of $L_i^{1}$ are equal to $\hat{M}_\phi$, then, if $L_i^{\hat{M}_\phi+1}$ is non-empty,
	the ${T}_v$-ranks of the patterns of its elements must be $0$, which means that they cannot have witnesses connected to them one-directionally by ${T}_v$.

The construction of  $\str{G}^\gamma$ is finished when the interface layer, $L_{2l(2t+1)+1}$ is defined (recall that it has only
its `initial part'). 

\subsubsection{Joining the components}
 
In this section we take some number of copies of extended pattern components and arrange them into the desired structure $\str{A}'_0$,
identifying interface elements of some components with the roots of some other.
Some care is needed in this process in order
to avoid any modifications of the internal structure
of closures $\str{C}_+$ of components $\str{C}$, which could potentially result from transitivity of relations. 
In particular we need to ensure that if for some $u$ a pair of elements of a component $\str{C}$ is not connected by 
${T}_u$ inside $\str{C}$, then it will not become connected by a chain of ${T}_u$-edges external to $\str{C}$.

%JOINING PATTERN COMPONENTS: 
We create a pattern component $\str{C}^\gamma$ together with its extension $\str{G}^{\gamma}$ for every $\gamma \in \GGG[A_0]$. 
Let $max$ be the maximal number of interface elements across all the $\str{G}^{\gamma}$. For each $\str{G}^{\gamma}$ we 
number its interface elements arbitrarily using the numbers from $1$ up to, potentially, $max$.

For each $\gamma\in\GGG[A_0]$ we take copies $\str{G}^{\gamma,g}_{i,\gamma'}$ of $\str{G}^{\gamma}$ for $g\in\{0,1\}$ ($g$ is often called a \emph{color}; it is advised to think that if an extended component is of color $g$, then the elements of the inner layer of the component are of color $g$, while the elements of the interface layer are of color $1-g$ since the latter will become identified with some roots of components of color $1-g$), $1\leq i\leq max$ and $\gamma'\in\GGG[A_0]$, together with the previously chosen numbering of the interface elements. We also take an additional copy $\str{G}^{\gamma_{a_0},0}_{\bot,\bot}$ of $\str{G}^{\gamma_{a_0}}$. Its root will become the \origin{} of the whole $\str{A}_0'$.

 For each $\gamma$, $g$ consider components of the form $\str{G}^{\gamma,g}_{\cdotp,\cdotp}$. Perform the following
procedure for each $i$---the number of an interface element. 
Let $b$ be the $i$-th interface element of any such component, let  $\gamma'$ be the type of $\str{A}_{\fp(b)}$.
Identify the $i$-th interface elements of all $\str{G}^{\gamma,g}_{\cdotp,\cdotp}$
with the root $c_0$ of $\str{G}^{\gamma',1-g}_{i,\gamma}$. See Fig.~\ref{f:reductions}

Note that the values of $\fp(c_0)$ and $\fp(b)$ (the latter equals to the value of $\fp$ on the $i$-th interface element in all the $\str{G}^{\gamma,g}_{\cdotp,\cdotp}$) may differ. However, by construction,   $\str{A}_{\fp(b)}\cong\str{A}_{\fp(c_0)}$
(in particular, the 1-types of $b$ and $c_0$ match). For the element $c^*$ obtained in this identification step we 
define $\fp(c^*)=\fp(c_0)$.

For the extended component $\str{G}^{\gamma,g}_{i,\gamma'}$ denote $\str{C}^{\gamma,g}_{i,\gamma'}$ its restriction to its interface elements (being naturally a copy of the pattern component $\str{C}^{\gamma}$).
Define the graph of components used in the above construction, $G^{comp}$, by joining two components by an edge iff we identified an interface element of extension of one of them with the root of the other. Take $\str{A}_0^0$ as the structure restricted to the components accessible from $\str{C}^{\gamma_{a_0},0}_{\bot,\bot}$ in $G^{comp}$. 
Note that in $\str{A}_0^0$ we still do not take the transitive closures of relations. 
We define $\str{A}_0'$ by transitively closing all appropriate relations in $\str{A}_0^0$. 
Later we will keep using the convention of marking some auxiliary structures in which the transitive closures are not yet applied
with the superscript $0$.
Finally, we choose as the \origin{} $a_0'$ of $\str{A}_0'$ the root of the pattern component $\str{C}^{\gamma_0,0}_{\bot,\bot}$.  

\subsubsection{Correctness of the construction}
%THE PROOF:

%1: 
(b\ref{aone}) By the construction, after taking the transitive closures, on each (extended) pattern component either all $\cE_{tot}$ relations are total or it consists of one element.
Next observe that if two extended components get joined, then at least one of them has cardinality greater than $1$ and all $\cE_{tot}$ relations, after taking the transitive closure, are total on their sum.
So, by the definition of the graph of components $G^{comp}$, all $\cE_{tot}$ relations are total on $\str{A}_0'$ of it consists of a single element.

%2: 
\noindent
(b\ref{atwo}) As $a_0'$ we take the root of $\str{C}^{\gamma_{a_0},0}_{\bot,\bot}$. Recall that we explicitly map
the root of the pattern component $\str{C}^{\gamma_{a_0}}$  by $\fp$ to $a_0$.

\noindent
(b\ref{atwohalf}) The interpretations of the $W^i$ are defined in the step of providing witnesses where, implicitly, we take
care of this condition for every element $a'$ of the inner layers by extending the fragment of the partial $\phi$-witness structure
 for $a'$
 created on the previous level of induction
by a copy of a further fragment of the same pattern  $\phi$-witness structure.  Note also that during the step of joining the components all the interface elements become identified with some roots, which are elements of inner layers, and that the identifications 
do not spoil the required property.

\noindent
%3:
(b\ref{athree})
 For simplicity, let us ignore the `moreover' part of this condition for some time. We will explain how to take care of it near the end of this proof.
  Now we find a homomorphism $\fh$ such that 
$\str{A}_{\fp(a)}\cong\str{A}_{\fh(a)}$ for all $a\in\bar{a}$ (we say that such a homomorphism has \emph{the subtree isomorphism property}). 
Later we will show that its restrictions to the substructures $\str{W}_a$  are indeed isomorphisms. 
The proof starts with several homomorphic reductions which show that instead of $\str{A}_0'$ we can consider a structure looking like a pattern component but twice as high. 

\smallskip
\noindent
%RED 0: 
\emph{Reduction 0.}
First observe that if $\str{A}_0$ consists of just one element, then the structure $\str{A}_0'$ consists of one element and the only map $\fh:A_0'\to A_0$ is the required homomorphism. For the rest of the proof of (b\ref{athree}) we assume the $A_0$ has cardinality at least $2$, in particular all $\cE_{tot}$ relations are total on it.
Consider a tuple $\bar{a}\subseteq {A}_0'$ such that $|\bar{a}|\leq t$. 
 Observe that for each $a\in\bar{a}$ the structure $\str{W}_a$ is connected in $G_l(\str{A}_0'\restr W_{\bar{a}})$ (recall the definition of Gaifman graph $G_l(\str{S})$ and the interpretation of the symbols $W^i$).
Let $\str{W}_{\bar{a}_1},\ldots,\str{W}_{\bar{a}_K}$ be the connected components of $\str{W}_{\bar{a}}$ in $G_l(\str{A}_0'\restr W_{\bar{a}})$.
If we have homomorphisms $\fh_i:\str{W}_{\bar{a}_i}\to\str{A}_0$ satisfying the subtree isomorphism property then we can take $\fh=\bigcup\fh_i:\str{W}_{\bar{a}}\to\str{A}_0$ which is a homomorphism, 
since all $\cE_{tot}$ relations are total on $\str{A}_0$,
that still has the subtree isomorphism property. Owing to this reduction \emph{we can restrict attention to tuples} $\bar{a}$ \emph{with} $\str{W}_{\bar{a}}$ \emph{connected  (in the above sense)}.

\smallskip\noindent
%RED 1: 
\emph{Reduction 1.}  
By the construction, for all $1 \le i \le 2l(2t+1)$ and $v=1+(i-1\mod 2l)$, there is no ${T}_v$-path in any component from an element of $L_i$ to an element of $L_{i+1}$. 
Thus, if we divide (inner)
layers of components into groups of size $2l$, a transitive path may join at most elements of two neighboring groups.
Obviously,
non-transitive relations join only tuples consisting of elements of at most two consecutive layers, and, in particular,  each of the $\str{W}_a$ lies in at most two consecutive layers. It follows,
that given a connected $\str{W}_{\bar{a}}$, $|\bar{a}| \le t$, by our choice of the number of layers in a component, 
 there exists  $g\in\{0,1\}$ such that 
removing all the connections between leaves of color $1-g$ and roots of color $g$ (in other words: any connections between
elements of $L_{2l(2t+1)}$ and elements of $L_{2l(2t+1)+1}$ in components of color $1-g$) does not remove any connections
among the elements of $\str{W}_{\bar{a}}$.

More formally, let $\str{D}_0^0$ be the structure obtained from $\str{A}_0^0$ by removing all the connections as described above, 
and let $\str{D}_0'$ be the transitive closure of $\str{D}_0^0$.
Then the inclusion map $\iota:\str{W}_{\bar{a}}\to\str{D}_0'$ is a homomorphism.   Clearly $\str{A}_{\fp(\iota(a))}=\str{A}_{\fp(a)}$ since $a=\iota(a)$. 
Thus \emph{we can restrict attention to a tuple $\bar{a}$ for which  $\str{W}_{\bar{a}}$ is connected and search for a homomorphism $\str{W}_{\bar{a}} \rightarrow \str{A}_0$ treating $\str{W}_{\bar{a}}$ as a substructure of} $\str{D}_0'$.

\smallskip\noindent
%RED 2: 
\emph{Reduction 2.} Observe that by our scheme of arranging the copies of pattern components there is at most one type $\gamma$ of components of color $g$ (where $g$ is the color from the previous reduction) that contains some element of a connected $\str{W}_{\bar{a}}$ (consider the shape of a connected fragment of the graph of components $G^{comp}$ with connections between leaves of color $g$ and roots of color $1-g$ removed). Furthermore, all elements of $\bar{a}$ of color $1-g$ are contained in components of the form $\str{C}^{\cdotp,1-g}_{\cdotp,\gamma}$. Choose one component of type $\gamma$ of color $g$ and call it $\str{C}^{\gamma}$. 
Consider the structure $\str{E}_0^0$ (resp.~$\str{F}_0^0$) obtained as the restriction of the structure $\str{D}_0^0$ from the previous reduction to the union of the domains of the components of the form $\str{C}^{\gamma,g}_{\cdotp,\cdotp}$ (resp.~the domain of $\str{C}^{\gamma}$) and the domains of all the components of the form $\str{C}^{\cdotp,1-g}_{\cdotp,\gamma}$. Let $\str{E}_0'$ (resp. $\str{F}_0'$) be their transitive closures.  Consider a projection $\pi$ that projects all elements of $\str{E}_0^0$ of color $g$ onto $\str{C}^{\gamma}$ and is the identity on the others. We claim that $\pi\restr W_{\bar{a}}:\str{W}_{\bar{a}}\to\str{F}_0'$ is a homomorphism. To see this, observe that the paths connecting elements of $\str{W}_{\bar{a}}$ in $\str{D}_0^0$ are contained in $E_0^0$ and $\pi:\str{E}_0^0\to\str{F}_0^0$ is a homomorphism. 
See Fig.~\ref{f:reductions}.
 Clearly for each $a\in \bar{a}$ we have $\str{A}_{\fp(a)}=\str{A}_{\fp(\pi(a))}$ since $\fp(a)=\fp(\pi(a))$.
Thus, finally, \emph{we can restrict attention to a tuple $\bar{a}$ for which $\str{W}_{\bar{a}}$ is connected and search for a homomorphism $\str{W}_{\bar{a}} \rightarrow \str{A}_0$ treating $\str{W}_{\bar{a}}$ as a substructure of $\str{F}_0'$}.

\begin{figure}  
\begin{center}
	\begin{tikzpicture}[scale=0.75]

%lower triangles
\foreach \x in {1,7,13}{
\draw[color=gray, fill=red!20] (\x-0.2,-0.5) -- (\x+4+0.2,-0.5) -- (\x+2,3) -- (\x-0.2,-0.5);
\fill[black] (\x+2,2.6) circle (0.07);
\draw[color=gray, fill=blue!40] (\x-0.2,-0.5) -- (\x+4+0.2, -0.5) -- (\x+4+0.2-0.32,0) -- (\x-0.2+0.32, 0) -- (\x-0.2,-0.5);
}

%upper triangles
\foreach \x in {4,10}{
\draw[color=gray, fill=blue!40] (\x-0.2,4.5) -- (\x+4+0.2,4.5) -- (\x+2,8) -- (\x-0.2,4.5);
\draw[color=gray, fill=gray!20] (\x+0.1,5) -- (\x+4-0.1,5);
\draw[color=gray, fill=red!20] (\x-0.2,4.5) -- (\x+4+0.2, 4.5) -- (\x+4+0.2-0.32,5) -- (\x-0.2+0.32, 5) -- (\x-0.2,4.5);
}

\coordinate [label=center:$\cong$] (A) at ($(9,6.3)$); 
\coordinate [label=center:$\str{F}_0'$] (A) at ($(2.2,4.8)$); 
\coordinate [label=center:$g$] (A) at ($(3.5,6.2)$); 
\coordinate [label=center:$1{-}g$] (A) at ($(0.4,1.2)$); 
\coordinate [label=center:$\str{C}^{\gamma}$] (A) at ($(6,6.9)$); 
\coordinate [label=center:$\pi$] (A) at ($(9,7.7)$);

\draw[->] (2.6, 4.5) -- (3,4);

%elements
\foreach \x in {4.5,6,7.5 , 10.5,12, 13.5}{
\fill[black] (\x,4.75) circle (0.07);
}

\coordinate [label=center:$b_1$] (A) at ($(4.8,5.25)$); 
\coordinate [label=center:$b_2$] (A) at ($(6,5.25)$); 
\coordinate [label=center:$b_3$] (A) at ($(7.3,5.25)$); 

\coordinate [label=center:$c_0$] (A) at ($(3,2.2)$);

%identifications
\draw[dashed] (4.5,4.75) -- (3,2.6) -- (10.5,4.75);
\draw[dashed] (6,4.75) -- (9,2.6) -- (12,4.75);
\draw[dashed] (7.5,4.75) -- (15,2.6) -- (13.5,4.75);

%snakes
\draw[decorate, decoration={snake, segment length=10, amplitude=1}] (4.9,6.2) -- (7.1,6.2);
\draw[decorate, decoration={snake, segment length=10, amplitude=1}] (10.9,6.2) -- (13.1,6.2);
\draw[decorate, decoration={snake, segment length=10, amplitude=1}] (1.9,1.2) -- (4.1,1.2);
\draw[decorate, decoration={snake, segment length=10, amplitude=1}] (7.9,1.2) -- (10.1,1.2);
\draw[decorate, decoration={snake, segment length=10, amplitude=1}] (13.9,1.2) -- (16.1,1.2);

\draw [decorate,decoration={brace,amplitude=10pt},rotate=0] (14,6.2) -- (17,1.2) node [black,midway, xshift=17, yshift=5] {$\;\;\;\supseteq \str{W}_{\bar{a}}$};

\draw[dotted] (18,-1) -- (0,-1) -- (6,9) -- (9,4) -- (15,4) -- (18,-1);
\draw[loosely dotted] (18.3,-1.2) -- (-0.3, -1.2) -- (6,9.3) -- (12,9.3) -- (18.3,-1.2); 
\coordinate [label=center:$\str{E}_0'$] (A) at ($(1.2,3.8)$);  
\draw[->] (1.6, 3.5) -- (2.1,2.9);

%homomorphism

%wygieta strzlka
\draw[bend right=10, ->,  dashed] (10.5,7.2) to (7.5,7.2);

\end{tikzpicture}
	\caption{Joining the components and Reductions 1 and 2. Elements connected by dashed lines are identified.} \label{f:reductions}
	\end{center}
\end{figure}
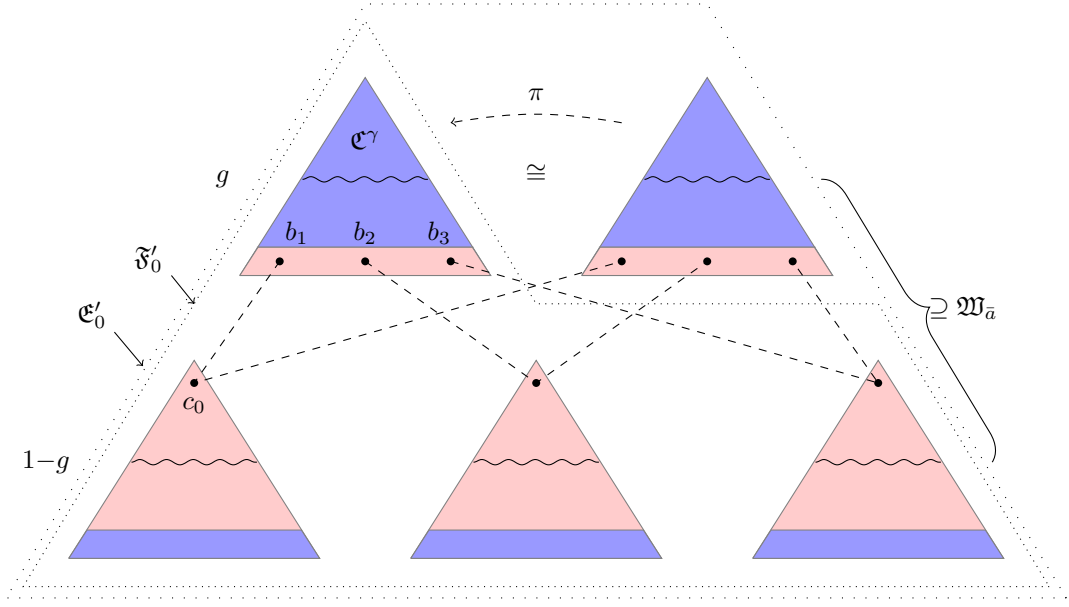

\smallskip\noindent
%FINAL STEP: 
\emph{Essential homomorphism construction.} 
Note that $\str{F}_0^0$ looks like a single component but is twice as high. Consider the \emph{tree of subcomponents} of $\str{F}_0^0$, $\tau$,
defined as follows: make a subcomponent $\str{B}$ the parent of $\str{B}'$ if $\str{B}'$ contains a witness for an element of 
$\str{B}$. Observe that so obtained $\tau$ is indeed a tree.
For a subcomponent $\str{B}\in\tau$ denote by $B^{\wedge}$ 
the union of 
domains of 
subcomponents belonging to the subtree of $\tau$ rooted at $\str{B}$.

Since we might have cut some connections between an element and some of its witnesses during Reduction 1, we define for each $a\in F_0'$ the surviving part $\str{V}_a$ of $\str{W}_a$ by $\str{V}_a=\str{F}_0'\restr V_a$ where $V_a=\{b:\exists i\;\str{F}_0'\models W^iab\}$. For a tuple $\bar{b}$ denote $V_{\bar{b}}=\bigcup_{b\in\bar{b}} V_b$ and $\str{V}_{\bar{b}}=\str{F}_0'\restr V_{\bar{b}}$. Note that $V_a\subseteq W_a$, and generally, this inclusion may be strict,
but for all $a\in\bar{a}$ we have $\str{V}_a = \str{W}_a$, and thus, in particular, the claim below finishes the proof of the currently considered part of (b\ref{athree}), that is the
proof of the existence of a homomorphism satisfying the subtree isomorphism property.

Returning to the shape of $\str{F}_0^0$, it consists of some subcomponents arranged into tree $\tau$ glued together by the structure on the surviving parts of witness structures. Note that all such building blocks (that is both the subcomponents and the surviving parts of the partial witness structures) are transitively closed. Moreover, by the tree structure of $\tau$, if some elements of such a  building block are connected by some atom in $\str{F}_0'$, then they already have been connected by the same atom in $\str{F}_0^0$, therefore  the identity map from $\str{F}_0^0$ to $\str{F}_0'$ acts as an isomorphism when restricted to such a building block. 

\begin{claim}\label{c:joiningtr}
	 For every subcomponent $\str{B}_0\in\tau$ with \origin{} $b_0$, and  $\bar{a}\subseteq B_0^{\wedge}$, $|\bar{a}|\leq t$, there exists a homomorphism $\fh:\str{V}_{\bar{a}}\to\str{A}_{\fp(b_0)}\restr[\fp(b_0)]_{\cE_{tot}}$ 
	 	 such that for all $a\in\bar{a}$ we have $\str{A}_{\fh(a)}\cong\str{A}_{\fp(a)}$, 
	and if $b_0\in\bar{a}$ then $\fh(b_0)=\fp(b_0)$.
\end{claim}
\begin{proof}
Bottom-up induction over tree.

%BASE: 
\smallskip\noindent
\textit{Induction base:} If $\str{B}_0$ is a leaf of $\tau$ then $\str{V}_{\bar{a}}\subseteq\str{B}_0$ and the claim follows by the inductive assumption of Lemma \ref{l:finitetr} (note that here we implicitly use the fact that the identity map is an isomorphism between $\str{F}_0^0\restr B_0$ and $\str{F}_0'\restr B_0$).

\smallskip\noindent
\textit{Induction step:}  Let $\str{B}_1,\ldots,\str{B}_K$ be the list of all the children of $\str{B}_0$ in $\tau$ such that $B_i^{\wedge}$ contains some element of $\bar{a}$.  
If $K=1$ and $\bar{a}\subseteq B_1^{\wedge}$ the thesis follows from the inductive assumption of this claim. 

Otherwise, for $1 \le i \le K$ (note that it is possible that $K=0$), denote by $b_i$ the \origin{} of $\str{B}_i$ and let $c_i\in\str{B}_0$ be such that $b_i$ is a witness chosen by $c_i$ in the step of providing witnesses or during the step of joining the components. By the inductive assumption of this claim there exist homomorphisms $\fh_i:\str{V}_{(\bar{a}\cap B_i^{\wedge})b_i}\to\str{A}_{\fp(b_i)}\restr[\fp(b_i)]_{\cE_{tot}}$
such that $\fh_i(b_i)=\fp(b_i)$.
From the inductive assumption of Lemma \ref{l:finitetr} we have a homomorphism $\fh_0:\str{V}_{(\bar{a}\cap B_0)c_1\ldots c_K}\restr B_0\to\str{A}_{\fp(b_0)}\restr[\fp(b_0)]_{\cE_{tot}}$. 
We extend it in the only possible way to $\fh_0^*$ defined on the whole $\str{V}_{(\bar{a}\cap B_0)c_1\ldots c_K}$: for each $a\in\bar{a}$ and $c\in V_a\setminus B_0$ (by construction $\str{V}_a\models W^iac$ for some $i$) we set $\fh(c)$ to be the only element satisfying $\str{A}_0\models W^i\fh(a)\fh(c)$ (such an element exists since $\str{A}_{\fh(a)}\cong\str{A}_{\fp(a)}$---in particular the $\phi$-witness structures of $\fh(a)$ and $\fp(a)$ are isomorphic). 
Note that the sizes of the tuples used to build the homomorphisms $\fh_i$ are bounded by $t$, as required.

We construct from the above maps a homomorphism $\fh:\str{V}_{\bar{a}b_1\ldots b_Kc_1\ldots c_K}\to \str{A}_{\fp(b_0)}\restr[\fp(b_0)]_{\cE_{tot}}$. 
See Fig.~\ref{f:joininghom}. Here, the crucial property is that $\str{A}$ has  a regular shape.
Indeed, for each $i$, $1 \le i \le K$, 
for the witness $b_i'$ for $\fh_0(c_i)$ (that is an element $b_i'$ satisfying $\str{A}_0\models W^j\fh_0(c_i)b_i'$ for the appropriate  $j$), corresponding to the witness $b_i$ for $c_i$ (such that  $\str{F}_0'\models W^jc_ib_i$), we have $\str{A}_{b_i'}\cong\str{A}_{\fp(b_i)}$. This is the case since, 
by the inductive assumption of Lemma \ref{l:finitetr} we have $\str{A}_{\fp(c_i)}\cong\str{A}_{\fh_0(c_i)}$, by construction we have $\str{A}_{b_i^{\fp}}\cong\str{A}_{\fp(b_i)}$ (here $b_i^{\fp}$ is the $j$-th witness of $\fp(c_i)$; note that during the step of providing witnesses we set $\fp(b_i)=b_i^{\fp}$; we need to consider $\cong$ since the value
of $\fp(b_i)$ may change due to a possible identification applied in the step of joining the components) and the 
numbering of witnesses is preserved by subtree isomorphisms.
Thus there is a homomorphism  $\fh_i^*:\str{V}_{(\bar{a}\cap B_i^{\wedge})b_i} \to \str{A}_{b_i'}$ with $\fh_i^*(b_i)=b_i'$.

We naturally join $\fh_0^*, \fh_1^*, \ldots, \fh_K^*$ into $\fh$:  
$\fh=\bigcup\fh_i^*$.
Note that such $\fh$ is well defined, even though the value of $\fh$ on each of the $b_i$ is defined twice, since $b_i$ belongs
to both $\dom \fh^*_0$ and $\dom \fh_i^*$ ($\fh$ has been defined on the other elements exactly once). 
 For each $a\in\dom\fh_i$ ($=\dom \fh_i^*$, when $i>0$) we have 
$\str{A}_{\fh(a)}=\str{A}_{\fh_i^*(a)}\cong\str{A}_{\fh_i(a)}$($\cong\str{A}_{\fp(a)}$,
by the inductive assumptions of this claim and Lemma \ref{l:finitetr}). Since $\bar{a}\subseteq \dom\fh_0\cup\bigcup_{i>0}\dom\fh_i^*$, we can conclude that for each $a\in\bar{a}$ we have $\str{A}_{\fp(a)}\cong\str{A}_{\fh(a)}$.

The fact that $\fh$ is a homomorphism follows from the tree structure of $\tau$. In particular, there cannot be any connections (before taking the transitive closures) between (non-\origin{}) elements of two different
$B_i^{\wedge}$ (for $1 \le i \le K$).
The full proof that $\fh$ is a homomorphism is tedious, therefore we show two representative cases that use all the major ideas required.
First, consider $a,a'\in V_{\bar{a}b_1\ldots b_Kc_1\ldots c_K}$
such that $a\in B_i^{\wedge}$, $a'\in B_j^{\wedge}$ for some $i,j$ such that $c_i\neq c_j$. Assume that $\str{F}_0'\models T_uaa'$ for some $u$.  We will prove that $\str{A}_0\models\ T_u\fh(a)\fh(a')$. By a standard argument, owing to the tree structure of $\tau$ (some more care is needed since there may be some connections in the structures $\str{V}_b$), there exist $d_i\in V_{c_i}\cap B_0$ and $d_j\in V_{c_j}\cap B_0$ such that $\str{F}_0'\models T_uab_i$, $\str{F}_0^0\restr V_{c_i}\models T_ub_id_i$, $\str{B}_0\models T_ud_id_j$, $\str{F}_0^0\restr V_{c_j}\models T_ud_jb_j$ and $\str{F}_0'\models T_ub_ja'$ (we assumed that $a,b_i,d_i,d_j,b_j,a'$ are pairwise different; otherwise some parts of such path become trivial). Since $\fh_i^*$, $\fh_0$ and $\fh_j^*$ are homomorphisms, $\str{A}_0\models T_u\fh(a)\fh(b_i)\wedge T_u\fh(d_i)\fh(d_j)\wedge\ T_u\fh(b_j)\fh(a')$. Now we show that $\str{A}_0\models T_u\fh(b_i)\fh(d_i)$. By construction of $\str{V}_a$, there exist indices $i_b$ and $i_d$ such that $\str{F}_0'\models W^{i_b}c_ib_i\wedge W^{i_d}c_id_i$ and therefore by the choice of $\fh_0$, $\str{A}_0\models W^{i_d}\fh(c_i)\fh(d_i)$ and by the choice of the extension of $\fh_0$ to $\fh_0^*$, $\str{A}_0\models W^{i_b}\fh(c_i)\fh(b_i)$. Let $b_i^{\fp}$ be the $i_b$-th witness of $\fp(c_i)$ and $d_i^{\fp}$ be the $i_d$-th witness of $\fp(c_i)$. By construction, $\str{A}_0\models T_ub_i^{\fp}d_i^{\fp}$. But $\str{A}_{\fp(c_i)}\cong\str{A}_{\fh(c_i)}$ and by the uniqueness of the numbers of the witnesses, any isomorphism between these subtrees sends $b_i^{\fp}$ to $\fh(b_i)$ and $d_i^{\fp}$ to $\fh(d_i)$, therefore $\str{A}_0\models T_u\fh(b_i)\fh(d_i)$. Similarly $\str{A}_0\models T_u\fh(d_j)\fh(b_j)$. Joining the pieces together, by transitivity of $T_u$, $\str{A}_0\models T_u\fh(a)\fh(a')$.
Secondly, we consider the case when $\str{F}_0'\models R(\bar{a}')$ for some non-transitive symbol $R$ and $\bar{a}'\subseteq V_{\bar{a}b_1\ldots b_Kc_1\ldots c_K}$. 
By construction, $R(\bar{a}')$ was set either during in the process of building some subcomponent or during the step of providing witnesses. Thus $\bar{a}'$ is either contained in $B_0$ or one of the $B_i^{\wedge}$ or one of the $V_a$
for some $a\in\bar{a}b_1\ldots b_Kc_1\ldots c_K$. 
Now we can prove, using arguments similar to ones used for appropriate parts of the path in the previous case, that $\str{A}_0\models R(\fh(\bar{a}'))$.

Since by construction $\fh_0\subseteq\fh$, if $b_0\in\bar{a}$ then $\fh(b_0)=\fh_0(b_0)$($=\fp(b_0)$ by the inductive assumption of Lemma \ref{l:finitetr}). To finish the inductive step, we restrict $\fh$ to $V_{\bar{a}}$. 
\end{proof}

	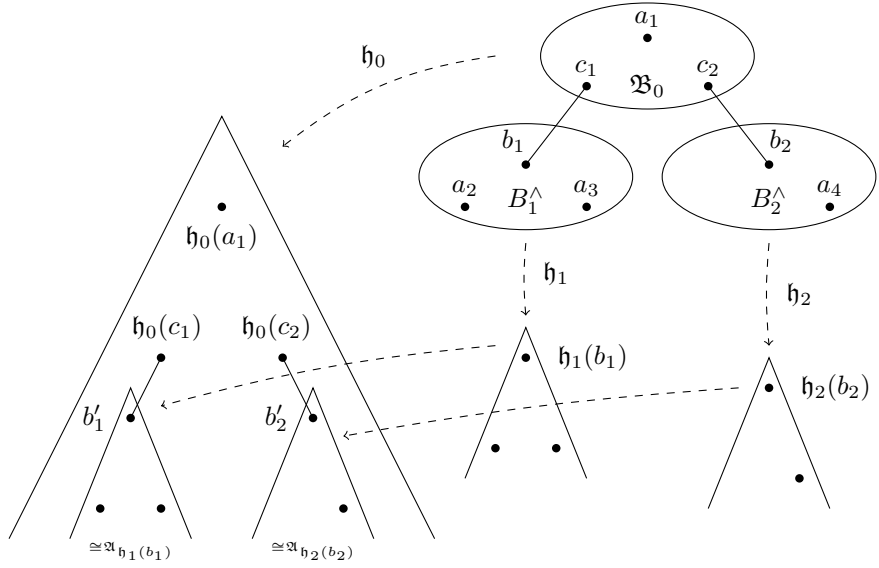
\begin{figure}  
	\begin{center}
		
		\begin{tikzpicture}[scale=0.8]
		
		%elipse upper
		\draw (10,10) ellipse (50pt and 25pt);
		\fill[black] (10,10.3) circle (0.07);
		\coordinate [label=center:$a_1$] (A) at ($(10,10.6)$); 
		\fill[black] (9,9.5) circle (0.07);
		\coordinate [label=center:$c_1$] (A) at ($(9,9.8)$); 
		\fill[black] (11,9.5) circle (0.07);
		\coordinate [label=center:$c_2$] (A) at ($(11,9.8)$); 
		\coordinate [label=center:$\str{B}_0$] (A) at ($(10,9.5)$); 
		
		%wintess str
		\draw (9,9.5) -- (8,8.2);
		\draw (11,9.5) -- (12,8.2);
		\draw (1.5,4) -- (2,5);
		\draw (4.5,4) -- (4,5);
		\fill[black] (2,5) circle (0.07);
		\fill[black] (4,5) circle (0.07);
		\coordinate [label=center:$\fh_0(c_1)$] (A) at ($(2.1,5.5)$); 
		\coordinate [label=center:$\fh_0(c_2)$] (A) at ($(3.9,5.5)$);

		%elipsse left
		\draw (8,8) ellipse (50pt and 25pt);
		\fill[black] (8,8.2) circle (0.07);
		\coordinate [label=center:$b_1$] (A) at ($(7.8,8.6)$); 
		\fill[black] (7,7.5) circle (0.07);
		\coordinate [label=center:$a_2$] (A) at ($(7,7.8)$); 
		\fill[black] (9,7.5) circle (0.07);
		\coordinate [label=center:$a_3$] (A) at ($(9,7.8)$); 
		\coordinate [label=center:${B}_1^{\wedge}$] (A) at ($(8,7.6)$); 
		
		%ellipse right
		\draw (12,8) ellipse (50pt and 25pt);
		\fill[black] (12,8.2) circle (0.07);
		\coordinate [label=center:$b_2$] (A) at ($(12.2,8.6)$); 
		\fill[black] (13,7.5) circle (0.07);
		\coordinate [label=center:$a_4$] (A) at ($(13,7.8)$); 
		\coordinate [label=center:${B}_2^{\wedge}$] (A) at ($(12,7.6)$); 
		
		%triangle 1
		\draw (7,3) -- (8,5.5) -- (9,3);
		\fill[black] (8,5) circle (0.07);
		\coordinate [label=center:$\fh_1(b_1)$] (A) at ($(9.1,5)$); 
		\fill[black] (7.5,3.5) circle (0.07);
		\fill[black] (8.5,3.5) circle (0.07);
		
		%triangle 2
		\draw (11,2.5) -- (12,5.0) -- (13,2.5) ;
		\fill[black] (12,4.5) circle (0.07);
		\coordinate [label=center:$\fh_2(b_2)$] (A) at ($(13.1,4.5)$); 
		\fill[black] (12.5,3) circle (0.07);

		%big triangle
		\draw (-0.5,2) -- (3,9) -- (6.5,2);
		\fill[black] (3,7.5) circle (0.07);
		\coordinate [label=center:$\fh_0(a_1)$] (A) at ($(3,7)$);

		%triangle 1'
		\draw (0.5,2) -- (1.5,4.5) -- (2.5,2) ;
		\fill[black] (1.5,4) circle (0.07);
		\fill[black] (1,2.5) circle (0.07);
		\fill[black] (2,2.5) circle (0.07);
		\coordinate [label=center:${\sss \cong \str{A}_{\fh_1(b_1)}}$] (A) at ($(1.5,1.8)$); 
		\coordinate [label=center:$b_1'$] (A) at ($(0.9,4)$);

		%triangle 2'
		\draw (3.5,2) -- (4.5,4.5) -- (5.5,2) ;
		\fill[black] (4.5,4) circle (0.07);
		\fill[black] (5,2.5) circle (0.07);
		\coordinate [label=center:${\sss \cong \str{A}_{\fh_2(b_2)}}$] (A) at ($(4.5,1.8)$); 
		\coordinate [label=center:$b_2'$] (A) at ($(3.9,4)$);

		%tiny triangle 1
		
		%tiny triangle 2
		
		%bent arrows
		
		\coordinate [label=center:$\fh_0$] (A) at ($(5.5,10)$); 
		\draw[bend right=15, ->, dashed] (7.5,10) to (4,8.5);	
		
		\draw[bend right=5, ->, dashed] (8,6.9) to (8,5.7);	
		\coordinate [label=center:$\fh_1$] (A) at ($(8.5,6.4)$); 
		
		\draw[bend right=5, ->, dashed] (12,6.9) to (12,5.2);	
		\coordinate [label=center:$\fh_2$] (A) at ($(12.5,6.05)$);

		\draw[bend right=5, ->, dashed] (7.5,5.2) to (2,4.2);	
		
		\draw[bend right=3, ->, dashed] (11.5,4.5) to (5,3.7);	
		
		\end{tikzpicture}
		\caption{Joining homomorphisms} \label{f:joininghom}
		\end{center}
	\end{figure}

Now we prove the additional property required for $\fh$ by (b\ref{athree}), that is, that for each $a\in\bar{a}$, $\fh\restr W_a$ is an isomorphism.
By the 
numbering of witnesses, as explained before the statement of this lemma, $\fh$ moves $W_a$ into the part of the witness structure of $\fh(a)$ contained in $A_0$ and is one-to-one by the uniqueness of the numbers of witnesses in a witness structure.
The other way around, we can use a similar argument as in the first case presented in the proof that the map built in Claim \ref{c:joiningtr} is a homomorphism. That is, if for some $\bar{a}'\subseteq W_a$ and some (arbitrary) relation $R$, $\str{A}_0\models R(\fh(\bar{a}'))$, then, since $\str{A}_{\fh(a)}\cong\str{A}_{\fp(a)}$ and any isomorphism preserves the 
numbering of witnesses and the structure on $W_a$ was copied from a part of the witness structure for $\fp(a)$ (together with such 
numbering), $\str{A}_0'\models R(\bar{a}')$ and therefore the inverse of $\fh\restr\str{W}_a$ is also a homomorphism, so $\fh\restr\str{W}_a$ is an isomorphism.

Now we return to the `moreover' part of (b\ref{athree}). Let us assume that $a_0'\in\bar{a}$. We will slightly modify the above proof. Reductions 0 and 1 do not move $a_0'$ and we keep them unchanged. Notice that in Reduction 1 we have that $g=0$. Now, in Reduction 2 we have that $\gamma=\gamma_{a_0}$ and we choose $\str{C}^{\gamma}=\str{C}^{\gamma,0}_{\bot,\bot}$. This way application of $\pi$ does not move $a_0'$. To finish the proof, it is sufficient to see that by Claim \ref{c:joiningtr} $\fh(a_0')=\fp(a_0')=a_0$. 

\smallskip\noindent
%4: 
(b\ref{afour}) Apply (b\ref{athree}) to a tuple consisting of just $a$ to obtain an isomorphism $\fh:\str{W}_a\to\str{A}_0\restr\fh(W_a)$ and then apply an isomorphism between $\str{A}_{\fh(a)}$ and $\str{A}_{\fp(a)}$.

\subsubsection{Size of models and complexity}

To complete the proof of Thm.~\ref{t:maintr} we need to show an appropriate upper bound on 
the size of finite models produced by our construction.
The following routine estimation shows that $|A_0'|$ is triply exponential in $n=|\phi|$, regardless of the choice of the initial 
tree-like model $\str{A}$.
We calculate a bound $S_{2l}$
on the size of the structure obtained in the proof of Lemma \ref{l:finitetr} for $|\mathcal{E}_0|=2l$. We are interested in $ S_{2k+2}$, which is the desired bound on the size of $\str{A}_0'$ (we use $ S_{2k+2}$ here, rather than $ S_{2k}$, because we may potentially
introduce the auxiliary identity relation in the base step of induction).
By the construction any pattern component is a tree of subcomponents consisting of at most $2l(2t+1)(\hat{M}_{\phi}+1)$ sublayers (so, also this is a bound on the depth of the tree) . 
In the 
sublayer of depth $1$
we have at most $ S_{2l-2}$ elements, in the 
sublayers in the 
second one---at most $ S_{2l-2}n$ subcomponents; this jointly gives $ S_{2l-2}^2n$ elements. Iterating, we have at most $ S_{2l-2}^in^{i-1}$ elements in the sublayers of depth $i$, which jointly gives an estimate $( S_{2l-2}n)^{2l(2t+1)(\hat{M}_{\phi}+1)+1}$ on both the number of inner elements and the number of interface elements
in a pattern component. Multiplying it by the number of components used in the joining phase, and then
estimating $t$ and $l$ in the exponent by $n$ and $n+1$ respectively, 
we get a bound 
$ S_{2l}=2|\GGG[A]|^2( S_{2l-2}n)^{4(n+1)(2n+1)(\hat{M}_{\phi}+1)+2}$.
Solving this recurrence relation, and recalling that $\hat{M}_{\phi}$ and $|\GGG[A]|$ are doubly exponential in $|\phi|$ we obtain a triply exponential bound on $ S_{2k+2}$.

This finishes the proof of Thm.~\ref{t:maintr}.
We do not know if our construction is optimal with respect to the size of models. The best we can do for the lower bound is to enforce 
models of doubly exponential size (actually, it can be done in \UNFO{} even without transitive relations).

Thm.~\ref{t:maintr} immediately gives the decidability of the finite satisfiability problem for \UNFOTR{} and
 suggests a simple \ThreeNExpTime{}-procedure: convert a given formula $\phi$ into normal form $\phi'$, guess a finite structure
of size bounded triply exponentially and verify that it is a model of $\phi'$. We can however do better and show 
a doubly exponential upper bound matching 
the known complexity of the general satisfiability problem. The following theorem has already been stated in the main
body of this paper.

\begin{theorem}[restating of Thm.~\ref{t:algo}]\label{t:algorepeat}
The finite satisfiability problem for \UNFOTR{} is \TwoExpTime-complete.
\end{theorem}
\begin{proof}
The lower bound is inherited from pure \UNFO{} \cite{StC13} or from \UNFOtTR{} (Thm.~\ref{t:lalgo}).

For the upper bound, we describe an algorithm in \AExpSpace. Fix $\phi$ in normal form. 
We have proved that $\phi$ has a finite model iff it has a tree-like model with doubly exponentially bounded transitive paths (as in Lemma \ref{l:short}).
We will look for the latter. We advise the reader to recall the proof of Lemma \ref{l:reg}, as we presently use a similar apparatus. 
In our procedure we produce, in an alternating fashion,  a finite tree $\str{A}^*$, corresponding to some number of the upper levels 
of a model. Simultaneously, we define a function $\fg^*$ returning for an element of $A^*$ its 1-type together with some $\phi$-declaration and one stopwatch for each of the ${T}_u$ (cf.~the proof of Lemma \ref{l:reg}).

More precisely, let $\hat{M}_\phi$ be the bound on transitive paths obtained in Lemma \ref{l:short} and $M$
be a bound on $|\rng\fg^*|$ (we use $({T}_u,\hat{M}_{\phi})$-stopwatches in $\fg^*$).
The alternating algorithm works as follows.
Calculate $\hat{M}_\phi$ and $M$. Note that both are doubly exponential in $|\phi|$. Construct the root of $\str{A}$ and guess 
its $1$-type $\alpha$, a $\phi$-declaration $\fd$ containing all the formulas of the form $\phi_0^j(\bar{x})\wedge\bigwedge_{i\in Q}x_i=y\wedge\bigwedge_{i\in\mathcal{Q}\setminus Q}x_i\neq y$ for any $Q\subseteq\mathcal{Q}$ and $1\leq j\leq z$ (recall that $\phi_0$ is equivalent to $\phi_0^1\vee\ldots\vee\phi_0^z$ with the $\phi_0^j$ being conjunctions of some $\mathcal{R}$ and $\mathcal{T}$ formulas).
Set $\fg^*(a)=(\alpha,\fd,(0)_{u=1}^{2k})$.
Now construct the downward family of $a$, $F=\{a,a_1,\ldots,a_s\}$, for some $s<|\phi|$, guess its (transitively closed) 
structure, and guess the values $\fg^*(a_1),\ldots,\fg^*(a_s)$. Check whether $F$ is a $\phi$-witness structure for $a$, the 1-types assigned by $\fg^*$ agree with the structure, the declarations assigned by $\fg^*$ satisfy the LCCs and the stopwatches assigned by $\fg^*$ satisfy the local condition described in the definition of $({T}_u,\hat{M}_\phi)$-stopwatch labeling. If not, reject. Next universally choose one of the $a_i$. Then proceed as for $a$---guess the downward family of $a_i$ and  values of $\fg^*$, and check their consistency as above, universally choose one of the children of $a_i$ and so on. We additionally keep a counter containing the number of the current level in $\str{A}^*$. If it reaches $M+1$, we accept. 

It is clear that the described algorithm can be implemented in \AExpSpace{}: we only need to store the structure and the
values of $\fg^*$ on a single family, plus a counter. All of these can be written using exponentially many bits.

\smallskip\noindent
\emph{Correctness proof}. To see that if $\phi$ has a model $\str{A}$ with bounded transitive paths then the algorithm accepts, it is sufficient to
make the guesses in accordance with $\str{A}^*$---the structure induced on the first $M+1$ levels of $\str{A}$ with $\fg^*$ defined as follows $A^*\ni a\mapsto(\type{\str{A}}{a},\dec{\str{A}}{\phi}{a},(\mathcal{S}_u)_{u=1}^{2k})$ where $\mathcal{S}_u$ is the $({T}_u,\hat{M}_\phi)$-stopwatch labeling of $\str{A}$. The fact that such a strategy leads to an accepting run of the algorithm is almost straightforward.
In particular, the local consistency of declarations follows from Lemma \ref{l:locglob}(ii). 
The opposite implication uses ideas  similar to the ones from the proof of Lemma \ref{l:reg}.
Assume that the algorithm has
an accepting run. From this run we can naturally infer a tree-like structure $\str{A}^*$ consisting of $M{+}1$ levels, and 
a function $\fg^*$. Note that on each path from the root to a leaf in $\str{A}^*$ some value of $\fg^*$ appears at least twice. 
Cut each branch at the first position on which the value of $\fg^*$ reappears and make a link from this point to
the first occurrence  of this value on the considered branch. Naturally unravel so obtained structure into an infinite tree-like structure $\str{A}$.
Define on $\str{A}$ function $\fg$ just copying the values of $\fg^*$. We show that $\str{A}\models\phi$ and has transitive paths bounded by $\hat{M}_\phi$. Note that the downward families in $\str{A}$ and the values of $\fg$ on them are copies of some downward families in $\str{A}^*$ and their values of $\fg^*$, so each $a\in A$ has a $\phi$-witness structure ($\str{A}$ satisfies all the $\forall\exists$-conjuncts of $\phi$) and also $\fg$ gives a locally consistent set of declarations and $({T}_u,\hat{M}_\phi)$-stopwatch labelings. The latter guarantee that $\str{A}$ has bounded transitive paths; the former, together with the choice of the declaration $\fd$ for the root of $\str{A}^*$, allows us to conclude that $\str{A}$ satisfies the $\forall$-conjunct of $\phi$. 
\end{proof}

As remarked in the Introduction, we can state our results in a slightly stronger way, for a setting in which we may not only require some 
binary symbols to be interpreted as arbitrary transitive relations, but we can, more specifically, require some of them to be
equivalences and some other---partial order. Indeed, assuming that $T_u$ is transitive we can enforce it in \UNFO{} to be a (strict) partial order,
writing $\neg \exists xy (T_uxy \wedge T_uyx)$. Non-strict partial orders can be then simulated by disjunctions $T_uxy \vee x = y$.
An equivalence relation can be simulated by some $T_u$ by replacing every usage of $T_uxy$ by $T_uxy \wedge T_u^{-1}xy \vee x=y$ (and then ignoring
the non-symmetric interpretations of $T_u$; 
we remark that it is not possible to enforce in \UNFOTR{} $T_u$ to be interpreted as an equivalence
\cite{DK18}).
\begin{corollary}
The finite satisfiability problem for \UNFO{} with transitive relations, equivalences and partial orders is \TwoExpTime-complete.
\end{corollary}

We note that our approach does not allow us to deal with \emph{linear} orders.
Actually, the presence of a strict linear order $<$ makes the satisfiability problem for \UNFO{} undecidable, as
it allows for a reduction from  \UNFO{} with inequalities, which is known to be undecidable \cite{StC13}:  $x \not=y$ can be then expressed as $x < y \vee y<x$. 
See also \cite{ABBB16}. To the best of our knowledge, the decidability of the (finite) satisfiability problem for \UNFO{} with \emph{non-strict} linear orders is open.

\section{Capturing expressive description logics} \label{s:edl}

\subsection{Constants} \label{s:constbin}

To show a small model property, and establish the decidability of \UNFOTR{} with constants, \UNFOSO{}, we are not going to 
design any new transformations of models. We will just use Thm.~\ref{t:maintr} and Thm.~\ref{t:algo}.
Our plan is to simulate constants with freshly introduced unary predicates. Such predicates will be called \emph{pseudoconstants}.
Of course, our transformations of models from Section \ref{s:general} do not respect the uniqueness of interpretations of pseudoconstants. We thus introduce a simple quotient construction which given
a structure, for every pseudoconstant, shrinks all its interpretations into a single element. There is a potential danger here: the shrinking operation may lead to some new patterns of connections. E.g., if a $a$ sends a ${T}_u$-edge  to an interpretation of a pseudoconstant, $c$, and
$b$ receives a ${T}_u$ from another incarnation $c'$ of the same pseudoconstant, then in the resulting model, $a$ and $b$ 
become ${T}_u$-connected, even thought they need not be ${T}_u$-connected in the original model. To make this operation safe 
we will perform some manipulations on the input formula.

\medskip\noindent
\emph{Pseudoconstants.}
Let $\phi$ be a \UNFOTR{} 
formula 
with constants $c_1,\ldots,c_K$. 
Take a set of fresh unary symbols $\sigma_{\const}:=\{C_1,\ldots,C_K\}$ and let $\phi_{\const}$ be a formula obtained from $\phi$ by 
simulating constants with symbols from $\sigma_{\const}$. Namely, for every relational symbol $R$, and every atom
$R(x_1, \ldots, x_s, c_1, \ldots, c_t)$, where the $x_i$ are its variables and the $c_i$ are its constants,
we replace this atom with $\exists y_1 \ldots y_t R(x_1, \ldots, x_s, y_1, \ldots, y_t)$, where the $y_i$ are fresh variables.
Note that $\phi_{\const}$ remains a \UNFOTR{} formula.

Let $\exi:=\bigwedge_i\exists xC_ix$ (there exists an interpretation  of every pseudoconstant) 
and $\uni:=\bigwedge_i\forall xy (C_ix\wedge C_iy\to x=y)$ (every pseudoconstant is uniquely interpreted). Notice that
$\exi$ is in \UNFO{}, but $\uni$ is not. Clearly, $\phi$ is (finitely) satisfiable iff $\phi_{\const}\wedge\exi\wedge\uni$ is (finitely) satisfiable. Furthermore, assuming that $\bar{\phi}_{\const}$ is a normal form of $\phi_{\const}$, due to Lemma \ref{l:nf}
we can check (finite) satisfiability of $\bar{\phi}_{\const}\wedge\exi\wedge\uni$ instead of $\phi$.

So, w.l.o.g., we will consider the finite satisfiability problem for formulas of the form 
$\phi\wedge\exi\wedge\uni$, where $\phi$ is a \UNFOTR{} formula in normal form, over some signature $\sigma=\sigma_{\cBase}\cup\sigma_{\cDist}\cup\sigma_{\const}$, with $\sigma_{\const}$ consisting of auxiliary unary relation symbols.

\medskip\noindent
\emph{Shrinking.} Let us now introduce our shrinking  operation.
Let $\widetilde{\str{A}}$ be a structure.
Define a relation $\sim$ on $\widetilde{A}$ by setting $a\sim b$ iff $\widetilde{\str{A}}\models\varepsilon(a,b)$, where $\varepsilon(x,y)=(x=y)\vee((\bigvee_iC_ix)\wedge\bigwedge_i C_ix\leftrightarrow C_iy)$ ($x$ and $y$ correspond to the same constant or they are the same non-constant). Observe that $\sim$ is an equivalence relation 
and that $\varepsilon$ is in \UNFO. 
Let $A^0:=\widetilde{A}/\sim$ and $\fq:\widetilde{A} \rightarrow A^0$ be the corresponding quotient map.
If $\widetilde{\str{A}}\models R\bar{a}$ for some relation symbol $R$ and a tuple $\bar{a}$, then put $\str{A}^0\models R\fq(\bar{a})$. 
Let $\str{A}$ be the result of applying the transitive closure to all the ${T}_u$ in $\str{A}^0$.
Observe that $\fq:\widetilde{\str{A}}\to\str{A}$ is a homomorphism, $\str{A} \models \uni$ and if $\widetilde{\str{A}} \models \uni$ then $\fq:\widetilde{\str{A}}\to\str{A}$ is an isomorphism.
A triple of the form $(\widetilde{\str{A}},\str{A},\fq)$ will be called a \emph{shrinking triple}.
We will sometimes refer to the intermediate model $\str{A}^0$ (as usual, in some arguments involving transitive connections in $\str{A}$ corresponding to paths in $\str{A}^0$).

\medskip\noindent
\emph{Modifications of $\phi$.} 
Now we 
introduce a modification of the
given formula $\phi$. 
Let $\forall\neg \bar{x} \phi_0(\bar{x})$ be the universal conjunct of $\phi$.
Note that a ${T}_u$-connection  may  appear in the shrinking $\str{A}$ of a given model
$\widetilde{\str{A}}$ either as a direct copy of some $2$-type from $\widetilde{\str{A}}$ or as the transitive closure of 
a $T_u$-path that goes through some constants. To capture this second possibility we replace every atom ${T}_uzz'$
in $\phi_0$ by
$\tau_u(z,z'):=\exists x_0,y_1,x_1,y_2,\ldots,x_K,y_{K+1} (\varepsilon(z,x_0)\wedge\bigwedge_{i=1}^K\varepsilon(y_i,x_i)\wedge\varepsilon(y_{K+1},z')\wedge T_ux_0y_1\wedge\bigwedge_{i=1}^K(T_ux_iy_{i+1}\vee x_i=y_{i+1}))$.
We also replace any non-transitive atom in $Rz_1 \ldots z_s$ in $\phi_0$ by  $\rho_R(z_1,\ldots,z_s):=\exists x_1\ldots x_s(\bigwedge_i\varepsilon(x_i,z_i)\wedge Rx_1\ldots x_s)$.
Let us denote by $\phi_0^{path}$ the formula resulting from such operations on $\phi_0$,
and by $\phi^{path}$ the result of substituting the $\forall$-conjunct of $\phi$ with $\forall\bar{x}\neg\phi_0^{path}(\bar{x})$. Note that $\phi^{path}$ is (equivalent to) a \UNFOTR{} formula.

For technical reasons we introduce two additional formulas. Let $\con$ be a formula saying that the $\sim$-equivalent elements 
have the same $1$-types. Let $\typ$ be a formula whose aim is to prevent the 1-types from enlarging after the application of the shrinking operation,  $\typ:=\bigwedge_u\forall x(\neg T_uxx\to\neg\tau_u(x,x))$. Clearly both $\con$ and $\typ$ can be treated as
\UNFOTR{} formulas.  

Observe that for any structure $\str{A}$, if  $\str{A}\models\con$ then there are at most $K$ equivalence classes of the relation $\sim$ containing some pseudoconstants. 

Let us collect some basic properties of our transformation.
\begin{claim}\label{c:pathform-1}
Let $\str{A}$ be an arbitrary structure. Then for any $a,b\in A$, $\bar{a}\subseteq A$, non-transitive symbol $R$ and transitive symbol $T_u$
\begin{enumerate}[(i)]

\item If $\str{A}\models T_uab$ then $\str{A}\models\tau_u(a,b)$.

\item If $\str{A}\models R\bar{a}$ then $\str{A}\models\rho_R(\bar{a})$.

\item If $\str{A}\models\typ$ then $\str{A}\models T_u aa$ iff $\str{A}\models\tau_u(a,a)$.

\item If $\str{A}\models\con$ then $\str{A}\models Raa\ldots a$ iff $\str{A}\models\rho_R(a,a,\ldots,a)$.

\item If $\str{A}\models\con\wedge\typ$ and $\str{A}\models\phi_0^{path}(\bar{a})$ then $\str{A}\models\phi_0(\bar{a})$.

\item In particular, if $\str{A}\models\con\wedge\typ$ and $\str{A}\models\phi^{path}$ then $\str{A}\models\phi$.	
\end{enumerate}
\end{claim}
\begin{proof}
\begin{enumerate}[(i)]
\item If $\str{A}\models T_uab$ then substitute
$a$ for $x_0$ and $b$ for all the other quantified variables in $\tau_u(a,b)$ to obtain that $\str{A}\models\tau_u(a,b)$.

\item Substitute the existentially quantified variables in $\rho_R(\bar{a})$ with $\bar{a}$.

\item ($\Rightarrow$) follows from (i). ($\Leftarrow$) is exactly $\typ$. 

\item ($\Rightarrow$) follows from (ii). ($\Leftarrow$) We have then some $a'\in A$ such that $a\sim a'$ and $\str{A}\models Ra'a'\ldots a'$, so, from the fact that $\str{A}\models\con$ it follows that $\str{A}\models Raa\ldots a$.

\item  $\phi_0$ is in \UNFOTR, so its non-unary atoms cannot be negated. So by (i)--(iv) the thesis follows. 

\item The formulas $\phi$ and $\phi^{path}$ differ only on their $\forall$-conjunts. So the thesis follows, since by (v), if $\str{A}\models\forall\bar{x}\neg\phi_0(\bar{x})$ then $\str{A}\models\forall\bar{x}\neg\phi_0^{path}(\bar{x})$

\end{enumerate}
\end{proof}

\begin{claim} \label{c:pathform0}
Let $\str{A} \models \uni$. 
Then for any $a,b\in A$, and $\bar{a} \subseteq A$
\begin{enumerate}[(i)]

\item For any $1\leq u\leq 2k$, if $\str{A} \models\tau_u(a,b)$ then $\str{A}\models T_uab$. In particular $\str{A}\models\typ$.
\item For any non-transitive symbol $R$, if $\str{A}\models\rho_R(\bar{a})$ then $\str{A}\models R\bar{a}$.
\item If $\str{A}\models\phi_0^{path}(\bar{a})$ then $\str{A}\models\phi_0(\bar{a})$.
\item If $\str{A}\models\phi$ then $\str{A}\models\phi^{path}$.
\end{enumerate}
\end{claim} 
\begin{proof}

\begin{enumerate}[(i)]

\item
Since $\str{A}\models\uni$, the relation $\sim$ is the identity and $\varepsilon$ defines the identity. Therefore, since $\str{A}\models\tau_u(a,b)$ we have that $\str{A}\models\exists x_0\ldots x_{K+1} a=x_0\wedge x_{K+1}=b\wedge T_ux_0x_1\wedge\bigwedge_{i=1}^K(T_ux_ix_{i+1}\vee x_i=x_{i+1})$.
Since $\str{A}$ is transitively closed, $\str{A}\models T_uab$. 
\item
Analogous to (i).
\item 
Use (i), (ii) and Claim \ref{c:pathform-1}(iii), (iv) ($\str{A}\models\typ$ by (i) and $\uni$ clearly implies $\con$) to mimic the proof of Claim \ref{c:pathform-1}(v)

\item 
Follows from (iv) in such a way as  Claim \ref{c:pathform-1}(vi) follows from Claim \ref{c:pathform-1}(v). 
\end{enumerate}
\end{proof}

Next, let us now see how our formula transformation interacts with the shrinking operation.

\begin{claim}\label{c:pathform1}
Let $(\widetilde{\str{A}},\str{A},\fq)$ be a shrinking triple
such that $\widetilde{\str{A}}\models\con \wedge \typ$.
Let $a,b\in\widetilde{A}$, let $\bar{a} \subseteq \widetilde{A}$. Then
\begin{enumerate}[(i)]
\item For any $1\leq u\leq 2k$, if $\str{A}\models \tau_u(\fq(a),\fq(b))$ then $\widetilde{\str{A}}\models  \tau_u(a,b)$

\item For any non-transitive symbol $R$, if $\str{A}\models\rho_R(\fq(\bar{a}))$ then $\widetilde{\str{A}}\models\rho_R(\bar{a})$

\item For any $1\leq u\leq 2k$, $\str{A}\models\tau_u(\fq(a),\fq(a))$ iff $\widetilde{\str{A}}\models\tau_u(a,a)$

\item For any non-transitive symbol $R$, $\str{A}\models \rho_R(\fq(a),\ldots,\fq(a))$ iff $\widetilde{\str{A}}\models\rho_R(a,\ldots,a)$

\item If $\str{A}\models\phi^{path}_0(\fq(\bar{a}))$ 
then $\widetilde{\str{A}}\models\phi_0^{path}(\bar{a})$. 

\item In particular, if $\widetilde{\str{A}}\models\forall\bar{x}\neg\phi_0^{path}(\bar{x})$, then $\str{A}\models\forall \bar{x}\neg\phi_0^{path}(\bar{x})$
 \end{enumerate}
\end{claim}
\begin{proof}
\begin{enumerate}[(i)]

\item Assume $\str{A}\models \tau_u(\fq(a),\fq(b))$. Recall the  intermediate model $\str{A}^0$ used in the definition of shrinking.
Then there exist elements $a_0',a_1',\ldots,a_N'\in A^0$ 
such that $\str{A}^0\models \fq(a)=a_0'\wedge\bigwedge_i T_u\fq(a_i') \fq(a_{i+1}'))\wedge a_N'=\fq(b)$.
This path can be "lifted" to $\widetilde{\str{A}}$, that is there exist elements $b_0,a_0,b_1,a_1\ldots,a_{N-1},b_N,a_N\in\widetilde{A}$ such $b_0=a$, $a_N=b$, and for all $i$ we have $\fq(a_i)=\fq(b_i)=a_i'$, and $\widetilde{\str{A}}\models T_ua_ib_{i+1}$.
Assume that $N$ is the smallest possible.
Observe, that if we prove that $N\leq K+1$ then $\widetilde{\str{M}}\models\tau_u(a,b)$ (just substitute the existentially quantified variables in $\tau_u$ with $a_0,b_1,a_1,\ldots,a_{N-1},b_N,b,\ldots,b$ respectively). We claim that indeed $N \le K+1$.
Note first that for each $1\leq i\leq N-1$ element  $a_i$ is an interpretation of a  pseudocontant and that  $b_i$ is (a different incarnation of)
the same pseudoconstant. Otherwise $a_i=b_i$, so we can remove them both from the sequence to obtain a shorter sequence satisfying the required properties (recall that $\widetilde{\str{A}}$ is transitively closed).
Furthermore, all the  $a_i$, for $1\leq i\leq N-1$, are different pseudoconstants -- if they were not, that is,  for some $i<i'$ we had $a_i\sim a_i'$, we could cut out $b_i,a_i,\ldots,b_{i'-1},a_{i'-1}$ and obtain a shorter sequence, with the required properties. 
As we have at most $K$ constants
 it follows that $N-1\leq K$.  

\item Analogous to (i), yet simpler. It just suffices to "lift" $\fq(\bar{a})$ to some $\bar{a}^*\subseteq\widetilde{A}$ such that $\widetilde{\str{A}}\models R(\bar{a}^*)$.

\item ($\Rightarrow$) Follows from (i). ($\Leftarrow$) By $\widetilde{\str{A}}\models\typ$ we have that $\widetilde{\str{A}}\models T_uaa$, so by the definition of shrinking $\str{A}\models T_u\fq(a)\fq(a)$. Since $\str{A}\models\uni$, it satisfies $\con$ and by Claim \ref{c:pathform0}(i) it satisfies $\typ$, by Claim \ref{c:pathform-1}(i) $\str{A}\models\tau_u(\fq(a),\fq(a))$.

\item ($\Rightarrow$) Follows from (ii). ($\Leftarrow$) If there exists $a'\sim a$ such that $\widetilde{\str{A}}\models Ra'\ldots a'$ then by the definition of shrinking $\str{A}\models R\fq(a')\ldots\fq(a')$. Since $\fq(a)=\fq(a')$ we have that $\str{A}\models\varepsilon(\fq(a),\fq(a'))$, so $\str{A}\models\rho_R(\fq(a),\fq(a'))$.

\item Follows from (i)--(iv) as Claim \ref{c:pathform-1}(v) follows from Claim \ref{c:pathform-1}(i)--(iv).

\item Recall that $\fq$ is onto $A$ and use (v). 	
\end{enumerate}
\end{proof}

	\medskip\noindent
	\emph{Putting all the pieces together.}
After such preparations, we can state now our crucial lemma.
\begin{lemma}\label{l:constmain}
Let $(\widetilde{\str{A}},\str{A},\fq)$ be a shrinking triple. 
Assume that $\widetilde{\str{A}}\models\phi^{path}\wedge\con\wedge\exi\wedge\typ$. 
Then $\str{A}\models\phi^{path}\wedge\con\wedge\exi\wedge\uni\wedge\typ$.  
\end{lemma}
\begin{proof}
\emph{$\fq$ preserves the 1-types.} 
Assume to the contrary that for some $a\in\widetilde{M}$ we have $\type{\widetilde{\str{M}}}{a}\neq\type{\str{M}}{\fq(a)}$. Then, by the definition of shrinking and the fact that $\widetilde{\str{A}}\models\con$, we have that in fact the $\subsetneq$ inclusion holds and that the inequality follows from some transitive atom $T_uxx$ belonging to the latter and not to the former.
But this means that, by Claim \ref{c:pathform1}(i) we have that $\widetilde{\str{A}}\models\tau_u(a,a)$.
Since $\widetilde{\str{A}}\models\typ$ we have that $\widetilde{\str{A}}\models T_uaa$. Contradiction.

\smallskip\noindent
\emph{$\str{A}$ satisfies the $\forall\exists$ conjuncts
 of $\phi^{path}$.} Follows from the following facts: 
$\fq$ preserves the 1-types and the set of positive atoms of arity greater 
 than $1$ can only be enlarged when pushed using $\fq$. 
So it suffices to just push the witnesses through $\fq$.  	

\smallskip\noindent
\emph{$\str{A}\models\con\wedge\exi\wedge\uni$.} Straightforward, by the definition of shrinking.

\smallskip\noindent
\emph{$\str{A}$ satisfies the $\forall$-conjuncts of  $\phi$ and $\phi^{path}$.} Use Claim \ref{c:pathform1}(vi).

\smallskip\noindent
\emph{$\str{A}\models\typ$.} Use Claim \ref{c:pathform0}(i).

\end{proof}

Let us now prove the main result of this subsection.
\begin{theorem}
	The finite satisfiability problem for \UNFOSO{} is \TwoExpTime-complete. If a \UNFOSO{} formula has has a finite model then it has a model of size triply exponential in its length.	
\end{theorem}
\begin{proof}
Let us first show the second part of this theorem. Take a finitely satisfiable \UNFOSO{} formula $\phi^*$ and let $\phi$ be its \UNFOTR{} version with constants simulated by pseudoconstants. Let $\str{A}$ be a finite model of $\phi \wedge \exi \wedge \uni$. Thus,
we also have $\str{A} \models \con \wedge \typ$, and by part (iv) of Claim \ref{c:pathform0}
it holds $\str{A} \models \phi^{path}$. 
Note that the last formula has size at most quadratic in  $|\phi|$.
We now take a triply exponentially
bounded model $\widetilde{\str{B}}$ of $\phi^{path} \wedge \exi  \wedge \con \wedge \typ$, guaranteed by Thm.~\ref{t:maintr}.
Let $\str{B}$ be its
shrinking. By Lemma \ref{l:constmain} we have $\str{B} \models \phi^{path} \wedge \exi \wedge \uni \wedge \con \wedge \typ$.
By part (i) of Claim \ref{c:pathform0} we have $\str{B} \models \phi \wedge \exi \wedge \uni$ and thus $\str{B} \models \phi^*$.

Regarding the complexity,
as explained above, $\phi^*$ has a finite model iff $\phi^{path}\wedge\exi\wedge\con\wedge\typ$ has a finite model. As the latter is a \UNFOTR{} formula its satisfiability can be checked in \TwoExpTime{} by Thm \ref{t:algo}.
\end{proof}

\subsection{Constants and binary inclusions}
Here we show how to extend the proofs and techniques from the previous subsection to cover simultaneously constants and  inclusions
of binary relation.
In this subsection we assume that, similarly to  transitive symbols from $\sigma_\cDist$, also binary symbols from $\sigma_\cBase$ come in pairs
$B$, $B^{-1}$ and that the symbols from such pairs are always interpreted as inverses of each other. This can be done w.l.o.g., since all our constructions
from section \ref{s:general} respect this property.

The \emph{(finite) satisfiability problem for the unary negation fragment with constants and inclusions of binary relations}, \UNFOSOH{},
is defined as follows. Given a \UNFOSOH{} formula $\phi$ 
and a set $\mathcal{H}$ of inclusions of the form  $B \subseteq B'$, for binary symbols $B, B'$, check if there exists
a (finite) model of $\phi$ in which for every $B \subseteq B' \in \mathcal{H}$ the interpretation of $B$ is contained
in the interpretation of $B'$.

As in the previous subsection we simulate constants by pseudoconstants
and  search for finite models of a formula $\phi\wedge\uni\wedge\exi\wedge\mathcal{H}$,
 with normal form $\phi$, over a purely relational signature $\sigma=\sigma_\cBase \cup \sigma_\cDist \cup \sigma_\const$.

We first remark that all our constructions from Sections \ref{s:general} and \ref{s:constbin}, without \emph{literally any changes}, respect
inclusions of the form $T \subseteq T'$, $B \subseteq B'$ and $B \subseteq T$ for any $T, T' \in \sigma_{\cDist}$
and $B, B' \in \sigma_{\cBase}$. The only problematic inclusions are 
those of the form $T \subseteq B$ for $B \in \sigma_{\cBase}$ and $T \in \sigma_{\cDist}$.
To deal with them we will introduce an operation of taking the \emph{pseudotransitive closure}. 

For a given set of inclusions $\mathcal{H}$ let $\mathcal{H}^+$ denote the smallest set such that (i) $\mathcal{H} \subseteq \mathcal{H}^+$, (ii) if $B_1 \subseteq B_2 \in \mathcal{H}^+$ then $B_1^{-1} \subseteq B_2^{-1} \in \mathcal{H}^+$, (iii)   if $B_1 \subseteq B_2 \in \mathcal{H}^+$ and $B_2 \subseteq B_3 \in \mathcal{H}^+$  then $B_1 \subseteq B_3 \in \mathcal{H}^+$. For any structure
$\str{A}$ we have that $\str{A} \models \mathcal{H}$ iff $\str{A} \models \mathcal{H}^+$. 
So, w.l.o.g., from now we assume  that $\mathcal{H}$ itself satisfies (i)--(iii).
Denote by $\mathcal{H}_0$ the subset of "safe" inclusions of $\mathcal{H}$,  $\mathcal{H}_0:=\{B_1\subseteq B_2\in\mathcal{H}: B_1 \not\in
\sigma_\cDist \text{ or } B_2 \not\in \sigma_{\cBase}\}$.
If $T \subseteq B \in \mathcal{H}$, for $T \in \sigma_{\cDist}$ and $B \in \sigma_{\cBase}$ then $B$ is called \emph{pseudo-transitive}. 
Pseudo-transitive relations must be treated in a special way.
Let $\str{A}$ be a structure. Then we define its pseudotransitive closure as the structure $\str{A}'$ as follows: $A=A'$, for all $a,b\in A$ and $B\in\sigma_{\cBase}$, $\str{A}\models Bab$ iff $\str{A}'\models Bab\vee\bigvee_{T\subseteq B\in\mathcal{H},T\in\sigma_{\cDist}}Tab$, for all the other relations we copy their interpretations from $\str{A}$ to $\str{A}'$. Observe that if $\str{A}\models\mathcal{H}_0$ then $\str{A}'\models\mathcal{H}$, if $\str{A}\models\mathcal{H}$ then $\str{A}=\str{A}'$, and that the identity map $\iota:A\to A'$ is a homomorphism.

To deal simultaneously with constants and inclusions, we plan to apply both the shrinking operation and the pseudotransitive closure. More precisely, starting from a model respecting $\mathcal{H}_0$ we first apply to it the shrinking operation and then, to the result, the pseudotransitive closure, obtaining a model satisfying both $\mathcal{H}$ and $\uni$. 
As in the previous subsection, to ensure that the above transformations respect the given formula $\phi$ we will need
to perform some syntactic manipulations, and formulate a counterpart of Lemma \ref{l:constmain} taking into acount both 
types of operations. To make this work for the pseudotransitive closure, slightly different manipulations and slightly stronger assumptions on the initial model will be needed.

In particular, for technical reasons, we need a simple condition on the realized  1-types, namely, that they respect the inclusions from $\mathcal{H}$.
Formally, this is captured by a \UNFO{} formula $\res:=\forall x\bigwedge_{B_1\subseteq B_2\in\mathcal{H}}(B_1xx\to B_2xx)$. 
Observe that for any model $\str{A}$ such that $\str{A}\models\mathcal{H}$, $\str{A} \models \res$.  

Given $\phi$ in normal form with universal conjunct $\forall \bar{x} \neg \phi_0(\bar{x})$
we now describe its modifications
corresponding to \emph{path} modifications from  Section \ref{s:constbin}. The modifications are performed on formula $\phi_0$: for the atoms $T_uxy$ and $R\bar{x}$ with $R$ of arity other than $2$, replace them, as before, with the formulas $\tau_u(x,y)$ and $\rho_R(\bar{x})$, respectively.
In the case of atoms $Rxy$  (with  $R \in \sigma_\cBase$), replace them with $\omega_R(x,y):=\rho_R(x,y)\vee\bigvee_{T_u\subseteq R\in\mathcal{H}}\tau_u(x,y)$. Let $\phi_0^{path{-}inc}$ be the result of this transformation, and let 
$\phi^{path{-}inc}$ be the result of substituting the conjunct $\forall x \neg\phi_0(\bar{x})$ of $\phi$ with $\forall\bar{x}\neg\phi_0^{path{-}inc}$. Note that $\phi^{path-inc}$ is (equivalent to) an \UNFOTR{} formula.

From now we can proceed analogously as we did for the $path$ modification.   
 Let us now collect some basic properties of this syntactic transformation and see how it interplays with the operations of shrinking
and taking the pseudotransitive closures.
\begin{claim}\label{c:pathincform-1}
	Let $\str{A}$ be an arbitrary structure, $a,b\in A$ and $R\in \sigma_\cBase$
		be a relational symbol of arity $2$. Then
	\begin{enumerate}[(i)]
	\item If $\str{A}\models Rab$ then $\str{A}\models\omega_R(a,b)$
	
	\item If $\str{A}\models\con\wedge\typ\wedge\res$ then $\str{A}\models Raa$ iff $\str{A}\models\omega_R(a,a)$
	
	\item If $\str{A}\models\con\wedge\typ\wedge\res$ then if $\str{A}\models\phi^{path-inc}$ then $\str{A}\models\phi$
	\end{enumerate}
\end{claim}
\begin{proof}
\begin{enumerate}[(i)]
	\item Use Claim \ref{c:pathform-1}(ii) and the definition of $\omega_R$.
	
	\item ($\Rightarrow$) Follows from (i). ($\Leftarrow$) By the definition of $\omega_R$ and Claim \ref{c:pathform-1}(iii), (iv) we have that $\str{A}\models Raa$ or $\str{A}\models Taa$ for some $T\subseteq R\in\mathcal{H}$. Using the fact that $\str{A}\models\res$ the thesis follows.
	
	\item Use Claim \ref{c:pathform-1}(i)--(iv) and (i), (ii) as usual.
\end{enumerate}
\end{proof}

\begin{claim}\label{c:pathincform0}
Let $\str{A}\models\uni \wedge \mathcal{H}$. 
Then for any $a,b\in A$
\begin{enumerate}[(i)]
\item If $\str{A}\models\omega_R(a,b)$ then $\str{A}\models Rab$

\item If $\str{A}\models\phi$ then $\str{A}\models\phi^{path-inc}$.
\end{enumerate}
\end{claim}

\begin{proof}
\begin{enumerate}[(i)]
\item By the definition of $\omega_R$ and Claim \ref{c:pathform-1}(i), (ii) we have that $\str{A}\models Rab$ or $\str{A}\models Tab$ for some $T\subseteq R\in\mathcal{H}$. Since $\str{A}\models\mathcal{H}$, the thesis follows.

\item Use Claims \ref{c:pathform0}(i), (ii), \ref{c:pathform-1}(iii), (iv) ($\str{A}\models\uni$ so it satisfies the required assumptions), \ref{c:pathincform-1}(ii) ($\str{A}\models\mathcal{H}$ implies $\str{A}\models\res$) and (i) as usual.	
\end{enumerate}
\end{proof}

\begin{claim}\label{c:pathincform1}
Let $(\widetilde{\str{A}},\str{A},\fq)$ be a shrinking triple such that $\widetilde{\str{A}}\models\con\wedge\typ\wedge\res$. Let $a,b\in\widetilde{A}$ and $R$ be a binary non-transitive symbol. Then
\begin{enumerate}[(i)]
\item If $\str{A}\models\omega_R(\fq(a),\fq(b))$ then $\widetilde{\str{A}}\models\omega_R(a,b)$ 

\item $\widetilde{\str{A}}\models\omega_R(a,a)$ iff $\str{A}\models\omega_R(\fq(a),\fq(a))$.

\item If $\widetilde{\str{M}}\models\forall\bar{x}\neg\phi_0^{path-inc}(\bar{x})$, then $\str{M}\models\forall \bar{x}\neg\phi_0^{path-inc}(\bar{x})$
\end{enumerate} 
\end{claim}

\begin{proof}
\begin{enumerate}[(i)]
\item Use the definition of $\omega_R$ and Claim \ref{c:pathform1}(i), (ii).

\item Use the definition of $\omega_R$ and Claim \ref{c:pathform1}(iii), (iv).

\item Follows from (i), (ii) and Claim \ref{c:pathform1}(i)--(iv) as usual.
\end{enumerate}
\end{proof}

\begin{claim}\label{c:pathincform2}
	Let $\str{A}'$ be the pseudotransitive closure of $\str{A}$, $\str{A}\models\mathcal{H}_0$ 
	and
	$a,b\in A(=A')\supseteq$, and $\bar{a} \subseteq A$. Then
	\begin{enumerate}[(i)]
	\item For any $1\leq u\leq 2k$ we have $\str{A}'\models \tau_u(a,b)$ iff $\str{A}\models  \tau_u(a,b)$ and $\str{A}'\models\rho_R(\bar{a})$ iff $\str{A}\models\rho_R(\bar{a})$.
	\item For any non-transitive symbol $R$ of arity $2$ we have $\str{A}'\models\omega_R(a,b)$ iff $\str{A}\models\omega_R(a,b)$
	\item For any $1\leq u\leq 2k$ we have $\str{A}'\models\tau_u(a,a)$ iff $\str{A}\models\tau_u(a,a)$. For any non-transitive symbol $R$ of arity other than $2$ it holds that $\str{A}'\models\rho_R(a,\ldots,a)$ iff $\str{A}\models\rho_R(a,\ldots,a)$
	\item For any binary non-transitive symbol $R$ we have $\str{A}'\models\omega_R(a,a)$ iff $\str{A}\models\omega_R(a,a)$.
	\item 
	If $\str{A}\models\phi^{path-inc}$ then $\str{A}'\models\phi^{path-inc}$
	\end{enumerate}

\end{claim} 

\begin{proof}
\begin{enumerate}[(i)]
\item Follows from the fact that the process of the pseudotransitive closure does not change 
neither  the transitive relations nor the non-transitive relations of arity other than $2$ nor the equalities (and these are the types of atoms that $\rho_R$ and $\tau_u$ consist of).
	
\item ($\Rightarrow$) If $\str{A}'\models\neg\rho_R(a,b)$ then use (i). Otherwise there exist $a',b'$ such that $\str{A}'\models\varepsilon(a,a')\wedge\varepsilon(b,b')\wedge Ra'b'$, so $\str{A}\models Ra'b'$ or $\str{A}\models T_ua'b'$ for some transitive $T_u$ such that $T\subseteq R\in\mathcal{H}$. Thus, since the interpretation of $\varepsilon$ is not changed by the application of the pseudotransitive closure, by Claim \ref{c:pathform-1}(i), (ii) we have that  $\str{A} \models\rho_R(a,b)$ or $\str{A} \models\tau_u(a,b)$ and therefore $\str{A}\models\omega_R(a,b)$.
($\Leftarrow$) If $\str{A}\models\tau_u(a,b)$ for some $T_u\subseteq R\in\mathcal{H}$ then by (i) we have $\str{A}'\models\tau_u(a,b)$. If $\str{A}\models\rho_R(a,b)$ then, since the interpretation of $R$ in $\str{A}$ is contained in the  interpretation of $R$ in $\str{A}'$ and $\varepsilon$ is interpreted identically in $\str{A}$ and $\str{A}'$ we have that $\str{A}'\models\rho_R(a,b)$.	
	
\item Follows from (i).

\item Follows from (ii).

\item Follows from (i)--(iv) as usual.

\end{enumerate}
\end{proof}

Now we are ready to state the counteparts for Lemma \ref{l:constmain}. 

\begin{lemma}\label{l:constinccl}
	Let $\str{A}'$ be the pseudotransitive closure of $\str{A}$. Assume that $\str{A}\models\phi^{path-inc}\wedge\exi\wedge\con\wedge\typ\wedge \res \wedge \mathcal{H}_0$
	Then  $\str{A}' \models\phi^{path-inc}\wedge\exi\wedge\con\wedge\typ \wedge \res \wedge \mathcal{H}$
	Furthermore, if $\str{A}\models\uni$ then $\str{A}'\models\uni$.
\end{lemma}
\begin{proof}
Let $\iota:A\to A'$ be the inclusion map.

\medskip\noindent
\emph{$\iota$ preserves the 1-types.} The only possibility for non-preservation of the 1-types is when there exists some $a\in A$ and non-transitive $R$ such that $Rxx\in\type{\str{A}'}{a}\setminus\type{\str{A}}{a}$, appearing due to the fact that $\str{A}\models Taa$ for some transitive $T$, and $T \subseteq R \in \mathcal{H}$.
 But then $Txx\in\type{\str{A}}{a}$ and since $\str{A}\models\res$ we get that $Rxx\in\type{\str{A}}{a}$. Contradiction.

\medskip\noindent
\emph{$\str{A}'$ satisfies the $\forall\exists$ conjunct of $\phi^{path-inc}$, $\con, \exi$.} Exactly as in the proof of Lemma \ref{l:constmain} with the role of $\fq$ played now by $\iota$.

\medskip\noindent
\emph{$\str{A}' \models \res$.} $\res$ uses only unary atoms and $\iota$ respects 1-types.

	\medskip\noindent
\emph{$\str{A}'$ satisfies the $\forall$ conjunct of $\phi^{path-inc}$.} Use Claim \ref{c:pathincform2}(v). 	

\medskip\noindent
\emph{$\str{A}'\models\typ$.} The formula $\typ$ uses only some transitive symbols, some symbols of arity $1$ and equalities, 
whose interpretation is not changed after the application of the pseudotransitive closure.
\end{proof}  

\begin{lemma}\label{l:constincq}
	Let $(\widetilde{\str{A}},\str{A},\fq)$ be a shrinking triple. Assume that $\widetilde{\str{A}}\models\phi^{path-inc}\wedge\exi\wedge\con\wedge\typ\wedge\res\wedge \mathcal{H}_0$. 
	Then  $\str{A}\models\phi^{path-inc}\wedge\exi\wedge\uni\wedge\con\wedge\typ\wedge \res\wedge\mathcal{H}_0$.
	\end{lemma}
\begin{proof}

\medskip\noindent
\emph{$\str{A}\models\mathcal{H}_0$.} By the definition of the shinking operation, $\str{A}^0\models\mathcal{H}_0$. Clearly, taking the transitive closure to obtain $\str{A}$ does not spoil this condition.
	
\medskip\noindent
\emph{$\fq$ preserves the 1-types.} Exactly as in the proof of Lemma \ref{l:constmain}. 

\medskip\noindent
\emph{$\str{A}$ satisfies the $\forall\exists$ conjuncts of $\phi^{path-inc}$, and the formulas $\con,\exi,\typ,\uni$.} Exactly as in the proof of Lemma \ref{l:constmain}.
	
\medskip\noindent
\emph{$\str{A} \models \res$.} $\res$ uses only unary atoms and $\fq$ respects 1-types.
	
\medskip\noindent
\emph{$\str{A}$ satisfies the $\forall$ conjunct of $\phi^{path-inc}$.} Use Claim \ref{c:pathincform1}(iii). 	
\end{proof}  

Now we are ready to put the pieces together and prove the main result for the \UNFOSOH.
The first part of the following theorem has already been stated as Thm.~\ref{t:full}.

\begin{theorem}
	The finite satisfiability problem for \UNFOSOH{} is \TwoExpTime-complete. If a \UNFOSOH{} formula has has a finite model then it has a model of size triply exponential in its length.	
\end{theorem}
\begin{proof}
Let us first show the second part of this theorem. Take a finitely satisfiable \UNFOSOH{} formula $\phi^*$ and let $\phi$ be its \UNFOTR{} version with constants simulated by pseudoconstants. Let $\mathcal{H}$ be its inclusions. Let $\str{A}$ be a finite model of $\phi \wedge \exi \wedge \uni \wedge \mathcal{H}$. Therefore
we also have $\str{A} \models \con \wedge \typ \wedge \res \wedge\mathcal{H}_0$, and by part (ii) of Claim \ref{c:pathincform0}
it holds $\str{A} \models \phi^{path-inc}$. 
Note that the last formula has size at most quadratic in  $|\phi|$.
We now take a triply exponentially
bounded model $\widetilde{\str{B}}$ of $\phi^{path-inc} \wedge \exi  \wedge \con \wedge \typ \wedge \res \wedge\mathcal{H}_0$, guaranteed by Thm.~\ref{t:maintr} (recall our previous observation, that the inclusions from $\mathcal{H}_0$ are satisfied by
all the transformations of models from Section \ref{s:general}).
Let $\str{B}$ be its
shrinking. By Lemma \ref{l:constincq} we have $\str{B} \models \phi^{path-inc} \wedge \exi \wedge \uni \wedge \con \wedge \typ \wedge \res\wedge\mathcal{H}_0$.
Let $\str{B}'$ be the pseudotransitive closure of $\str{B}$. By Lemma \ref{l:constinccl} we have 
that  $\str{B}' \models\phi^{path-inc}\wedge\exi\wedge \uni\wedge\con\wedge\typ\wedge \res\wedge\mathcal{H}$. 
By part (iii) of Claim \ref{c:pathincform-1} we have $\str{B} \models \phi \wedge \exi \wedge \uni\wedge \mathcal{H}$ and thus $\str{B} \models \phi^*$.

Regarding the complexity,
as explained above $\phi^*$ has a finite model iff $\phi^{path-inc}\wedge\exi\wedge\con\wedge\typ\wedge\res\wedge\mathcal{H}_0$ has a finite model.
 The latter is a conjunction of  a \UNFOTR{} formula and "safe" inclusions 
 and its  satisfiability can be checked in \TwoExpTime{} by an adaptation of the algorithm
from the proof of Thm \ref{t:algo}, 
forcing it to search for models satisfying $\mathcal{H}_0$ just by ensuring that the structures built on the downward families respect $\mathcal{H}_0$.
\end{proof}

As mentioned in the Introduction, \UNFOSOH{} captures several interesting description logics. This implies that we
can solve FOMQA problem for them. In particular, we have the following corollary, which, up to our knowledge is the first decidability result for FOMQA
in the case of a description logic with both transitive roles and role hierarchies.   
\begin{corollary}[restating of Cor.~\ref{c:fomqa}]
Finite ontology mediated query answering, FOMQA, for the description logic $\mathcal{SHOI}^{\sqcap}$ is decidable and \TwoExpTime-complete.
\end{corollary}

$\mathcal{SHOI}^{\sqcap}$ and some related logics are considered, e.g., in \cite{GK08}. For more about FOMQA for description logics with transitivity see \cite{GYM18}. For more about OMQA for description logics see, e.g., references in \cite{GYM18}.

\section{Towards guarded negation fragment with transitivity} \label{s:tgn}

\subsection{1-dimensional guarded negation fragment with transitivity} 
\emph{Guarded negation fragment}, \GNFO{} \cite{BtCS15}, is a common decidable extension of \UNFO{} and \GF{}
in which negated subformula $\psi$ must be used in conjunction with an atom, called a \emph{guard}, containing
all the free variables of $\psi$. Formally, it is defined by the following grammar:
$$\phi=B\bar{x} \mid x=y \mid \phi \wedge \phi \mid \phi \vee \phi \mid \exists x \phi \mid \gamma(\bar{x}, \bar{y}) \wedge \neg \phi(\bar{y}),$$
where $\gamma$ is an atomic formula, called a \emph{guard}. Equality statements of the form $x=x$ can be used as guards, so \UNFO{} can be seen
as a fragment of \GNFO{}.

The (finite) satisfiability problem for \GNFO{} with transitive relations is undecidable,
since already the two-variable guarded fragment with transitive relations, \GFtTR, is undecidable, \cite{Kie05,Kaz06}.
Recall that the decidability of the general satisfiability problem is regained when transitive symbols are admissible only on non-guard positions
\cite{ABBB16}. We call this decidable variant the \emph{base-guarded negation fragment with transitivity}, \GNFOTR{},
and recall that it embeds \UNFOTR{}.
We do not solve its finite satisfiability problem here, but, analogously to the extension with equivalence relations,
\UNFOEQ{} \cite{DK18}, we are able to lift our results to its one-dimensional restriction, \GNFOTRonedim{}. 
We say that a first-order formula is \emph{one-dimensional} if its every maximal block of quantifiers leaves at most one variable free.
E.g., $\neg \exists 
yz Bxyz$ is one-dimensional, and $\neg \exists z Bxyz$ is not.

\begin{theorem}[rastating of Thm~\ref{t:gnfo}]
The finite satisfiability problem for \GNFOTRonedim{} is \TwoExpTime-complete.
\end{theorem}

The required modifications in our constructions are analogous to those described in \cite{DK18}, in the case of equivalences.
We put them here for the reader's convenience.

	Using a natural adaptation of the standard Scott translation \cite{Sco62} we can transform any sentence belonging to \GNFOTRonedim{} into a normal form sentence $\phi$ of the shape
	as in (\ref{eq:nf}), where the $\phi_i$ are quantifier-free \GNFO{} formulas. Assume that some finite structure $\str{A}$ is a model of $\phi$.
	First, we need a slightly stronger version of condition (a1) in Lemma \ref{l:homomorphisms}---each of the considered homomorphisms should additionally be an isomorphism when restricted to a guarded substructure.
	We need to extend the notion of a declaration so that it treats subformulas of the form $\gamma(\bar{x},\bar{y})\wedge\neg\phi'(\bar{y})$ like non-transitive atomic  formulas. 
	This allows us to perform surgery making the transitive paths bounded and then to  construct a regular tree-like model 
	$\str{A}' \models \phi$ as it is done in the proofs of Lemma \ref{l:short} and Lemma \ref{l:reg}, respectively. 
	The key facts are that Lemma \ref{l:locglob} holds (with the new declarations) and that $\phi_0$  is equivalent to a disjunction of some formulas generated by declarations. 
	Finally we apply, without any changes, the construction from the proof of Lemma \ref{l:finitetr} to $\str{A}'$ and $\phi$ obtaining eventually a finite structure $\str{A}''$. Note that during the step of providing witnesses we build isomorphic copies of partial witness structures, which means that we preserve not only positive atoms but also their negations. Thus the elements of $A''$ have all witness structures
	required by $\phi$. Consider now the conjunct $\forall x_1, \ldots, x_t \neg \phi_0(\bar{x})$, and take arbitrary elements $a_1, \ldots, a_t \in A''$. 
	From Lemma \ref{l:finitetr} we know that there is a homomorphism $\fh:\str{A}'' \restr \{a_1, \ldots, a_t\} \rightarrow \str{A}'$ preserving $1$-types. 
	If  $\gamma(\bar{z}, \bar{y}) \wedge \neg \phi'(\bar{y})$ is a subformula of $\phi_0$ with $\gamma$ a $\sigma_{\cBase}$-guard and $\str{A}'' \models \gamma (\bar{b}, \bar{c}) \wedge \neg \phi'(\bar{c})$ for some
	$\bar{b}, \bar{c} \subseteq \bar{a}$ then, by our construction, all elements of $\bar{b} \cup \bar{c}$ are members 
	of the $\phi$-witness structure for some element.
	As mentioned above such witness structures are isomorphic copies of substructures from $\str{A}$ and $\fh$ 
	works
	on them as an isomorphism, and thus $\fh$  preserves on $\bar{c}$ not only 
	$1$-types and positive atoms but also negations of atoms in witnesses structures. Since  $\str{A}' \models \neg \phi_0(\fh(a_1), \ldots, \fh(a_t))$ this means that $\str{A}'' \models \neg \phi_0(a_1, \ldots, a_t)$.

 The algorithm for checking finite satisfiability presented in the proof of 
Thm.~\ref{t:algo} also works without any changes and, moreover, its  correctness proof does not need any modifications. 
It is the case since a key role is played here by Lemma \ref{l:locglob} that still holds with the new version of declarations.

\subsection{Undecidability with inclusions of binary relations}
Note that  \GNFOTRonedim{} can express inclusions $B \subseteq T$ and $B \subseteq B'$, and our constructions respect the inclusions of the form $T\subseteq T'$ for $B, B' \in \sigma_{\cBase}$ and
$T,T' \in \sigma_{\cDist}$. However, it turns out that if we extend it with inclusions of 
the form $T \subseteq B$ then the (finite) satisfiability problem becomes undecidable. This can be easily shown by a reduction
from the already-mentioned  (finite) satisfiability problem for \GFtTR{}. Indeed, all the negations in a \GFtTR{} formula $\phi$ can be guarded by the guards of quantifiers. If $\phi$ uses a transitive guard $T$ then we can add an inclusion $T \subseteq B_T$, for a fresh $B_T \in \sigma_{\cBase}$ and then use $B_T$ to guard negations. More precisely, subformulas of $\phi$ of the form $\exists y (Txy \wedge \psi(x,y))$
are replaced by $\exists y (Txy \wedge \psi^*(x,y))$, where $\psi^*$ is obtained by replacing negated subformulas $\neg \varsigma(x,y)$
(not in the scope of a deeper quantifier) 
of $\psi$ by, properly base-guarded,
$B_Txy \wedge \neg \varsigma(x,y)$. 
(In subformulas of the original formula of the form $\exists y (Bxy \wedge \psi(x,y))$, for non-transitive $B$, we just add the guard $Bxy$ to all binary
negations in $\psi$, not in the scope of a deeper quantifier.) 
Note that since \GFtTR{} uses only two variables all its formulas are one-dimensional, which is not
changed by the described reduction. We conclude:

\begin{theorem}
The (finite) satisfiability problem for \GNFOTRonedim{} with inclusions of binary relations is undecidable.
\end{theorem}

\end{appendix}

\end{document}